\documentclass[12pt]{article}
\usepackage[top=1.25in, bottom=1.25in, right=1.25in,left=1.25in]{geometry}

\usepackage{amsthm,amsmath,amsfonts,mathrsfs,amssymb}
\usepackage{natbib}
\usepackage{upgreek}
\usepackage{bm}
\usepackage{color}
\usepackage{bbm}
\usepackage{tikz,epstopdf}
\usetikzlibrary{arrows,decorations.pathreplacing, patterns}
\usepackage[titletoc, title]{appendix}
\usepackage{enumerate}
  \usepackage{graphicx}
   \usepackage{array}
   \usepackage{cases}
   \usepackage{setspace}
\usepackage{verbatim}
\usepackage{ragged2e} 
\usepackage{subcaption}
\newcommand{\interior}[1]{
  {\kern0pt#1}^{\mathrm{o}}
}

\newtheorem{thm}{Theorem}
\newtheorem{lem}{Lemma}
\newtheorem{prop}{Proposition}
\newtheorem{col}{Corollary}
\theoremstyle{definition}

\newtheorem{asmp}{Assumption}

\usepackage{hyperref}
\hypersetup{colorlinks=true,
            citecolor=blue,
            linkcolor=blue,
            urlcolor=blue,
			bookmarksopen=true,
			pdfstartview={XYZ null null 1.00},
			pdfpagelayout={SinglePage}
}

\begin{document}

\onehalfspacing

\title{
Estimating Nonseparable Selection Models: \\
A Functional Contraction Approach\thanks{Wu: Peking University HSBC Business School, University Town of Shenzhen, China. Email: \href{mailto:fanwu@phbs.pku.edu.cn}{fanwu@phbs.pku.edu.cn}. Xin: Division of the Humanities and Social Sciences, California Institute of Technology, 1200 East California Blvd., MC 228-77, Pasadena, CA 91125. Email: \href{mailto:yixin@caltech.edu}{yixin@caltech.edu}. We appreciate valuable discussions with 
Yonghong An, Jeremy Fox, Michael Keane, Philip Haile, 
Yingyao Hu, Yao Luo, Luciano Pomatto, Yuya Sasaki, Robert Sherman, Xun Tang,  Omer Tamuz, and Ao Wang. We thank seminar and conference participants at Caltech, CUFE, Peking University, Tsinghua University, UC Irvine, University of Michigan, USC, 2025 Asian Summer School in Econometrics and Statistics, Econometric Society World Congress 2025, DSE Conference 2025, and BSE Summer Forum 2026.
Financial support from the Ronald and Maxine Linde Institute of Economic and Management Sciences is gratefully acknowledged.}
}
\author{Fan Wu 
 \qquad 
 Yi Xin 
 }
  
\maketitle
\vspace*{-0.7cm}
\begin{abstract}
\singlespacing
We propose a novel method for estimating nonseparable selection models. 
We show that, for a given selection function and under suitable contraction
conditions, the potential outcome distributions are nonparametrically identified from the selected outcome distributions and can be recovered using a simple iterative algorithm based on a contraction mapping. This result enables a full-information approach to estimating selection models without imposing parametric or separability assumptions on the outcome equation. We propose a two-step estimation strategy for the potential outcome distributions and the parameters of the selection function and establish the consistency and asymptotic normality of the resulting estimators. Monte Carlo simulations demonstrate that our approach performs well in finite samples. The method is applicable to a wide range of empirical settings, including consumer demand models with only transaction prices, auctions with stochastic winning rules and incomplete bid data, and Roy-type models with data on accepted wages.

\bigskip

\noindent \textbf{Keywords}: Sample Selection, Nonseparable Models, Functional Contraction, Potential Outcome Distribution, Semiparametric Estimation, Demand Estimation, Auction, Roy Models.

\noindent \textbf{JEL Codes}: C14, C24, C51, L11, D44, J31.
\bigskip
\end{abstract}
\newpage

\section{Introduction}
\label{eq:intro}

Sample selection issues arise in many empirical settings. 
In consumer demand studies, researchers often observe only the transaction prices of chosen products
\citep{goldberg1996dealer, cicala2015does, crawford2018asymmetric, allen2019search, d2019automobile, sagl2023dispersion, cosconati2024competing}. In auctions, available data may consist solely of the winning bids
\citep{athey2002identification, komarova2013new, guerre2019nonparametric,allen2024resolving}. 
In labor economics, wage data are typically observed only for individuals who choose to work \citep{gronau1974wage, heckman1974shadow}, and in the Roy model \citep{roy1951some}, earnings within an occupation are observed only for those who self-select into that sector.

Observing only a selected sample of outcomes---such as prices, bids, or wages---poses significant challenges for estimating two key elements.
First is the selection function, which specifies how agents choose among alternatives, for example, through a consumer demand system, an auction’s winning rule, or a labor force participation decision.
Second is the distribution of outcomes \emph{prior to selection}, often referred to as ``potential outcomes" in the literature. 
Flexibly estimating potential outcome distributions is crucial in many empirical contexts, such as analyzing price distributions to understand firms' pricing strategies and wage distributions to examine inequality.

Our paper proposes a new approach to estimating nonseparable selection models by exploiting a previously unrecognized one-to-one mapping between the outcome distributions before and after selection. 
Our key contribution is a constructive identification result showing that, for a given selection function and under suitable contraction conditions, the potential outcome distributions are nonparametrically identified from the selected outcome distributions and can be recovered using a simple iterative algorithm.
Consequently, the remaining object to be estimated is the selection function, which can be recovered from observed choice patterns.
Our method enables a full-information approach to estimating selection models without imposing parametric or separability assumptions on the outcome equation.

Formally, we consider a discrete choice problem in which each price-bearing alternative is associated with a potential outcome distribution. A selection function maps a vector of realized potential outcomes to a probability distribution over the alternatives. For example, in the consumer demand setting, each alternative represents a product, and the potential outcome is the offered price, with the selection function micro-founded by the consumer's utility maximization problem. We allow the outcome equations to be fully nonparametric with nonseparable error terms and to vary flexibly across different alternatives.
We assume that potential outcomes across alternatives are independent conditional on observable and unobservable characteristics, which allows for correlation across outcomes when conditioning only on observables. Our framework also permits the unobservable component in the outcome equation to enter directly into agents’ preferences over alternatives, thereby capturing selection on unobservables.

We analyze how the selection model
maps the potential outcome distributions to the distributions of selected outcomes and seek to \emph{invert} the mapping. 
The key insight of our approach is that, given the selection model and potential outcome distributions across the price-bearing alternatives, we can derive the likelihood of an outcome being selected at each price. Conversely, if this selection
likelihood were known, we could recover the potential outcome distributions from
the selected outcome distributions using Bayes’ rule.  
This two-way relationship characterizes a fixed-point problem.

Building on this intuition, we construct an operator whose fixed
point is the potential outcome distributions and establish sufficient conditions for it to be a
functional contraction (Theorems \ref{thm: contraction} and \ref{thm: contraction special}). 
Our results imply that, given the selection function and the distributions of selected outcomes, we can nonparametrically identify the potential outcome distributions under the contraction condition. Importantly, the contraction condition is automatically satisfied with at most two price-bearing alternatives, including a binary decision between one product and the outside option. Moreover, this identification result is constructive:
starting with any initial guess for the potential outcome distributions, we iteratively
apply the operator. This process converges to the potential outcome distributions associated with the selection function.

We then embed this identification result into a two-step estimation strategy for the unobserved potential outcome distributions and the parameters of the selection function. In the first step, we estimate the selected outcome distribution conditional on both observable and unobservable covariates using the instrumental variable approach from the nonclassical measurement error literature \citep{hu2008identification,hu2008instrumental}. 
In the second step, we propose a nested fixed-point algorithm to estimate the parameters of the selection function: in the inner loop, for any candidate selection function, we recover the potential outcome distributions by iterating the operator, while in the outer loop we search for the parameter values that maximize the likelihood of the observed choice patterns. The potential outcome distributions are then obtained by reapplying the fixed-point algorithm at the estimated parameters of the selection function.

Theorems \ref{thm: consistency} and \ref{thm: normal} establish consistency of the proposed estimators and semiparametric asymptotic normality of the selection-parameter estimator. Corollary \ref{col:cdf_asymptotics} gives the convergence rate of the recovered outcome CDFs and, under an additional first-step functional central limit theorem, their Gaussian process limit. 
To examine their finite sample properties, we conduct Monte Carlo
simulations across various designs of the outcome equation. Our results show that the biases in our estimators are generally small, and the standard deviation decreases as the
sample size increases across all simulation designs. Our nonparametric estimation of the potential outcome distributions outperforms the classic Heckman parametric two-step approach and the quantile selection model of \cite{arellano2017quantile} with linear quantile functions and a Gaussian copula, particularly when the outcome equation contains nonseparable error terms.

In \cite{cosconati2024competing}, we apply our method to estimate consumer demand for auto insurance products when only transaction prices are observed. We nonparametrically estimate the offered price distribution for each insurance company and allow these distributions to vary fully flexibly across firms. The substantial heterogeneity in the recovered price distributions reflects differences in firms’ information technologies and cost structures, which are key primitives we estimate through a supply-side competition model. We omit the details of this application here and refer readers to \cite{cosconati2024competing} for the full empirical setup and results.

\paragraph{Related Literature}

Our paper contributes to the extensive theoretical literature on sample selection models. 
An early solution to sample selection bias is full information maximum likelihood (FIML) estimation based on parametric assumptions, as in \cite{heckman1974shadow} and \citeauthor{lee1982some} (\citeyear{lee1982some}, \citeyear{lee1983generalized}). 
More commonly employed methods for sample selection models are the two-step control function approach pioneered by \citeauthor{heckman1976common} (\citeyear{heckman1976common}, \citeyear{heckman1979sample}). 
A substantial body of theoretical work has been developed to relax the distributional assumptions in the two stages of the estimation procedure \citep[see, among many others,][]{ahn1993semiparametric,andrews1998semiparametric,chen2003semiparametric, das2003nonparametric, newey2007nonparametric,newey2009two,fernandez2024nonseparable,chernozhukov2025distribution}. For a comprehensive survey of semiparametric two-step estimation methods for selection models, see \cite{vella1998estimating}.

Our approach provides an alternative framework that differs from existing methods along several dimensions.
First, we allow the outcome equation to be nonparametric and nonseparable in the error terms, and we exploit the full information in the selected outcome distribution to recover the entire distribution of potential outcomes.  \cite{newey2007nonparametric} and \cite{fernandez2024nonseparable} use control function approaches to correct for sample selection in nonseparable models with binary and censored selection rules, respectively, and they focus on identifying certain global and local parameters of the  outcome distribution.

Second, our method accommodates fully heterogeneous effects of covariates on outcomes.
By contrast, many existing approaches that estimate conditional mean models restrict covariates to affecting only the location of the outcome distribution.
Notable exceptions include \cite{arellano2017quantile}, who develop a sample-selection correction for quantile regression by modeling the copula of the error terms in the outcome and selection equations, and \cite{chernozhukov2025distribution}, who propose a semiparametric generalization of the Heckman selection model that allows for rich forms of heterogeneity in the effects of covariates on both outcomes and selection.

Third, our approach does not require the conventional exclusion restriction used in many sample selection models, namely a variable that shifts the choice probability without directly affecting outcomes.
Such variables are central to identification in many two-step procedures and may not always be readily available in practice
(see \cite{vella1998estimating} and \cite{honore2020selection} for further discussion).
Furthermore, our method does not rely on identification-at-infinity arguments.\footnote{
Relatedly, \cite{d2013another} and \cite{d2018extremal} develop estimation methods for semiparametric sample selection models without an instrument or a large-support regressor, leveraging the independence-at-infinity assumption. 
\cite{lee2013nonparametric} study nonparametric identification in competing risks models without exclusion restrictions or identification-at-infinity (or near-zero) arguments.
\cite{honore2020selection} derive identified sets for sample selection models without exclusion restrictions.}

Instead, we assume that the potential outcomes are conditionally independent
given the observed and unobserved characteristics, and our method relies on instruments satisfying the conditional independence requirements of the nonclassical measurement-error literature, as in \cite{hu2008identification} and \cite{hu2008instrumental}, to estimate the selected outcome distributions conditional on the unobservable in the first step.
These instruments serve a different role from the exclusion restrictions commonly used in sample selection models.
Neither identifying strategy is uniformly weaker than the other; rather, they represent alternative sources of identifying variation that may be more or less plausible depending on the application.

Estimation of the selection function in our model is closely related to the demand estimation literature following the seminal work of \cite{berry1994estimating} and \cite{berry1995automobile}. In particular, observed choice patterns play the same role as market shares in recovering consumer preference parameters.
Our method addresses the problem of missing full price menu that arises in many demand estimation contexts \citep[e.g.,][]{goldberg1996dealer,cicala2015does,crawford2018asymmetric,allen2019search, d2019automobile, sagl2023dispersion, cosconati2024competing}, an issue that is especially relevant in the presence of price discrimination or personalized pricing.\footnote{A recent paper by \cite{d2019automobile}
addresses a related challenge in demand estimation under unobserved price discrimination
by imposing supply-side restrictions, such as assumptions about firm conduct (e.g., Bertrand competition), and assuming identical costs across consumers.
}

At a broader conceptual level, 
our reliance on the structural restrictions implied
by the selection model resonates with the nonparametric identification literature on auction models with missing bids and competing risks models \citep[see, e.g.,][]{athey2002identification, komarova2013new,meilijson1981estimation,guerre2019nonparametric}.\footnote{\cite{athey2002identification} show that the symmetric independent private values (IPV) models are identified with the transaction price by exploiting a one-to-one mapping between an order statistic and its parent distribution. 
\cite{komarova2013new} analyzes asymmetric second-price auctions where only the winning bids and the winner's identity are observed.
A related result for generalized competing risks models can be found in \cite{meilijson1981estimation}.
More recently, \cite{guerre2019nonparametric} examine nonparametric identification of symmetric IPV first-price auctions with only winning bids, accounting for unobserved competition.}
In these models, the selection rule is deterministic (e.g., the highest bidder wins), which allows order-statistic arguments to be applied. In contrast, our selection model assigns a probability distribution over alternatives and is therefore closer in spirit to multi-attribute auction environments \citep[see, e.g.,][]{krasnokutskaya2020role}. 
Because selection is probabilistic rather than determined by an extremal outcome, classical order-statistic techniques are not applicable. Instead, our identification strategy relies on a fixed-point characterization of the selection process.
Moreover, our framework can flexibly accommodate asymmetries across alternatives, whereas bidder asymmetries are known to pose significant challenges in auction models (see the discussion in the handbook chapter by \cite{athey2007nonparametric}).

\paragraph{Outline} The rest of the paper is organized as follows. Section \ref{sec:model_results} formally introduces our model and  
presents the main theoretical results. In Section \ref{sec:estimation}, we describe our estimation strategy and establish the asymptotic properties of the estimators. Section \ref{sec:monte_carlo} reports results from our Monte Carlo simulations, and Section \ref{sec:applications} discusses the empirical application in \cite{cosconati2024competing}. Section \ref{sec:conclusion} concludes. All proofs are collected in the appendix.

\section{Model and Main Results}
\label{sec:model_results}

In this section, all analyses are conditional on a vector of characteristics $(x, x^\ast)$, where $x$ denotes observables and $x^\ast$ denotes unobservables. 
The structure of the model and the main theoretical results \emph{do not} depend on whether the conditioning set includes unobserved components, although the presence of unobservables introduces additional challenges for estimation, which we address in Section \ref{sec:estimation}. Because all results in this section are stated conditional on $(x, x^\ast)$, we omit these variables from the notation to simplify exposition.

Consider a discrete choice problem with complete choice set
$\mathcal J=\{0,1,\ldots,J\}$, where $0$ denotes the outside option
and $1,\ldots,J$ are the price-bearing alternatives. 
For $j=1,\ldots,J$, let
\[
 \mathcal P_j=[\underline p_j,\overline p_j],
 \qquad
 \mathcal P=\prod_{j=1}^J\mathcal P_j,
\]
where $\underline p_j<\overline p_j$ are fixed and finite.
Let $G_j\in\Delta(\mathcal P_j)$ denote the offered price distribution
of alternative $j$, where $\Delta(A)$ is the set of probability
measures on $A$. Write $G=(G_1,\ldots,G_J)$ for the collection of
marginals and, when integrating over the price vector, identify it
with the product measure $\bigotimes_{j=1}^JG_j$.
Unless otherwise indicated, indices in products of price measures
and in price-distribution operators run from $1$ to $J$; sums over
the complete choice set also include $0$.

We assume that the offered prices $p_j\sim G_j$ are independently distributed across alternatives, \emph{conditional on $x$ and $x^\ast$}. This assumption allows
prices across alternatives to be correlated conditional on observables. For example, if $x^\ast$ captures a worker's unobserved productivity, wage offers from different firms may appear correlated when conditioning only on observable characteristics.
However, the assumption imposes a more restrictive dependence structure: any residual correlation across prices must be mediated by $x^\ast$, rather than allowing prices to be arbitrarily correlated conditional on observables.
Nevertheless, $x^\ast$ is not restricted to be one-dimensional and may instead represent a high-dimensional vector of latent characteristics. 
As the conditioning information becomes richer, more of the dependence across alternatives is captured by the observed and latent characteristics, leaving less residual dependence to be ruled out by the conditional independence assumption.

A \textit{selection function} is
$f=(f_0,f_1,\ldots,f_J)$, where
$f_j(\boldsymbol p)$ is the probability of choosing $j\in\mathcal J$
given the vector of inside prices $\boldsymbol p=(p_1,\ldots,p_J)$.
For $j=1,\ldots,J$, assume that $f_j$ is continuous and strictly
positive on $\mathcal P$. The outside probability is
\[
 f_0(\boldsymbol p)=1-\sum_{j=1}^J f_j(\boldsymbol p)\geq0,
 \qquad
 \sum_{j=0}^J f_j(\boldsymbol p)=1.
\]
Thus an outside choice is allowed without introducing an additional
price coordinate. The contraction and price-identification results
below concern only $j=1,\ldots,J$ and require strict positivity only
for these alternatives. The case without an outside option is
obtained by setting $f_0\equiv0$ and omitting alternative $0$ from
the choice criterion.

The strict positivity assumption ensures that every price realization has a positive probability of generating an observed choice and therefore can, in principle, be observed in the data. To see why this is important, suppose instead that there exists an interval $A \subset [\underline{p}_j, \overline{p}_j]$ such that whenever $p_j \in A$, the probability of choosing alternative $j$ is zero. Then prices in $A$ are effectively censored from the data; that is, the data contain no information about the frequency with which such prices are offered.
As a result, $G_j$ cannot be identified on $A$ without imposing additional restrictions.\footnote{The strict positivity assumption is analogous to the overlap assumption in the treatment effect literature, which requires each individual to have a positive probability of receiving each treatment level. 
The strict positivity assumption excludes deterministic selection rules, such as standard auction, competing-risks, and pure Roy models, in which some alternatives are selected with probability zero for certain realizations of $\boldsymbol p$.
Identification in these settings typically relies on arguments based on ordered statistics, which are well established in the literature. See the discussion in the handbook chapter by \cite{athey2007nonparametric}.
}

The selection function $f$ is taken as a primitive of our framework. Nevertheless, it can be microfounded through a standard utility-maximization problem.
Under particular assumptions on the distribution of utility shocks, familiar discrete choice models such as the multinomial logit and probit arise as special cases.
We also want to emphasize that the selection function $f$ can freely depend on observable and unobservable characteristics $(x, x^\ast)$. This flexibility allows the model to capture dependence between the potential outcome distribution and the selection mechanism. For example, unobserved productivity may simultaneously influence wage offers and labor-force participation decisions, thereby generating endogenous selection.

The \emph{offered} price distribution (or \emph{potential} outcome distribution) $G_j$ is generally not directly observed in the data, yet it is often the primary object of interest. 
Instead, data typically reveal only the \emph{selected} (or \emph{accepted}) outcome distribution. For example, researchers may observe the distribution of transacted prices or accepted wage offers, but not the distribution of all offered prices or wage offers. Let $G^{s}_j\in \Delta ([\underline p_j, \overline p_j])$ represent the price distribution of alternative $j$ conditional on $j$ being selected. 
Similarly, $G^{s}=\prod_{j=1}^J G^{s}_j$ denotes the collection of selected outcome distributions across the price-bearing alternatives.

In the remainder of this section, we establish identification of the potential outcome distributions $G$ from the selected outcome distributions $G^{s}$, holding the selection function $f$ fixed. 
We then turn in Section \ref{sec:estimation} to the empirically relevant case in which $f$ is specified up to a finite-dimensional parameter vector that is estimated from observed choice behavior.

\subsection{A Fixed-Point Characterization of Potential Outcome Distributions}\label{sec:fixed point density}

We begin with a density-based formulation that highlights the main ideas most transparently. Throughout Section \ref{sec:fixed point density}, we assume that all relevant distributions admit densities with full support and use lower-case letters to denote densities corresponding to probability measures denoted by upper-case letters. 
In particular, let $g_j$ and $g^{s}_j$ denote the densities of the probability measures $G_j$ and $G^{s}_j$, respectively. Similarly, let $g=\prod_{j=1}^J g_j$ and $g^{s}=\prod_{j=1}^J g^{s}_j$ denote the collections of potential and selected outcome densities across the price-bearing alternatives, respectively. A measure-theoretic formulation accommodating general probability measures, including discrete distributions, continuous distributions, and arbitrary mixtures thereof, is provided in Section \ref{sec:formal}.

The density of selected outcomes $g^{s}_j$ can be derived from the density of potential outcomes $g_j$ using Bayes' rule:
\begin{equation}\label{eq: expost}
    g^{s}_j(p_j)\propto   \operatorname{Pr}_j(p_j;g)  g_j(p_j), \text{ for }j=1,\ldots,J. 
\end{equation}
In Equation (\ref{eq: expost}), $\operatorname{Pr}_j(p_j;g)$ represents the probability of selecting alternative $j$ conditional on $p_j$ and is given by 
\begin{align}
\label{eq:choice_prob}
\operatorname{Pr}_j(p_j;g)=\int_{\boldsymbol{p}_{-j}}f_j(p_j, \boldsymbol{p}_{-j})\prod_{k\neq j}g_k(p_k)d\boldsymbol{p}_{-j},
\end{align}
where $\boldsymbol{p}_{-j}=(p_1, \ldots, p_{j-1}, p_{j+1}, \ldots, p_J)$ denotes the vector of prices excluding $j$'s price. 
Because we assume that prices are independently distributed across alternatives conditional on \((x,x^\ast)\),
the joint distribution of \(\boldsymbol p_{-j}\) factors into the product of the marginal densities,
\(\prod_{k\ne j} g_k(p_k)\).

Collecting Equations \eqref{eq: expost} and \eqref{eq:choice_prob} across the price-bearing alternatives $j=1,\ldots,J$ defines a forward mapping from the potential outcome distributions to the selected outcome distributions, where the selection mechanism is governed by the selection function $f$ through the conditional choice probabilities. 
Our objective, however, is the inverse: to recover the unobserved potential outcome distributions $g$ from the observed selected outcome distributions $g^{s}$.

Whether such recovery is possible depends critically on the information contained in the observed data. In this regard, the conditional independence assumption plays a key role in identification.
Because the data reveal only the price of the selected alternative, they contain no direct information about the dependence structure of prices across alternatives. 
As a result, 
any correlation among competing alternatives is generally unidentifiable.\footnote{See \cite{heckman1990empirical} for an early discussion of this nonidentification result; see also the survey chapter by \cite{french2011identification} for a related discussion.} For example, if prices tend to move together across alternatives, this dependence would not be detectable from the observed data.
The conditional independence assumption circumvents this problem by restricting the joint distribution of prices to be determined by the marginal distributions. As a result, for a given selection function, the system defined by Equations \eqref{eq: expost} and \eqref{eq:choice_prob} contains as many unknown objects as identifying restrictions: both $g$ and $g^{s}$ consist of $J$ densities.
Even so, identification remains  challenging because it requires recovering a collection of infinite-dimensional objects from a nonlinear system of equations.

To address this challenge, we now derive an inversion of the forward mapping defined by Equations \eqref{eq: expost} and \eqref{eq:choice_prob}. In particular, dividing both sides of Equation \eqref{eq: expost} by the selection probability $\operatorname{Pr}_j(p_j;g)$ yields
\begin{equation}
    \label{eq:inversion}
   g_j(p_j)\propto    g^{s}_j(p_j)/\operatorname{Pr}_j(p_j;g), \text{ for }j=1,\ldots,J.
\end{equation}
We illustrate this inversion process using the simulated example in Figure \ref{fig:inversion}.
The red solid line plots the selected price density for an alternative. Dividing this density by the probability that the alternative is chosen at each price, $\operatorname{Pr}_j(p_j;g)$, yields the unnormalized offered price density shown by the blue dashed line.
The gap between these two densities captures the selection mechanism: when a lower price is offered, agents are more likely to accept it, whereas higher prices make them more likely to choose other alternatives.
The offered price distribution shown by the blue solid line is then obtained after normalization.

\begin{figure}[htbp!]
    \centering
    \includegraphics[width=0.6\linewidth]{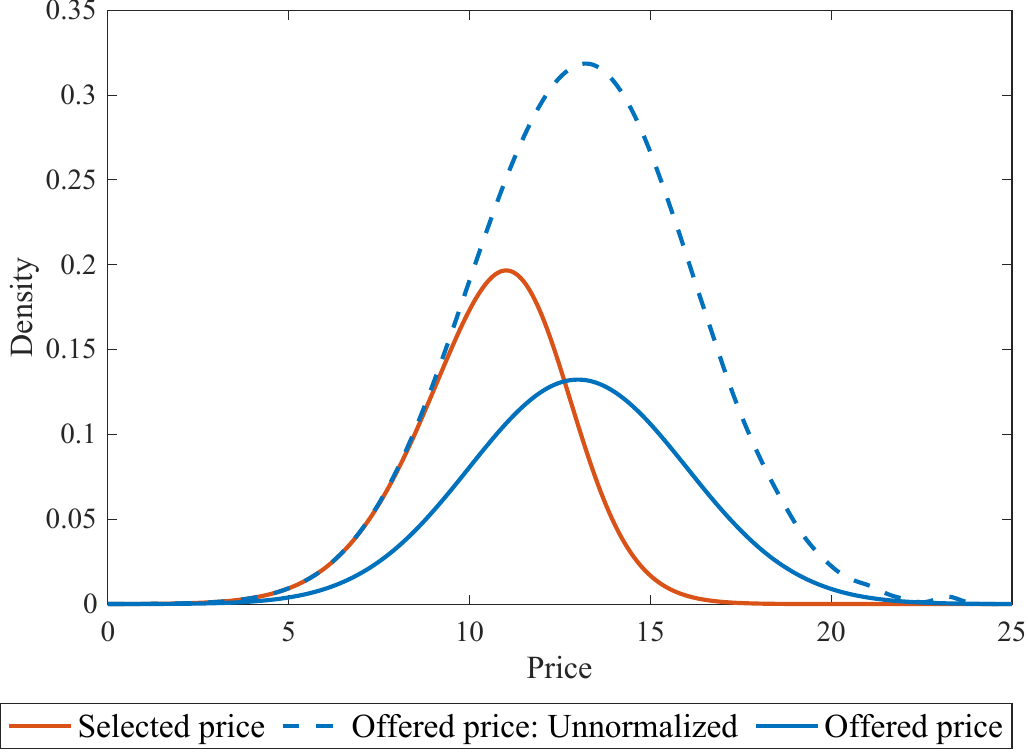}
    \caption{Densities of offered and selected prices. We draw offered prices from $\mathcal{N}(13,9)$, and the probability that the agent given price $p$ chooses this alternative is given by $\exp(10-p)/(0.1+\exp(10-p))$.}
    \label{fig:inversion}
\end{figure}

Equation (\ref{eq:inversion}) plays a central role in our analysis. 
Note that if the selection probability $\operatorname{Pr}_j(\cdot;g)$—that is, the probability of selecting product $j$ conditional on its offered price—were known, then recovering the offered price distribution from Equation (\ref{eq:inversion}) would be straightforward. However, $\operatorname{Pr}_j(\cdot;g)$ is \emph{not} known, because it depends on the offered price distribution $g$, which we seek to recover. 
A tentative solution is to start with a conjecture $\psi$ for the offered price distribution and use it to compute the implied selection probability $\operatorname{Pr}_j(\cdot;\psi)$. Equation (\ref{eq:inversion}) then delivers an \emph{updated} conjecture of the offered price distribution. 
This procedure, which maps a conjectured offered price distribution into its updated version, defines an operator
$T\colon \prod_j \Delta([\underline p_j, \overline p_j])\to \prod_j \Delta([\underline p_j, \overline p_j])$ as follows.\footnote{Here, densities are identified with their induced probability measures when specifying the domain of $T$.}
\begin{equation}\label{eq:operator}
(T\psi)_j(p_j)=\frac{
 g^{s}_j(p_j)/\operatorname{Pr}_j(p_j;\psi)
 }{
 \int_{\underline p_j}^{\overline p_j}
 g^{s}_j(p)/\operatorname{Pr}_j(p;\psi)dp
 }, \text{ for }j=1,\ldots,J.
\end{equation} where $\psi=(\psi_1,\psi_2,\ldots,\psi_J)\in \prod_j \Delta([\underline p_j, \overline p_j])$. 
The denominator is the normalizing constant that ensures $(T\psi)_j$ integrates to one.

Importantly, if the conjecture $\psi$ is correct, i.e., $\psi=g$, then the selection probability $\operatorname{Pr}_j(\cdot;\psi)$ is correct, ensuring that the updated conjecture $T\psi$ also equals $g$. Thus, the offered price distribution $g$ is a fixed point of the operator $T$.

The operator $T$ is a contraction if there exists some real number $0\leq \rho<1$ such that for all $\psi, \phi\in \prod_j \Delta([\underline p_j, \overline p_j])$,
\begin{equation}\label{eq:contraction definition}
    D(T\psi, T\phi)\leq \rho D(\psi, \phi),
\end{equation}
whenever $D(\psi,\phi)<\infty$, for an extended metric $D$. The fixed-point argument is applied on a complete finite-distance component; after one update every initial conjecture lies in the component determined by the selected distributions. 
We now introduce a suitable metric.

\paragraph{Hilbert's Projective Metric}
\label{sec:metric}

For two densities $\psi_j,\phi_j \in \Delta([\underline{p}_j,\overline{p}_j])$, we measure their distance using Hilbert’s projective metric. If
$\psi_j$ and $\phi_j$ are strictly positive on the same support,
\begin{equation}
\label{eq:hilbert}
d_H(\psi_j,\phi_j)
=
\ln
\bigg(\frac{
\operatorname*{sup}_{p\in[\underline{p}_j,\overline{p}_j]}
\frac{\psi_j(p)}{\phi_j(p)}}{\operatorname*{inf}_{p\in[\underline{p}_j,\overline{p}_j]}
\frac{\psi_j(p)}{\phi_j(p)}}
\bigg ),
\end{equation}
otherwise set $d_H(\psi_j,\phi_j)=+\infty$.

The choice of this metric is deliberate and closely tailored to the structure of the operator $T$. 
In particular, 
the usefulness of Hilbert's projective metric in our setting comes from its invariance to positive rescaling. Since the operator $T$ updates each density by reweighting $g^{s}_j$ by the reciprocal of selection probability $\operatorname{Pr}_j(\cdot;\psi)$ and then renormalizing, both the common factor $g^{s}_j$ and the normalizing constants cancel out when computing the projective distance.
Therefore, for any two conjectured densities $\psi$ and $\phi$, 
\begin{align*}
 d_H((T\psi)_j,(T\phi)_j)  & =   \ln
\bigg(\frac{
\operatorname*{sup}_{p\in[\underline{p}_j,\overline{p}_j]}
\frac{(T\psi)_j(p)}{(T\phi)_j(p)}}{\operatorname*{inf}_{p\in[\underline{p}_j,\overline{p}_j]}
\frac{(T\psi)_j(p)}{(T\phi)_j(p)}}
\bigg )  \\
& = \ln \bigg( \frac{\operatorname*{sup}_{p\in[\underline{p}_j,\overline{p}_j]} \frac{
 g^{s}_j(p)/\operatorname{Pr}_j(p;\psi)
 }{
 \text{norm.\ const.\ for $\psi$}
 }
 \frac{\text{norm.\ const.\ for $\phi$}
 }{
 g^{s}_j(p)/\operatorname{Pr}_j(p;\phi)
 }
 }{
 \operatorname*{inf}_{p\in[\underline{p}_j,\overline{p}_j]} \frac{
 g^{s}_j(p)/\operatorname{Pr}_j(p;\psi)
 }{
 \text{norm.\ const.\ for $\psi$}
 }
 \frac{\text{norm.\ const.\ for $\phi$}
 }{
 g^{s}_j(p)/\operatorname{Pr}_j(p;\phi)
 }
 }  \bigg) \\
& =\ln
\bigg(\frac{
\operatorname*{sup}_{p\in[\underline{p}_j,\overline{p}_j]}
\frac{\operatorname{Pr}_j(p;\phi)}{\operatorname{Pr}_j(p;\psi)}}{\operatorname*{inf}_{p\in[\underline{p}_j,\overline{p}_j]}
\frac{\operatorname{Pr}_j(p;\phi)}{\operatorname{Pr}_j(p;\psi)}}
\bigg )  = d_H(\operatorname{Pr}_j(\cdot;\phi),\operatorname{Pr}_j(\cdot;\psi)). 
\end{align*}
This equation reveals an important isometry: measuring the distance between $(T\psi)_j$ and $(T\phi)_j$ is equivalent to measuring the distance between the corresponding selection probabilities under Hilbert's projective metric. 
Consequently, the contraction properties of the operator $T$ can be characterized by studying how perturbations in the conjectured densities affect the implied choice probabilities.

\subsection{General Probability Measure Formulation}
\label{sec:formal}

The density formulation in Section \ref{sec:fixed point density} is useful for illustrating the key ideas but is not essential for the results. We now present the corresponding measure-theoretic formulation for arbitrary probability measures and adopt this framework for the rest of the paper.

As before, the selected outcome distribution $G^{s}_j$ can be derived from the potential outcome distribution $G_j$ via Bayes' rule.
For any Borel set $A\subseteq[\underline p_j,\overline p_j]$, the counterpart of Equation \eqref{eq: expost} for general probability measures is: 
\begin{equation}\label{eq: expost_app}
    G^{s}_j(A)
    =
    \frac{
        \int_A \operatorname{Pr}_j(p_j;G)\,dG_j(p_j)
    }{
        \int_{\underline p_j}^{\overline p_j} \operatorname{Pr}_j(p_j;G)\,dG_j(p_j)
    }, \text{ for }j=1,\ldots,J, 
\end{equation}
where 
\begin{equation}
\label{eq:choice_prob_app}
\operatorname{Pr}_j(p_j;G)=\int_{\boldsymbol{p}_{-j}}f_j(p_j, \boldsymbol{p}_{-j})\prod_{k\neq j}dG_k(p_k).
\end{equation}
Equation \eqref{eq: expost_app} defines a mapping $F\colon \prod_j \Delta([\underline p_j, \overline p_j])\to \prod_j \Delta([\underline p_j, \overline p_j])$ that maps the potential outcome distribution $G$ to the selected outcome distribution $G^{s}$.

The inversion argument developed in the density-based framework extends directly to the present setting. In particular, the counterpart of Equation \eqref{eq:inversion} is
\begin{equation}
    \label{eq:inversion_app}
   G_j(A)
    =
    \frac{
        \int_A dG^{s}_j(p_j)/\operatorname{Pr}_j(p_j;G)
    }{
        \int_{\underline p_j}^{\overline p_j} dG^{s}_j(p_j)/\operatorname{Pr}_j(p_j;G)
    }, \text{ for }j=1,\ldots,J. 
\end{equation}
Equation \eqref{eq:inversion_app} defines an operator 
$T\colon \prod_j \Delta([\underline p_j, \overline p_j])\to \prod_j \Delta([\underline p_j, \overline p_j])$ as follows.
\begin{equation}\label{eq:operator_app}
(T\Psi)_j(A)
    =
    \frac{
        \int_A dG^{s}_j(p_j)/\operatorname{Pr}_j(p_j;\Psi)
    }{
        \int_{\underline p_j}^{\overline p_j} dG^{s}_j(p_j)/\operatorname{Pr}_j(p_j;\Psi)
    }, \text{ for }j=1,\ldots,J, 
\end{equation} where $\Psi=(\Psi_1,\Psi_2,\ldots,\Psi_J)\in \prod_j \Delta([\underline p_j, \overline p_j])$.

For two probability measures $\Psi_j,\Phi_j \in \Delta([\underline{p}_j,\overline{p}_j])$, 
let $d_H(\Psi_j,\Phi_j)$ denote their distance under the general form of Hilbert's projective metric; see Appendix \ref{sec:Radon Nikodym} for details.
To account for the $J$ price-bearing alternatives, we define a metric in the space $ \prod_j \Delta([\underline p_j, \overline p_j])$ by taking the maximum distance among all alternatives: 
$$D(\Psi, \Phi)=\max_{1\leq j\leq J} d_H(\Psi_j, \Phi_j)$$
for any $\Psi,\Phi\in\prod_{j=1}^J\Delta(\mathcal P_j)$.
Henceforth, the product of the $J$ inside probability spaces is equipped
with this extended metric $D$.

\subsection{Functional Contraction}
\label{sec:contraction}
We now provide a primitive sufficient condition under which the Bayes-inversion operator \(T\) is a contraction. We first define the \textit{relative selection probability variation}:
\[ \Delta_j = \sup_{\substack{ p_j,p_j'\in [\underline p_j,\overline p_j] \\ \boldsymbol p_{-j},\boldsymbol p_{-j}'\in \prod_{k\neq j}[\underline p_k,\overline p_k] }} \left| \log \frac{f_j(p_j,\boldsymbol p_{-j})}{f_j(p_j',\boldsymbol p_{-j})} - \log \frac{f_j(p_j,\boldsymbol p_{-j}')}{f_j(p_j',\boldsymbol p_{-j}')} \right|. \] Intuitively, \(\Delta_j\) measures how much the relative selection probability of alternative \(j\) between two own-price realizations, \(p_j\) and \(p_j'\), can vary as competitors' prices change from \(\boldsymbol p_{-j}\) to \(\boldsymbol p_{-j}'\).\footnote{In Perron-Frobenius theory, $\Delta_j$ is also known as the projective diameter. Geometrically, \(\Delta_j\) measures the departure of \(f_j\) from multiplicative separability between \(p_j\) and \(\boldsymbol p_{-j}\). If \(f_j(p_j,\boldsymbol p_{-j})=a(p_j)b(\boldsymbol p_{-j})\), then \(\Delta_j=0\).}

Let
\[
\rho
=
(J-1)
\max_{1\leq j\leq J}
\tanh\left(\frac{\Delta_j}{4}\right),
\]
where $\tanh(\cdot)$ denotes the hyperbolic tangent function.  Since \(\Delta_j\ge 0\) is finite under the maintained assumption that $f_j$ is continuous and strictly positive on compact support, it follows that \(\tanh(\Delta_j/4)\in[0,1)\).

\begin{thm}[Contraction]\label{thm: contraction}
    If $\rho<1$, the operator $T$ is a contraction with modulus at most $\rho$. In particular, the condition holds when $J\leq 2$, whether or not an outside option is present. 
\end{thm}
\begin{proof}
    See Appendix \ref{sec:proof_1}.
\end{proof}

Theorem~\ref{thm: contraction} establishes a key identification result for selection models.
Whenever $\rho<1$ (which is automatically satisfied when $J\leq2$), the Bayes-inversion operator $T$ admits a unique fixed point. 
Since the offered price distribution $G$ is characterized as the fixed point of $T$, it follows that $G$ is uniquely determined by
$G^{s}$ given the selection function $f$. Thus $G$ is identified given $f$.
Notably, the theorem imposes no assumptions on the functional form of the potential outcome distributions, allowing the outcome equation to be fully nonparametric and nonseparable in the error terms.\footnote{Theorem \ref{thm: contraction} applies more generally to any selection function \(f\) that is strictly positive and has finite projective diameter. Thus, the continuity and compact-support assumptions on the selection function $f$ are not essential to this theorem.}
We formally state the identification results in the following corollary.

\begin{col}[Identification]\label{col: identification}
    Suppose $\rho<1$. Then, for every selected price distribution $G^{s}$, there exists a unique offered price distribution $G$ consistent with $G^{s}$ given the selection function $f$. Moreover, for any initial distribution $\Psi\in \prod_j \Delta([\underline p_j, \overline p_j])$, $G=\lim_{n\to \infty}T^n \Psi$. Hence, the offered price distribution $G$ is nonparametrically identified from $G^{s}$ given $f$.
\end{col}

Corollary \ref{col: identification} provides a \emph{constructive} identification result for the offered price distribution $G$.
Take any initial distribution $\Psi\in \prod_j \Delta([\underline p_j, \overline p_j])$, by Theorem \ref{thm: contraction},
$$D(T^n \Psi,G)=D(T^n \Psi,TG)\leq \rho D(T^{n-1} \Psi,G)\leq \rho^{n-1} D(T \Psi,G),$$
where $D(T \Psi,G)$ is finite.\footnote{Appendix \ref{sec:Radon Nikodym} shows that $D(T \Psi,G)$ is finite for any initial distribution $\Psi$ under the maintained assumption that $f_j$ is continuous and strictly positive on compact
support.} This implies
\begin{align*}
    &\lim_{n\to \infty}D(T^n \Psi,G)=0, \\
    &\lim_{n\to \infty}T^n \Psi=G.
\end{align*}
Thus, starting from any initial conjecture, repeated application of \(T\) converges geometrically to the potential outcome distributions associated with the selection function.

\paragraph{Interpreting the Contraction Condition.}
We now provide some economic intuition for the factors that determine the strength of the contraction in Theorem~\ref{thm: contraction}.

First, the factor $J-1$ counts the other unknown price
marginals entering the selection probability for a price-bearing
alternative. Adding the outside option does not add a price marginal
and therefore does not change this factor. With two price-bearing
alternatives ($J=2$), the contraction bound is automatically less
than one, even if the outside option is also available. With one
price-bearing alternative and the outside option ($J=1$), there is
no other price distribution to integrate out: $T$ is independent of
the conjecture and recovery takes a single Bayes update. This covers
binary purchase and labor-force participation decisions.

Another key primitive underlying the contraction condition is the relative selection probability variation $\Delta_j$. This quantity can be related to a familiar semi-elasticity measure.
Specifically, if \(\ln f_j\) is continuously differentiable in \(p_j\), define
\[
M_j
=
\sup_{p_j,\boldsymbol p_{-j},\boldsymbol p_{-j}'}
\left|
\frac{\partial \ln f_j(p_j,\boldsymbol p_{-j})}{\partial p_j}
-
\frac{\partial \ln f_j(p_j,\boldsymbol p_{-j}')}{\partial p_j}
\right|.
\]
By the mean value theorem,
\[
\Delta_j
\le
(\bar p_j-\underline p_j)M_j.
\]
The quantity $M_j$ measures the maximum variation in the own-price semi-elasticity of alternative $j$ across different configurations of competitors' prices. It is small when the own-price responsiveness of alternative $j$ remains
stable as competitors’ prices vary. This is likely to occur when alternative $j$ is highly differentiated or only weakly substitutable with other alternatives, so that changes in competitors' prices have little effect on consumers' sensitivity to $p_j$.
In the extreme case where the own-price semi-elasticity of alternative $j$ is independent of competitors' prices, $M_j=0$, which in turn implies $\Delta_j=0$.

In addition to the strength of interaction effects, as captured by the selection function $f_j$, the relative selection probability variation also depends on the range over which prices vary. Holding the selection function fixed, $\Delta_j$ is weakly increasing as the price support expands. 
Consequently, even when competitors' prices affect the responsiveness of demand to product $j$'s own price, the resulting $\Delta_j$ may remain small if prices vary over a sufficiently narrow range. Conversely, the same degree of interaction can generate a larger $\Delta_j$ when prices span a wider interval.

Finally, note that the contraction condition in Theorem \ref{thm: contraction} is sufficient rather than necessary. The operator $T$ may still be a contraction even when the upper bound exceeds one. Moreover, the contraction modulus established in the theorem is a global bound. Consequently, even when the operator is not globally contractive, it may still be locally contractive in a neighborhood of the fixed point.

This distinction is particularly important for empirical applications. While the global bound is driven by worst-case conjectures in Equation \eqref{eq:contraction definition}, estimation in practice is primarily concerned with regions near the economically relevant distributions implied by the observed data. Moreover, numerical inversion procedures based on fixed-point iterations are typically initialized from economically plausible distributions. 
As a result, local contraction may be sufficient to ensure stable recovery of the offered price distributions, even in settings where the global sufficient condition cannot be verified.

\subsection{Special Cases}
\label{sec:special}

Thus far, we have not imposed any structure on the selection function. For a general selection function, we have to take the supremum over prices to compute the relative selection probability variation. Now we impose an assumption on the selection function to determine where the supremum is attained.

\begin{asmp}[Log Supermodularity]\label{asmp: monotone}
    Suppose \(\ln f_j\) is continuously differentiable
in \(p_j\). For all $j=1,\ldots,J$ and $p_j\in[\underline p_j, \overline p_j]$, $\frac{\partial\ln f_j(p_j,\boldsymbol{p}_{-j})}{\partial p_j}$ is weakly increasing in each $p_k$ with $k\neq j$.
\end{asmp}
Under log supermodularity, the relative selection probability variation in Theorem \ref{thm: contraction}
is attained at the boundary. The result is stated below. Let
\[
\rho^*
=
(J-1)\max_{1\leq j\leq J}
\tanh\left\{
\frac{1}{4}
\left[
\ln \frac{f_j(\overline p_j,\overline{\boldsymbol p}_{-j})}{ f_j(\underline p_j,\overline{\boldsymbol p}_{-j})}
-\ln \frac{f_j(\overline p_j,\underline{\boldsymbol p}_{-j})
}{ f_j(\underline p_j,\underline{\boldsymbol p}_{-j})}
\right]
\right\}.
\]

\begin{thm}\label{thm: contraction special}
    Suppose that Assumption \ref{asmp: monotone} holds. If $\rho^*<1$, the operator $T$ is a contraction with modulus at most $\rho^*$.
\end{thm}
\begin{proof}
    See Appendix \ref{sec:proof_2}.
\end{proof}

Under Assumption \ref{asmp: monotone}, the modulus $\rho^\ast$ takes a much simpler form and is straightforward to compute. The log-supermodularity assumption holds in models widely adopted by empirical researchers.
For example, the multinomial logit model satisfies Assumption \ref{asmp: monotone}. 
However, Assumption \ref{asmp: monotone} may not hold for probit models with three or more alternatives; in such cases, the more general results in Theorem \ref{thm: contraction} can be applied.

\subsection{Identification without Contraction}
\label{sec:common_kernel}

The preceding results establish identification by proving that the
Bayes-inversion operator is a contraction. Identification itself,
however, requires only that the forward mapping $F$ be injective.
A restriction on the structure of the selection function can
deliver this uniqueness without imposing the contraction condition.

\begin{prop}[Identification under a common kernel]
\label{prop:common_kernel}
Suppose the given selection function admits the representation
\begin{equation}\label{eq:common_kernel}
 f_j(\boldsymbol p)=a_j(p_j)K(\boldsymbol p),
 \qquad j=1,\ldots,J,
\end{equation}
where $a_j:\mathcal P_j\to(0,\infty)$ and
$K:\mathcal P\to(0,\infty)$ are continuous.
The forward mapping $F$
defined by \eqref{eq: expost_app} is globally injective for every
finite $J$. Thus the offered price
distributions are nonparametrically identified from the selected
price distributions $G^s$, given $f$.
In particular, the result applies to the multinomial-logit rule
\begin{equation}\label{eq:common_kernel_logit}
 f_j(\boldsymbol p)
 =
 \frac{\exp(v_j(p_j))}
      {a_0+\sum_{k=1}^J\exp(v_k(p_k))},
 \qquad j=1,\ldots,J,
\end{equation}
where each $v_j:\mathcal P_j\to\mathbb R$ is continuous and
$a_0\geq0$ is a given outside-option weight.\footnote{Taking
\[
 a_j(p_j)=\exp(v_j(p_j)),
 \qquad
 K(\boldsymbol p)
 =
 \left(a_0+\sum_{k=1}^J\exp(v_k(p_k))\right)^{-1}
\]
verifies \eqref{eq:common_kernel}.}
\end{prop}

\begin{proof}
See Appendix \ref{sec:proof_common_kernel}.
\end{proof}

Proposition 1 does not require $\rho<1$, but impose a different type of assumption on $f$. The common-kernel restriction concerns the relationship between
selection probabilities across alternatives. In particular,
\[
 \frac{f_j(\boldsymbol p)}{f_k(\boldsymbol p)}
 =
 \frac{a_j(p_j)}{a_k(p_k)},
 \qquad j,k\in\{1,\ldots,J\}.
\]
It does not require $f_j$ to be multiplicatively separable between
$p_j$ and $\boldsymbol p_{-j}$: the common function $K$ may depend
nonseparably on the entire price vector.

Proposition \ref{prop:common_kernel} establishes
identification under multinomial-logit selection rule for any finite number of
alternatives.\footnote{ This implication concerns the displayed
conditional logit probabilities. Further averaging over unobserved
heterogeneity need not preserve the common-kernel representation;
a mixed-logit specification must therefore be checked separately.} It establishes uniqueness of the offered distributions for
model-generated selected distributions, but does not by itself
establish contraction or the convergence of the iterates $T^m\Psi$.
The contraction conditions in Theorems \ref{thm: contraction}
and \ref{thm: contraction special} provide both uniqueness and
geometric convergence for a different, more generally specified
class of selection functions.
The estimation and asymptotic results in Section
\ref{sec:estimation} continue to use their stated contraction
and regularity conditions.

\section{Estimation}
\label{sec:estimation}

We now turn to estimation of the offered price distributions $G$
and the selection function $f$. For each individual, the sampling
frame records the choice $j_i\in\mathcal J=\{0,1,\ldots,J\}$ and
observable characteristics $x_i\in X$; the selected price $p_i$ is
observed only when $j_i\geq1$. For a generic observation $\omega$,
$j$ denotes its realized choice; a dummy alternative index is
written as $k$ when needed to distinguish the two roles.
Let $x_{ij}$ be the observed characteristics of alternative $j$, and write
$x_i=(x_{i1}',\ldots,x_{iJ}')'$. Let $x_i^\ast\in X^\ast$ be an
unobservable characteristic that may affect both the selection
decision and potential outcome distributions.
Note that the unobserved characteristic need not be one-dimensional. More generally, $x^\ast$ may represent a vector of latent characteristics. In this case, $X^\ast$ denotes the corresponding product space of latent attributes.

Our theoretical results are developed under a highly flexible specification of the selection mechanism $f$.
In principle, the framework could be combined with parametric, semiparametric, or nonparametric approaches to estimating $f$. 
In practice, however, the selection mechanism may depend on the entire vector of prices and product characteristics, as well as consumer heterogeneity, rendering fully nonparametric estimation challenging.
We therefore adopt the standard differentiated-products demand framework as in \cite{berry1994estimating} and \cite{berry1995automobile}, which provides a tractable and widely used parameterization of the selection probabilities.
This parametric specification allows us to propose a simple estimation procedure while substantially reducing the data requirements needed to recover $f$, a feature that is particularly valuable when observed choice variation is limited.

In particular, we assume that the selection function $f$ is derived from a standard multinomial choice model with an indirect utility given by
$$u_{ij}=v_j(p_{ij}, x_{ij}, x_i^\ast, \varepsilon_{ij}; \theta), $$ where $v_j$ is a known function indexed by a finite-dimensional parameter vector $\theta$.
Here, $p_{ij}$ is the offered price of alternative $j$ for individual $i$, and the vector of unobserved shocks $\varepsilon_i=(\varepsilon_{i0}, \cdots, \varepsilon_{iJ})$ follows a known joint distribution, such as Type 1 extreme value.
The outside utility is $u_{i0}=v_0(x_i,x_i^\ast,\varepsilon_{i0};\theta)$,
with no price argument. Each individual chooses the alternative in
$\mathcal J$ that maximizes utility, and the selection function $f$ is captured by the parameter $\theta$. Throughout the paper, we use $\theta_0$ to denote the true parameter.

For example, a widely used specification takes the following form:
\begin{align}
    u_{ij} = \gamma p_{ij} + x_{ij}'\beta + \xi_j +  x_i^\ast\kappa_j+\varepsilon_{ij}, \quad j=1, 2, \cdots, J, \label{eq:utility_example}
\end{align}
where $\xi_j$ represents a scalar-valued unobserved characteristic of alternative $j$, such as product quality or brand loyalty.  
The term $x_i^\ast \kappa_j$ allows preferences for product $j$ to vary with the unobservable characteristic $x_i^\ast$. 
In this example, $\theta=(\gamma, \beta, \boldsymbol{\xi},\boldsymbol{\kappa})$, where $\boldsymbol{\xi}=(\xi_1, \cdots, \xi_J)$ and $\boldsymbol{\kappa}=(\kappa_1, \cdots, \kappa_J)$.

\subsection{Two-Step Estimation Strategy}
\label{sec:estimation_strategy}

We propose a two-step estimation procedure. In the first step, we estimate the selected outcome distribution conditional on both observable and unobservable covariates using instruments. Once the selected outcome distribution has been recovered, for any given selection function 
$f$, the potential outcome distribution can be recovered iteratively using the contraction mapping results in Section \ref{sec:model_results}.
The second step nests this fixed-point problem within an estimation routine that recovers the parameters of the selection function from agents' observed choice patterns. 
Using the resulting parameter estimates, we then re-run the fixed-point algorithm to recover the offered price distribution $G$.

\paragraph{Step 1: Estimating Selected Outcome Distribution}

We first estimate the selected outcome distributions $G^{s}$
conditional on $(x,x^\ast)$. When all relevant
covariates are observed, these distributions can be estimated directly,
for example by empirical or conditional distribution estimators. When
$x^\ast$ is unobserved, additional information is required.

We follow the proxy-variable approach of \citet{hu2008identification} and
its continuous-variable counterpart in \citet{hu2008instrumental}. Let
$z_1$ and $z_2$ be two instruments, and suppose that
$\omega_i=(x_i,j_i,p_i,z_{1i},z_{2i})$, $i=1,\ldots,n$, are i.i.d.
When additional proxies are used, they are included in $\omega_i$
as well. All first-step influence representations refer to this
complete observed record. Prices take values in the compact intervals
$\mathcal P_j$ introduced in Section
\ref{sec:model_results}.
We do not impose finite support on the covariates or instruments.
Write $H$ for a conditional probability measure and $h$ for a conditional
density, or a probability mass function for a discrete coordinate.
Here $p_i$ is recorded as missing for $j_i=0$ and is not used in
any outside-choice contribution below. For inside choices $j=1,\ldots,J$,
the proxy restriction is
\begin{equation}\label{eq: decomposition}
\begin{split}
&h_{p,z_1\mid z_2,x,j}(p,z_1\mid z_2,x,j)\\
&\quad=\int_{X^\ast}
 h_{p\mid x^\ast,x,j}(p\mid x^\ast,x,j)
 h_{z_1\mid x^\ast,x,j}(z_1\mid x^\ast,x,j)
 \,dH_{x^\ast\mid z_2,x,j}(x^\ast\mid z_2,x,j).
\end{split}
\end{equation}
The integral becomes a sum when the latent variable is discrete.
The same restriction can be stated directly in terms of conditional
distributions, without requiring the existence of densities. It expresses the conditional
independence of the price and the proxies given the latent type and
$(x,j)$.\footnote{An alternative specification uses three instruments
$(z_1,z_2,z_3)$ satisfying
\[
\begin{split}
&h_{z_3,z_1\mid z_2,x,j,p}(z_3,z_1\mid z_2,x,j,p)\\
&\quad=\int_{X^\ast}
 h_{z_3\mid x^\ast,x,j,p}(z_3\mid x^\ast,x,j,p)
 h_{z_1\mid x^\ast,x,j,p}(z_1\mid x^\ast,x,j,p)
 \,dH_{x^\ast\mid z_2,x,j,p}(x^\ast\mid z_2,x,j,p).
\end{split}
\]
This specification permits the instruments to depend on price while
maintaining the required conditional independence among the instruments.
The appropriate proxy-model identification conditions must still be
verified.}
The factorization alone does not establish identification: the relevant
rank or completeness, informativeness, and normalization conditions of
the proxy model are also required. We refer to
\citet{hu2008identification,hu2008instrumental} for these conditions in
their respective settings. 

In practice, finding such instruments is often feasible. For example, in the insurance pricing setting studied by
\citet{cosconati2024competing}, $x^\ast$ represents a consumer's unobserved
risk type, which affects the premiums offered by insurers. Realized
claims can provide proxies for this type. In labor applications, measures
of work performance or test scores may provide proxies for unobserved
productivity.

Use a first-step estimator appropriate for the specified proxy model,
such as a suitably regularized sieve likelihood estimator. Its selected
price component yields conditional CDF estimates
\[
\widehat{G^{s}}_j(t\mid x,x^\ast)
=\widehat H_{p\mid x^\ast,x,j}([\underline p_j,t]\mid x^\ast,x,j),
\qquad t\in \mathcal P_j.
\]
The other input to the second step is the latent-type distribution
conditional on $(x,j)$, for every $j\in\mathcal J$, including $j=0$.\footnote{When outside option is present and $x^\ast$ is unobserved, the first step must also identify
$H_{x^\ast\mid x,j}(\cdot\mid x,0)$ from nonprice information on
outside choosers. An outside share alone does not identify these
latent-type weights. One sufficient route is an appropriately
identified three-proxy model among outside choosers, with conditional
independence given $(x,x^\ast,j=0)$ and a latent-type normalization
common to the inside-choice models. For example, replacing
$p$ in the three-measurement factorization by an additional nonprice
proxy $z_3$ gives, among observations with $j=0$,
\[
\begin{split}
 &h_{z_3,z_1\mid z_2,x,j}(z_3,z_1\mid z_2,x,0)\\
 &\quad=\int_{X^\ast}
 h_{z_3\mid x^\ast,x,j}(z_3\mid x^\ast,x,0)
 h_{z_1\mid x^\ast,x,j}(z_1\mid x^\ast,x,0)
 \,dH_{x^\ast\mid z_2,x,j}(x^\ast\mid z_2,x,0).
\end{split}
\]
The corresponding rank or completeness and normalization conditions
must be verified. No price for alternative $0$ is used.
The available covariate and proxy data must support the conditional
weights used in the criterion, rather than merely an unconditional
outside count. With no latent covariate, these additional weights
are unnecessary.}
Its first-step estimate is obtained by
\[
\widehat H_{x^\ast\mid x,j}(A\mid x,j)
=\int
 \widehat H_{x^\ast\mid z_2,x,j}(A\mid z_2,x,j)
 \,d\widehat H_{z_2\mid x,j}(z_2\mid x,j).
\]

For the remainder of the estimation analysis, let
\[
 h=(G^{s},H_{x^\ast\mid x,j}),\qquad
 \hat h=(\widehat{G^{s}},\widehat H_{x^\ast\mid x,j})
\]
denote these two derived nuisance objects, and let $h_0$ denote their
true values. Other first-step estimates affect the second step through
$\hat h$. Identification of the proxy model does not by itself establish
the statistical properties of $\hat h$; the consistency, rate, and
influence-function conditions needed below are stated explicitly in
Section \ref{sec:asymptotics}.

\paragraph{Step 2: Estimating Selection Function Parameters and Offered Price Distribution}

Given $\hat h$, we use a weighted log-choice criterion and define
\begin{equation}\label{eq:estimator}
 \hat\theta\in\arg\max_{\theta\in\Theta}\hat Q_n(\theta),
\end{equation}
where
\begin{equation}\label{eq: estimator}
 \hat Q_n(\theta)
 =\frac1n\sum_{i=1}^n\int_{X^\ast}
 \log\pi_{j_i}(x_i,x^\ast;\theta,\widehat{G^{s}})
 \,d\widehat H_{x^\ast\mid x,j}(x^\ast\mid x_i,j_i),
\end{equation}
and
\begin{equation}\label{eq: prob j}
\begin{split}
 &\pi_j(x,x^\ast;\theta,G^{s})\\
 &\quad=\int_{\mathcal P} f_j(\boldsymbol p;x,x^\ast,\theta)
 \,d\bigl(F^{-1}(G^{s}(x,x^\ast);\theta,x,x^\ast)\bigr)(\boldsymbol p),
 \qquad j=0,\ldots,J.
\end{split}
\end{equation}
Integration with respect to $F^{-1}$ in \eqref{eq: prob j} means
integration with respect to the product of its recovered marginal
probability measures. We use $\pi_j$ for this integrated choice
probability, while $\operatorname{Pr}_j(p_j;G)$ continues to denote
the choice probability conditional on the alternative's own price
(with the conditioning characteristics and parameters suppressed).
For discrete $x^\ast$, \eqref{eq: estimator}
reduces to the corresponding weighted sum. This is a two-step
semiparametric $M$-estimator: the latent-type weights are estimated in
the first step and held fixed when optimizing over $\theta$.

Equation \eqref{eq: prob j} is defined for every $j=0,1,\ldots,J$,
but its product integral contains only the $J$ inside price marginals.
In particular,
\[
\begin{aligned}
 \pi_0(x,x^\ast;\theta,G^{s})
 =1-\sum_{j=1}^J\pi_j(x,x^\ast;\theta,G^{s}).
\end{aligned}
\]
The observations with $j_i=0$ contribute their outside log choice
probability to \eqref{eq: estimator}, using their own first-step
weights $\widehat H_{x^\ast\mid x,j}(\cdot\mid x_i,0)$. No selected CDF for alternative
$0$ is required. This complete probability vector is used in the
population identification and zero-mean score arguments below.

The inverse in \eqref{eq: prob j} is computed by the contraction
algorithm of Theorem \ref{thm: contraction}. Write
$\hat T_{\theta,\hat h}$ for the operator constructed with $\theta$ and
$\widehat{G^{s}}$, and let
$\hat T^{\infty}_{\theta,\hat h}\Psi$ denote its fixed point, obtained
from any initial probability distribution $\Psi$.\footnote{The
finite-iteration approximation and the numerical accuracy needed for
inference are discussed at the end of Section \ref{sec:asymptotics}
and proved in Appendix \ref{sec:proof_numerical}.}
In the inner loop, the fixed point determines the choice probabilities
in \eqref{eq: prob j}. The outer loop maximizes \eqref{eq: estimator}
over $\theta$. After estimating $\theta$, the offered price
distributions are recovered by
\[
 \hat G=\hat T^{\infty}_{\hat\theta,\hat h}\Psi
       =F^{-1}(\widehat{G^{s}};\hat\theta).
\]
The dependence on $(x,x^\ast)$ is suppressed in this last display.

\subsection{Consistency and Asymptotic Normality}
\label{sec:asymptotics}

We establish consistency and semiparametric asymptotic normality.
The main technical issue is the dependence of the criterion on the
functional inverse $F^{-1}$.

\begin{prop}\label{prop:homeomorphism}
Suppose $\rho<1$. For a fixed selection function, $F$ is a
homeomorphism in the extended Hilbert metric $D$. Moreover,
\[
\begin{split}
 D(F(G),F(G'))&\leq(1+\rho)D(G,G'),\\
 D(F^{-1}(G^{s}),F^{-1}(G^{s\prime}))
 &\leq\frac{D(G^{s},G^{s\prime})}{1-\rho}.
\end{split}
\]
The bounds are informative on the corresponding finite-distance
components.
\end{prop}
\begin{proof}
See Appendix \ref{sec:proof_3}.
\end{proof}

These inequalities concern the metric $D$, not the uniform distance
between CDFs. An empirical measure and a continuous population
measure need not be mutually absolutely continuous, so their Hilbert
distance can be infinite. For the statistical analysis we therefore
use CDF distances, retaining Hilbert contraction to control the
linearized fixed-point equation.

Probability measures are identified with their CDFs whenever a
function norm is used. For a vector $v$ of conditional CDFs or
CDF perturbations, write
\begin{equation}\label{eq:cdf_norm}
 \|v\|_\infty
 =\sup_{(x,x^\ast)\in X\times X^\ast}
   \max_{1\leq j\leq J}\sup_{p_j\in\mathcal P_j}
   |v_j(p_j\mid x,x^\ast)|.
\end{equation}
When $(x,x^\ast)$ is fixed and suppressed, the same notation denotes
only the maximum over $j$ and supremum over $p_j$.
The functions are defined on the maintained conditioning domain
$X\times X^\ast$; identification is interpreted under the actual
joint distribution of $(x,x^\ast)$, which need not have full product
support. For $h=(G^{s},H)$, with $H=H_{x^\ast\mid x,j}$, define
\begin{equation}\label{eq:nuisance_norm}
 \|h-h_0\|_{\mathcal H}
 =\|G^{s}-G^{s}_0\|_\infty
 +\sup_{x\in X,\,j\in\mathcal J}
   \|H(\cdot\mid x,j)-H_0(\cdot\mid x,j)\|_{\mathrm{TV}}.
\end{equation}
Here $\|\nu\|_{\mathrm{TV}}=
\sup_{|a|\leq1}|\int a\,d\nu|$ for a finite signed measure
$\nu$.\footnote{We work with fixed measurable versions of the
conditional distributions on the maintained conditioning domain.
Uniform consistency on this domain is a first-step requirement, not
an implication of almost-sure identification alone. When
measurability issues arise, $o_p$, $O_p$, and convergence in
probability are interpreted in outer probability; functional weak
convergence is understood in the outer-expectation sense.}
Total variation is imposed only on the latent-type mixing measure,
including the weights conditional on $j=0$, not on selected price
CDFs. For a density estimator of the mixing measure, this is an
$L^1$ condition. With no latent covariate, the mixing component is
omitted.

\begin{asmp}[Uniform regularity and first-step consistency]
\label{asmp:compact}
\leavevmode
\begin{enumerate}[(i)]
\item $\Theta\subset\mathbb R^{d_\theta}$ is compact and contains
$\theta_0$. The observations $\omega_i$ are i.i.d.
\item For $j\in\mathcal J$, each
$f_j(\boldsymbol p;x,x^\ast,\theta)$ is measurable in $(x,x^\ast)$
and continuous in $(\boldsymbol p,\theta)$. On the fixed compact
price rectangle $\mathcal P$ and an open neighborhood of $\Theta$,
the mixed derivatives
\[
 \partial_{\boldsymbol p}^{\alpha}\partial_\theta^{\beta}
 f_j(\boldsymbol p;x,x^\ast,\theta),
 \qquad \alpha\in\{0,1\}^{J},\quad |\beta|\leq2,
\]
exist, are uniformly continuous in $(\boldsymbol p,\theta)$
uniformly over $(x,x^\ast)$, and are uniformly bounded over all
indices and arguments. 
There exists $\varepsilon>0$ such that
\[
 f_j(\boldsymbol p;x,x^\ast,\theta)\geq\varepsilon
 \quad\text{for all }j\in\mathcal J,\
 (x,x^\ast)\in X\times X^\ast,\
 \boldsymbol p\in\mathcal P,\ \theta\in\Theta.
\]
\item For some $\bar\rho\in(0,1)$,
\[
 \sup_{(x,x^\ast)\in X\times X^\ast,\,\theta\in\Theta}
 \rho(x,x^\ast,\theta)\leq\bar\rho<1,
\]
where $\rho(x,x^\ast,\theta)$ is the bound in Theorem
\ref{thm: contraction}.
\item The first-step estimates are probability measures on the
maintained conditional supports and satisfy
$\|\hat h-h_0\|_{\mathcal H}=o_p(1)$.
\end{enumerate}
\end{asmp}

The price derivatives permit integration by parts in each coordinate
and provide a sufficient condition for CDF-space differentiation
without smoothness or parametric restrictions on offered price
densities. The bounds must hold uniformly over $(x,x^\ast)$ as well
as prices and parameters.\footnote{The first-step consistency
condition includes $H_{x^\ast\mid x,j}(\cdot\mid x,0)$ when the
outside option is available, and must be verified for the specified
proxy model and estimator.
If the outside option is absent ($f_0\equiv0$), omit $0$ from the
choice sums, the positivity condition, and the mixing-measure norm;
the same proofs then apply to the reduced choice set. In contrast,
when the outside option is present, its positive log-probability
contribution is included throughout.}  

Write $G_0$ and $G^{s}_0$ for the true collections of inside offered
and selected distributions. Define
\begin{equation}\label{eq:individual_criterion}
 \ell(\omega;\theta,h)
 =\int_{X^\ast}
 \log\pi_{j}(x,x^\ast;\theta,G^{s})
 \,dH(x^\ast\mid x,j).
\end{equation}
Then $\hat Q_n(\theta)=\mathbb P_n\ell(\theta,\hat h)$, where
$\mathbb P_n a=n^{-1}\sum_{i=1}^na(\omega_i)$. The population
criterion is
\begin{equation}\label{eq:population_criterion}
\begin{split}
 Q_0(\theta)
 &=\mathbb E[\ell(\omega;\theta,h_0)]\\
 &=\mathbb E_{x,x^\ast}\left[
  \sum_{k=0}^J\pi_k(x,x^\ast;\theta_0,G^{s}_0)
  \log\pi_k(x,x^\ast;\theta,G^{s}_0)
 \right].
\end{split}
\end{equation}
The last equality follows by iterated expectations using the true
latent-type distribution conditional on $(x,j)$, including $j=0$.

\begin{asmp}[Identification of the selection parameters]
\label{asmp: identification}
For $\theta\in\Theta$, if
\[
 \pi_j(x,x^\ast;\theta,G^{s}_0)
 =\pi_j(x,x^\ast;\theta_0,G^{s}_0)
 \quad\text{for all }j\in\mathcal J
\]
for almost every $(x,x^\ast)$ under its joint distribution, then
$\theta=\theta_0$.
\end{asmp}

Contraction identifies the inside offered distributions conditional
on any candidate selection parameter. Assumption
\ref{asmp: identification} requires that the resulting complete
choice-probability vector identify the remaining parameter
$\theta$. A difference at a single zero-probability conditioning
value is not sufficient for identification through this criterion.
This requirement is analogous to uniqueness conditions in structural
demand estimation: in \citet{berry1995automobile}, an inner fixed
point recovers mean utilities, and the remaining parameters are
identified by population moments. Here the inner fixed point
recovers offered price distributions. Additional instruments,
such as exogenous cost shifters, can provide further identifying
moments in applications.

Lemma \ref{lem: equicontinuity} in the appendix proves local Lipschitz continuity
and differentiability of the inverse in the CDF norm. Combining its
uniform bounds with \citet[Lemma 2.4 and Theorem 2.1]{newey1994large}
yields the following result.

\begin{thm}[Consistency]\label{thm: consistency}
Under Assumptions \ref{asmp:compact} and \ref{asmp: identification},
\[
 \hat\theta\xrightarrow{p}\theta_0,
 \qquad
 \|\hat G-G_0\|_\infty\xrightarrow{p}0,
 \qquad
 \hat G=\hat T^{\infty}_{\hat\theta,\hat h}\Psi.
\]
\end{thm}
\begin{proof}
See Appendix \ref{sec:proof_3}.
\end{proof}

For asymptotic normality, define the second-step score by
\begin{equation}\label{eq:semiparametric_score}
 \mathfrak g(\omega;\theta,h)=\nabla_\theta\ell(\omega;\theta,h),
 \qquad
 M_0=\mathbb E[\nabla_\theta\mathfrak g(\omega;\theta_0,h_0)].
\end{equation}
We use $D_h\mathfrak g(\omega;\theta_0,h_0)$ for its Fr\'echet
derivative in the nuisance norm \eqref{eq:nuisance_norm}.
Specifically, this is a bounded linear map on nuisance perturbations
such that, through the local extension used in the appendix,
\[
\begin{split}
 \mathfrak g(\omega;\theta_0,h_0+v)
 -\mathfrak g(\omega;\theta_0,h_0)
 -D_h\mathfrak g(\omega;\theta_0,h_0)[v]=o(\|v\|_{\mathcal H})
 \quad\text{as }\|v\|_{\mathcal H}\to0.
\end{split}
\]
Appendix \ref{sec:proof_4} establishes its existence and a uniform
quadratic remainder. In expressions involving $\hat h$, the notation
$\mathbb E_\omega$ means integration over a fresh observation with
the true law while holding the realized first-step estimate fixed.
It is not expectation over the training sample used to construct
$\hat h$.

\begin{asmp}[Semiparametric first-step conditions]
\label{asmp: normal}
\leavevmode
\begin{enumerate}[(i)]
\item $\theta_0\in\operatorname{int}(\Theta)$ and $M_0$ is nonsingular.
\item $\|\hat h-h_0\|_{\mathcal H}=o_p(n^{-1/4})$.
\item The linearized score satisfies
\[
\frac1{\sqrt n}\sum_{i=1}^n \left\{D_h\mathfrak g(\omega_i;\theta_0,h_0)[\hat h-h_0]-\mathbb E_\omega\!\left[D_h\mathfrak g(\omega;\theta_0,h_0)[\hat h-h_0]\right]\right\}=o_p(1).
\]
\item There is a measurable function
$\varphi(\omega)\in\mathbb R^{d_\theta}$ with
$\mathbb E\varphi(\omega)=0$ and
$\mathbb E\|\varphi(\omega)\|^2<\infty$ such that
\begin{equation}\label{eq:first_step_influence}
\begin{split}
 \sqrt n\,\mathbb E_\omega
 \bigl[D_h\mathfrak g(\omega;\theta_0,h_0)[\hat h-h_0]\bigr]=\frac1{\sqrt n}\sum_{i=1}^n\varphi(\omega_i)+o_p(1).
\end{split}
\end{equation}
\end{enumerate}
\end{asmp}

These conditions implement \citet[Theorems 8.1 and 8.2]{newey1994large}.
Part (ii) controls the quadratic remainder, part (iii) is stochastic
equicontinuity of the linearized score, and part (iv) describes the
first-order effect of nuisance estimation. Part (ii) alone does not
imply part (iv). The influence representation is a substantive
condition on the first-step estimator and this particular score
functional; it is not implied by contraction or proxy identification.
No $\sqrt n$ limit for the entire infinite-dimensional $\hat h$ is
imposed.

\begin{thm}[Asymptotic Normality]\label{thm: normal}
Suppose Assumptions \ref{asmp:compact}, \ref{asmp: identification},
and \ref{asmp: normal} hold. Then
\begin{equation}\label{eq:theta_asymptotic_linear}
\begin{split}
 \sqrt n(\hat\theta-\theta_0)
 &=-M_0^{-1}\frac1{\sqrt n}\sum_{i=1}^n
 \{\mathfrak g(\omega_i;\theta_0,h_0)+\varphi(\omega_i)\}
 +o_p(1)\\
 &\xrightarrow{d}\mathcal N(0,V),
\end{split}
\end{equation}
where
\begin{equation}\label{eq:semiparametric_variance}
\begin{split}
 V&=M_0^{-1}\Omega_0(M_0^{-1})',\\
 \Omega_0&=\mathbb E\!\left[
 \{\mathfrak g(\omega;\theta_0,h_0)+\varphi(\omega)\}
 \{\mathfrak g(\omega;\theta_0,h_0)+\varphi(\omega)\}'\right].
\end{split}
\end{equation}
\end{thm}
\begin{proof}
See Appendix \ref{sec:proof_4}.
\end{proof}

The influence function includes estimation of both selected CDFs and
latent-type weights, including outside-type weights and same-sample
dependence. Appendix \ref{sec:proof_4} gives a sandwich variance
estimator conditional on suitable estimates of this influence function.

The entire offered-price CDF requires a separate statement. Let
$\mathcal R(G^{s},\theta)=F^{-1}(G^{s};\theta)$ act on conditional
collections indexed by $(x,x^\ast)$. Write $D_{G^{s}}\mathcal R$ and
$D_\theta\mathcal R$ for its bounded linear derivatives at
$(G^{s}_0,\theta_0)$.

\begin{col}[Rates and functional limits for offered price distributions]
\label{col:cdf_asymptotics}
Under the conditions of Theorem \ref{thm: normal}:
\begin{enumerate}[(i)]
\item If $\|\widehat{G^{s}}-G^{s}_0\|_\infty=O_p(r_n)$ for
$r_n\to0$, then
\[
 \|\hat G-G_0\|_\infty=O_p(r_n+n^{-1/2}).
\]
\item If, in addition, the joint functional central limit theorem
\begin{equation}\label{eq:first_step_functional_clt}
 \bigl(\sqrt n(\widehat{G^{s}}-G^{s}_0),
       \sqrt n(\hat\theta-\theta_0)\bigr)
 \rightsquigarrow(\mathbb Z_{G^{s}},Z_\theta)
\end{equation}
holds in the product of the CDF supremum-norm space and
$\mathbb R^{d_\theta}$, with a tight, centered, jointly Gaussian
limit, then
\begin{equation}\label{eq:offered_functional_clt}
 \sqrt n(\hat G-G_0)
 \rightsquigarrow
 D_{G^{s}}\mathcal R[\mathbb Z_{G^{s}}]
 +D_\theta\mathcal R[Z_\theta].
\end{equation}
\end{enumerate}
The distributional conclusions also hold at a fixed $(x,x^\ast)$
when the corresponding first-step rate or joint functional limit is
imposed at that conditioning value, while the conditions for
$\hat\theta$ remain those of Theorem \ref{thm: normal}.
\end{col}
\begin{proof}
See Appendix \ref{sec:proof_4}.
\end{proof}

Continuous price support therefore does not preclude a $\sqrt n$
Gaussian limit for an estimated CDF. Such a limit is inherited from
a suitable first-step functional limit, not from Theorem
\ref{thm: normal} alone. Conditional distribution estimators or
latent-variable first steps may instead converge at a slower rate.

\paragraph{Finite-iteration approximation.}
Let $\hat Q_{n,m}$ use $\hat T^m_{\theta,\hat h}\Psi$ in place of the
exact inverse, with initialization independent of $\theta$, and let
$\hat\theta_m$ maximize this criterion. Appendix
\ref{sec:proof_numerical} establishes
\[
 \sup_{\theta\in\Theta}|\hat Q_{n,m}(\theta)-\hat Q_n(\theta)|
 \leq C\bar\rho^{m-1},
\]
and, on a neighborhood $\mathscr N$ of $\theta_0$,
\[
 \sup_{\theta\in\mathscr N}
 \|\nabla_\theta\hat Q_{n,m}(\theta)-\nabla_\theta\hat Q_n(\theta)\|
 \leq C m\bar\rho^{m-1}.
\]
Thus $m(n)\to\infty$ suffices for consistency. The limits in Theorem
\ref{thm: normal} and Corollary \ref{col:cdf_asymptotics}(ii) are
unchanged if
\begin{equation}\label{eq:numerical_rate}
 \sqrt n\,m(n)\bar\rho^{m(n)-1}\longrightarrow0.
\end{equation}
A sufficient condition is
\[
 \liminf_{n\to\infty}\frac{m(n)}{\log n}
 >\frac1{2\log(1/\bar\rho)}.
\]
Numerical optimization and any additional integration approximation
must also satisfy the stated criterion or score accuracy requirements.

\section{Monte Carlo Simulations}
\label{sec:monte_carlo}

To examine how our estimators for $\theta$ and the offered price distribution $G$ perform in finite samples, we conduct the original Monte Carlo experiment with two price-bearing
alternatives ($J=2$) and no outside option ($f_0\equiv0$). Alternative
$0$ and its log-probability contribution are omitted in this special
case; the first-step conditions are imposed only on choices 1 and 2.
The contraction condition $\rho<1$ is automatically satisfied.
The utilities of individual $i$ from the two alternatives are specified as follows:
\begin{align*}
    & u_{i1}=-\gamma \ln(p_{i1}) + \xi_1 + \beta x_{i1} + \kappa x_i^\ast+\varepsilon_i,\\
    & u_{i2}=-\gamma \ln(p_{i2}) + \xi_2,
\end{align*}
where $p_{ij}$ and $\xi_j$ are, respectively, the offered price and unobserved heterogeneity for alternative $j$; $x_{i1} \in \{0, 1\}$ is a binary observable with $Pr(x_{i1}=1)=0.5$ that shifts individual $i$'s choice probabilities; 
$x_i^\ast \in \{-1, 1\}$ is a binary unobservable with $Pr(x_{i}^\ast=1)=0.5$; 
and $\varepsilon_i \sim N(0,1)$ is the error term. 
Throughout the main simulation exercises, we set the utility parameters as follows: $\gamma=1$, $\xi_1=0$, $\xi_2=0.5$, $\beta=0.5$ and $\kappa=0.1$.\footnote{Note that $x_{i1}$ is included in the utility specification to facilitate the implementation of the two-step method for sample selection. Our method does not require this type of conventional excluded variable to exogenously shift the selection probability. We therefore also consider a simulation exercise in which $x_{i1}$ is omitted, that is, $\beta=0$. The results are reported in Tables \ref{table:est_theta_noz} and \ref{table:est_price_noz} in Appendix \ref{sec:tables_app}.
As shown, our estimator performs well in finite samples without an excluded variable in the selection model.} Let $j_i \in \{1, 2\}$ denote the choice of individual $i$.

We consider four data generating processes for the offered prices. Let $x_{i2}$ denote the observable characteristic of individual $i$ that enters the pricing equation. 
We assume that $x_{i2}$ takes values in $\{0, 0.25, 0.5, 0.75, 1\}$ with equal probability.

\begin{enumerate}[DGP 1:]
    \item 
    $\ln(p_{ij}) = \delta_{0j}+\delta_{1j}x_{i2}+\delta_{2j}x_i^\ast+\eta_{ij}$, where $\eta_{ij} \sim N(0, \sigma_j^2)$.
    For alternative 1, we set $\delta_{01}=0.2, \delta_{11}=0.5, \delta_{21}=0.1, \sigma_1=0.1$. For alternative 2, we set $\delta_{02}=0.1, \delta_{12}=1,\delta_{22}=0.1, \sigma_2=0.2$.
\item 
$\ln(p_{ij}) = \delta_{0j}+\delta_{1j}x_{i2}^2+\delta_{2j}x_i^\ast+\eta_{ij}$, where $\eta_{ij} \sim N(0, \sigma_j^2)$.
For alternative 1, we set $\delta_{01}=0.2, \delta_{11}=0.5, \delta_{21}=0.1, \sigma_1=0.1$. For alternative 2, we set $\delta_{02}=0.1, \delta_{12}=1,\delta_{22}=0.1, \sigma_2=0.2$. 
\item 
$\ln(p_{ij}) = \exp\left( (\delta_{0j}+\delta_{1j}x_{i2})(\delta_{2j}x_i^\ast+\eta_{ij})\right)$, where $\eta_{ij} \sim N(1, \sigma_j^2)$.
For alternative 1, we set $\delta_{01}=0.2, \delta_{11}=0.3, \delta_{21}=0.1, \sigma_1=0.1$. For alternative 2, we set $\delta_{02}=0.1, \delta_{12}=0.5,\delta_{22}=0.1, \sigma_2=0.2$. 
\item 
$\ln(p_{ij}) =  (\delta_{0j}+\delta_{1j}x_{i2}^2)(\delta_{2j}x_i^\ast+\eta_{ij})^{-1}$, where $\eta_{ij} \sim N(-2, \sigma_j^2)$.
For alternative 1, we set $\delta_{01}=0.2, \delta_{11}=0.1, \delta_{21}=0.1, \sigma_1=0.1$. For alternative 2, we set $\delta_{02}=0.1, \delta_{12}=0.3,\delta_{22}=0.1, \sigma_2=0.2$. 
\end{enumerate}

Across all data generating processes, the unobserved characteristic $x_i^\ast$ enters the pricing equations for both alternatives, which induces correlation in prices conditional on observables. In addition, $x_i^\ast$ also enters the utility specification, allowing the unobserved type to jointly affect the prices individuals face and their preferences over alternatives.
DGP 1 specifies an additively separable linear pricing equation, which is commonly assumed in empirical applications. DGP 2 introduces a nonlinear term.
DGPs 3 and 4 consider scenarios where the pricing function takes a nonseparable form.\footnote{Although the offered price distributions have unbounded support, in the implementation
we approximate the support by the realized price range. These numerical
experiments illustrate finite-sample performance for unbounded DGPs;
the fixed-compact-support asymptotic results do not, by themselves,
justify this expanding-range approximation.}

For each DGP, we simulate offered prices and individual choices. To implement our estimator, we require an instrument $z_i$ to recover the selected price distribution conditional on $(x_{i1},x_{i2},x_i^\ast)$ in the first step, since $x_i^\ast$ is unobserved. We construct such an instrument by assuming $z_i \sim Poisson(x_i^\ast)$ when $x_{i}^\ast=1$, and $z_i=0$ otherwise. This choice is motivated by settings where $x_i^\ast$ can be interpreted as an individual’s unobserved risk type, and such risk types may be reflected in the ex post realization of accidents, which are often modeled using a Poisson distribution. Because we impose a parametric relationship between the instrument and the unobservable, only one instrument is needed.

We assume that the econometrician observes $(j_i, x_{i1}, x_{i2}, p_i, z_i)$, where $p_i$ denotes the price of the chosen alternative. 
Using these data, we apply the procedure described in Section \ref{sec:estimation_strategy} to estimate the parameters of the selection function, $\theta=(\gamma, \xi_2, \beta, \kappa)$ with $\xi_1$ normalized to 0, along with the offered price distribution for each alternative.\footnote{We estimate the cumulative distribution function of prices at 300 grid points.} 
For comparison, we first implement the classic Heckman parametric two-step method, assuming that the pricing equations are linearly separable and that the error terms in the selection and pricing equations follow a bivariate normal distribution.
We also compare our estimator with the quantile selection model of \cite{arellano2017quantile}. To implement their approach, we follow the standard practice of assuming that the quantile functions are linear in $x_{i2}$ and that the dependence structure is governed by a Gaussian copula.\footnote{Although \cite{arellano2017quantile} discuss identification under more general settings, their empirical implementation focuses on cases in which the copula depends on a low-dimensional vector of parameters, which is the specification we adopt here.}
For each design, we run 500 simulations with sample sizes of 2,000 and 5,000 observations.

Table \ref{tab:est_theta} reports the Monte Carlo biases, standard deviations, and root mean squared errors for the estimates of $\theta$ obtained using our method.
Overall, the estimator performs well in finite samples across all DGPs, including those with nonseparable pricing equations. The biases are small, and the root mean squared errors remain modest for all parameters in the selection function. The standard deviation decreases as the sample size increases in all simulation designs.

\begin{table}[htbp!]
    \centering
    \caption{Simulation Results for Utility Parameters}
    \label{tab:est_theta}    \scalebox{0.9}{\begin{tabular}{lccc|ccc}
\hline \hline 
 & \multicolumn{6}{c}{DGP 1} \\ \cline{2-7} 
  & \multicolumn{3}{c|}{$N=2000$}  & \multicolumn{3}{c}{$N=5000$}  \\ \cline{2-7} 
  & Bias & Std. Dev. & RMSE & Bias & Std. Dev. & RMSE \\ \cline{2-7} 
 $\gamma$ & -0.1017 & 0.1762 & 0.2033 & -0.0697 & 0.1087 & 0.1290 \\ 
$\beta$ & -0.0058 & 0.0619 & 0.0621 & -0.0013 & 0.0372 & 0.0371 \\ 
$\kappa$ & -0.0224 & 0.0491 & 0.0540 & -0.0095 & 0.0350 & 0.0362 \\ 
$\xi_2$ & -0.0267 & 0.0570 & 0.0629 & -0.0156 & 0.0349 & 0.0382 \\ \hline 
 & \multicolumn{6}{c}{DGP 2} \\ \cline{2-7} 
  & \multicolumn{3}{c|}{$N=2000$}  & \multicolumn{3}{c}{$N=5000$}  \\ \cline{2-7} 
  & Bias & Std. Dev. & RMSE & Bias & Std. Dev. & RMSE \\ \cline{2-7} 
 $\gamma$ & -0.1003 & 0.1672 & 0.1949 & -0.0705 & 0.1067 & 0.1278 \\ 
$\beta$ & -0.0049 & 0.0626 & 0.0627 & -0.0013 & 0.0379 & 0.0379 \\ 
$\kappa$ & -0.0185 & 0.0487 & 0.0521 & -0.0072 & 0.0344 & 0.0351 \\ 
$\xi_2$ & -0.0197 & 0.0508 & 0.0544 & -0.0114 & 0.0319 & 0.0338 \\ \hline 
 & \multicolumn{6}{c}{DGP 3} \\ \cline{2-7} 
  & \multicolumn{3}{c|}{$N=2000$}  & \multicolumn{3}{c}{$N=5000$}  \\ \cline{2-7} 
  & Bias & Std. Dev. & RMSE & Bias & Std. Dev. & RMSE \\ \cline{2-7} 
 $\gamma$ & -0.1092 & 0.2557 & 0.2778 & -0.0685 & 0.1535 & 0.1680 \\ 
$\beta$ & -0.0011 & 0.0639 & 0.0639 & 0.0015 & 0.0367 & 0.0367 \\ 
$\kappa$ & -0.0194 & 0.0485 & 0.0522 & -0.0076 & 0.0330 & 0.0339 \\ 
$\xi_2$ & -0.0074 & 0.0458 & 0.0463 & -0.0032 & 0.0278 & 0.0279 \\ \hline 
 & \multicolumn{6}{c}{DGP 4} \\ \cline{2-7} 
  & \multicolumn{3}{c|}{$N=2000$}  & \multicolumn{3}{c}{$N=5000$}  \\ \cline{2-7} 
  & Bias & Std. Dev. & RMSE & Bias & Std. Dev. & RMSE \\ \cline{2-7} 
 $\gamma$ & 0.0934 & 0.8041 & 0.8087 & 0.0666 & 0.4806 & 0.4847 \\ 
$\beta$ & 0.0005 & 0.0613 & 0.0613 & 0.0038 & 0.0361 & 0.0363 \\ 
$\kappa$ & -0.0194 & 0.0474 & 0.0512 & -0.0103 & 0.0335 & 0.0350 \\ 
$\xi_2$ & -0.0003 & 0.0451 & 0.0451 & 0.0016 & 0.0280 & 0.0280 \\ \hline 
 \hline 
\end{tabular}}

\end{table}

For the cumulative distribution functions of $log(price)$, Tables \ref{tab:est_price1} and \ref{tab:est_price2} report the integrated squared biases and integrated mean squared errors for our proposed estimator, the Heckman two-step estimator, and the copula-based sample-selection correction estimator for quantile regression, separately for the two alternatives. Each row of the tables corresponds to the price distribution conditional on a specific value of $x_{i2}$. 
We also plot the true CDFs for alternatives 1 and 2 alongside the estimates produced by these models conditional on $x_{i2}=0.25$ and $x_{i2}=0.75$ in Figures \ref{fig:compare_cdf_x2} and \ref{fig:compare_cdf_x4}, respectively.
To save space, we report the CDF results only for the sample size of 2,000 observations, and we omit the figures for other values of the observable covariates.

Our method allows for nonparametric estimation of the offered price distributions, whereas the alternative approaches impose parametric restrictions on either the conditional mean or quantiles of the pricing distributions, or on the dependence structure through the copula.
Tables \ref{tab:est_price1} and \ref{tab:est_price2} show that our estimator achieves very low integrated squared bias and integrated mean squared error for the CDFs of $log(price)$ across all simulation designs and for all values of $x_{i2}$.  
In contrast, while the classic Heckman two-step method and the quantile selection model perform well in DGP 1, their biases and mean squared errors increase substantially as the pricing equation becomes more complex in DGPs 2–4. 
These results are expected, since the parametric assumptions underlying these methods, such as linear conditional mean or quantile functions and a Gaussian copula, are severely violated in these designs.

Figures \ref{fig:compare_cdf_x2} and \ref{fig:compare_cdf_x4} provide a visual illustration of these results.
We can see that across all simulation designs, the estimated CDFs of $log(price)$ for both alternatives produced by our functional contraction approach closely track the true CDFs, as indicated by the black curves with ``+'' markers and the red solid curve in Figures \ref{fig:compare_cdf_x2} and \ref{fig:compare_cdf_x4}. By comparison, the biases of the Heckman two-step method (blue dashed curves) and the quantile selection model (purple dash–dotted curves) can be substantial, particularly in DGPs 3 and 4. The direction and magnitude of these biases also vary with the values of the observable covariates.

\begin{table}[htpb!]
    \centering
     \caption{Simulation Results for CDF of $log(p_1)$}
     \label{tab:est_price1}\scalebox{0.9}{\begin{tabular}{lcc|cc|cc}
\hline \hline 
 & \multicolumn{6}{c}{DGP 1} \\ \cline{2-7} 
  & \multicolumn{2}{c|}{Functional Contraction}  & \multicolumn{2}{c|}{Heckman Two-Step}  & \multicolumn{2}{c}{Quantile Selection}  \\ \cline{2-7} 
  & IBias$^2$ & IMSE & IBias$^2$ & IMSE & IBias$^2$ & IMSE \\ \cline{2-7} 
 $x_{i2}=0$ & 0.0005 & 0.0017 & 0.0001 & 0.0039 & 0.0001 & 0.0046 \\ 
$x_{i2}=0.25$ & 0.0004 & 0.0015 & 0.0001 & 0.0033 & 0.0001 & 0.0038 \\ 
$x_{i2}=0.5$ & 0.0002 & 0.0012 & 0.0001 & 0.0027 & 0.0001 & 0.0031 \\ 
$x_{i2}=0.75$ & 0.0002 & 0.0010 & 0.0001 & 0.0023 & 0.0001 & 0.0026 \\ 
$x_{i2}=1$ & 0.0001 & 0.0010 & 0.0001 & 0.0021 & 0.0001 & 0.0023 \\ \hline 
 & \multicolumn{6}{c}{DGP 2} \\ \cline{2-7} 
  & \multicolumn{2}{c|}{Functional Contraction}  & \multicolumn{2}{c|}{Heckman Two-Step}  & \multicolumn{2}{c}{Quantile Selection}  \\ \cline{2-7} 
  & IBias$^2$ & IMSE & IBias$^2$ & IMSE & IBias$^2$ & IMSE \\ \cline{2-7} 
 $x_{i2}=0$ & 0.0005 & 0.0017 & 0.0235 & 0.0275 & 0.0185 & 0.0231 \\ 
$x_{i2}=0.25$ & 0.0005 & 0.0017 & 0.0017 & 0.0053 & 0.0031 & 0.0072 \\ 
$x_{i2}=0.5$ & 0.0003 & 0.0014 & 0.0124 & 0.0153 & 0.0145 & 0.0177 \\ 
$x_{i2}=0.75$ & 0.0002 & 0.0011 & 0.0037 & 0.0062 & 0.0052 & 0.0080 \\ 
$x_{i2}=1$ & 0.0001 & 0.0010 & 0.0133 & 0.0154 & 0.0106 & 0.0129 \\ \hline 
 & \multicolumn{6}{c}{DGP 3} \\ \cline{2-7} 
  & \multicolumn{2}{c|}{Functional Contraction}  & \multicolumn{2}{c|}{Heckman Two-Step}  & \multicolumn{2}{c}{Quantile Selection}  \\ \cline{2-7} 
  & IBias$^2$ & IMSE & IBias$^2$ & IMSE & IBias$^2$ & IMSE \\ \cline{2-7} 
 $x_{i2}=0$ & 0.0007 & 0.0019 & 0.0207 & 0.0259 & 0.0017 & 0.0071 \\ 
$x_{i2}=0.25$ & 0.0002 & 0.0013 & 0.0061 & 0.0102 & 0.0011 & 0.0052 \\ 
$x_{i2}=0.5$ & 0.0002 & 0.0013 & 0.0017 & 0.0049 & 0.0013 & 0.0049 \\ 
$x_{i2}=0.75$ & 0.0002 & 0.0012 & 0.0016 & 0.0041 & 0.0002 & 0.0037 \\ 
$x_{i2}=1$ & 0.0002 & 0.0013 & 0.0053 & 0.0074 & 0.0007 & 0.0040 \\ \hline 
 & \multicolumn{6}{c}{DGP 4} \\ \cline{2-7} 
  & \multicolumn{2}{c|}{Functional Contraction}  & \multicolumn{2}{c|}{Heckman Two-Step}  & \multicolumn{2}{c}{Quantile Selection}  \\ \cline{2-7} 
  & IBias$^2$ & IMSE & IBias$^2$ & IMSE & IBias$^2$ & IMSE \\ \cline{2-7} 
 $x_{i2}=0$ & 0.0014 & 0.0023 & 0.0432 & 0.0463 & 0.0390 & 0.0425 \\ 
$x_{i2}=0.25$ & 0.0014 & 0.0024 & 0.0147 & 0.0182 & 0.0153 & 0.0190 \\ 
$x_{i2}=0.5$ & 0.0011 & 0.0021 & 0.0412 & 0.0443 & 0.0358 & 0.0392 \\ 
$x_{i2}=0.75$ & 0.0012 & 0.0021 & 0.0106 & 0.0141 & 0.0059 & 0.0100 \\ 
$x_{i2}=1$ & 0.0005 & 0.0018 & 0.0270 & 0.0304 & 0.0346 & 0.0384 \\ \hline 
 \hline 
\end{tabular}}

\vspace{0.2cm} \\
\justifying \footnotesize 
\noindent Note: 
The $\text{IBias}^2$ of a function $h$ is calculated as follows. Let $\hat{h}_r$ be the estimate of $h$ from the $r$-th simulated dataset, and $\bar{h}(p) = \frac{1}{R}\sum_{r=1}^{R} \hat{h}_r(p)$ be the point-wise average over $R$ simulations. The integrated squared bias is calculated by numerically integrating the point-wise squared bias $(\bar{h}(p)-h(p))^2$ over the distribution of $p$. The integrated MSE is computed in a similar way. The values reported in each row correspond to the price distributions conditional on a given value of $x_{i2}$. The results shown in this table are based on 500 Monte Carlo replications with a sample size of 2,000. Corresponding results for a sample size of 5,000 are available upon request.   
\end{table}

\begin{table}[htpb!]
    \centering
     \caption{Simulation Results for CDF of $log(p_2)$}
     \label{tab:est_price2}\scalebox{0.9}{\begin{tabular}{lcc|cc|cc}
\hline \hline 
 & \multicolumn{6}{c}{DGP 1} \\ \cline{2-7} 
  & \multicolumn{2}{c|}{Functional Contraction}  & \multicolumn{2}{c|}{Heckman Two-Step}  & \multicolumn{2}{c}{Quantile Selection}  \\ \cline{2-7} 
  & IBias$^2$ & IMSE & IBias$^2$ & IMSE & IBias$^2$ & IMSE \\ \cline{2-7} 
 $x_{i2}=0$ & 0.0002 & 0.0008 & 0.0000 & 0.0016 & 0.0001 & 0.0018 \\ 
$x_{i2}=0.25$ & 0.0002 & 0.0009 & 0.0000 & 0.0019 & 0.0001 & 0.0020 \\ 
$x_{i2}=0.5$ & 0.0002 & 0.0010 & 0.0000 & 0.0023 & 0.0001 & 0.0024 \\ 
$x_{i2}=0.75$ & 0.0002 & 0.0011 & 0.0000 & 0.0028 & 0.0001 & 0.0030 \\ 
$x_{i2}=1$ & 0.0002 & 0.0012 & 0.0000 & 0.0034 & 0.0001 & 0.0038 \\ \hline 
 & \multicolumn{6}{c}{DGP 2} \\ \cline{2-7} 
  & \multicolumn{2}{c|}{Functional Contraction}  & \multicolumn{2}{c|}{Heckman Two-Step}  & \multicolumn{2}{c}{Quantile Selection}  \\ \cline{2-7} 
  & IBias$^2$ & IMSE & IBias$^2$ & IMSE & IBias$^2$ & IMSE \\ \cline{2-7} 
 $x_{i2}=0$ & 0.0002 & 0.0008 & 0.0231 & 0.0246 & 0.0211 & 0.0228 \\ 
$x_{i2}=0.25$ & 0.0002 & 0.0008 & 0.0058 & 0.0075 & 0.0067 & 0.0086 \\ 
$x_{i2}=0.5$ & 0.0002 & 0.0009 & 0.0200 & 0.0220 & 0.0213 & 0.0234 \\ 
$x_{i2}=0.75$ & 0.0002 & 0.0010 & 0.0030 & 0.0058 & 0.0040 & 0.0069 \\ 
$x_{i2}=1$ & 0.0003 & 0.0011 & 0.0368 & 0.0401 & 0.0329 & 0.0365 \\ \hline 
 & \multicolumn{6}{c}{DGP 3} \\ \cline{2-7} 
  & \multicolumn{2}{c|}{Functional Contraction}  & \multicolumn{2}{c|}{Heckman Two-Step}  & \multicolumn{2}{c}{Quantile Selection}  \\ \cline{2-7} 
  & IBias$^2$ & IMSE & IBias$^2$ & IMSE & IBias$^2$ & IMSE \\ \cline{2-7} 
 $x_{i2}=0$ & 0.0031 & 0.0036 & 0.0492 & 0.0515 & 0.0030 & 0.0047 \\ 
$x_{i2}=0.25$ & 0.0005 & 0.0011 & 0.0190 & 0.0212 & 0.0052 & 0.0072 \\ 
$x_{i2}=0.5$ & 0.0003 & 0.0010 & 0.0047 & 0.0068 & 0.0019 & 0.0041 \\ 
$x_{i2}=0.75$ & 0.0002 & 0.0009 & 0.0014 & 0.0035 & 0.0002 & 0.0026 \\ 
$x_{i2}=1$ & 0.0002 & 0.0011 & 0.0112 & 0.0131 & 0.0049 & 0.0072 \\ \hline 
 & \multicolumn{6}{c}{DGP 4} \\ \cline{2-7} 
  & \multicolumn{2}{c|}{Functional Contraction}  & \multicolumn{2}{c|}{Heckman Two-Step}  & \multicolumn{2}{c}{Quantile Selection}  \\ \cline{2-7} 
  & IBias$^2$ & IMSE & IBias$^2$ & IMSE & IBias$^2$ & IMSE \\ \cline{2-7} 
 $x_{i2}=0$ & 0.0011 & 0.0016 & 0.1443 & 0.1459 & 0.1123 & 0.1141 \\ 
$x_{i2}=0.25$ & 0.0009 & 0.0015 & 0.0568 & 0.0592 & 0.0455 & 0.0483 \\ 
$x_{i2}=0.5$ & 0.0005 & 0.0011 & 0.1095 & 0.1109 & 0.0952 & 0.0966 \\ 
$x_{i2}=0.75$ & 0.0003 & 0.0009 & 0.0290 & 0.0308 & 0.0133 & 0.0155 \\ 
$x_{i2}=1$ & 0.0002 & 0.0007 & 0.0530 & 0.0546 & 0.0722 & 0.0741 \\ \hline 
 \hline 
\end{tabular}}

\vspace{0.2cm} \\
\justifying \footnotesize 
\noindent Note: 
The $\text{IBias}^2$ of a function $h$ is calculated as follows. Let $\hat{h}_r$ be the estimate of $h$ from the $r$-th simulated dataset, and $\bar{h}(p) = \frac{1}{R}\sum_{r=1}^{R} \hat{h}_r(p)$ be the point-wise average over $R$ simulations. The integrated squared bias is calculated by numerically integrating the point-wise squared bias $(\bar{h}(p)-h(p))^2$ over the distribution of $p$. The integrated MSE is computed in a similar way. The values reported in each row correspond to the price distributions conditional on a given value of $x_{i2}$. The results shown in this table are based on a 500 Monte Carlo replications with a sample size of 2,000. Corresponding results for a sample size of 5,000 are available upon request.
\end{table}

\begin{figure}[htbp!]
	\centering
 \includegraphics[width=0.9\linewidth]{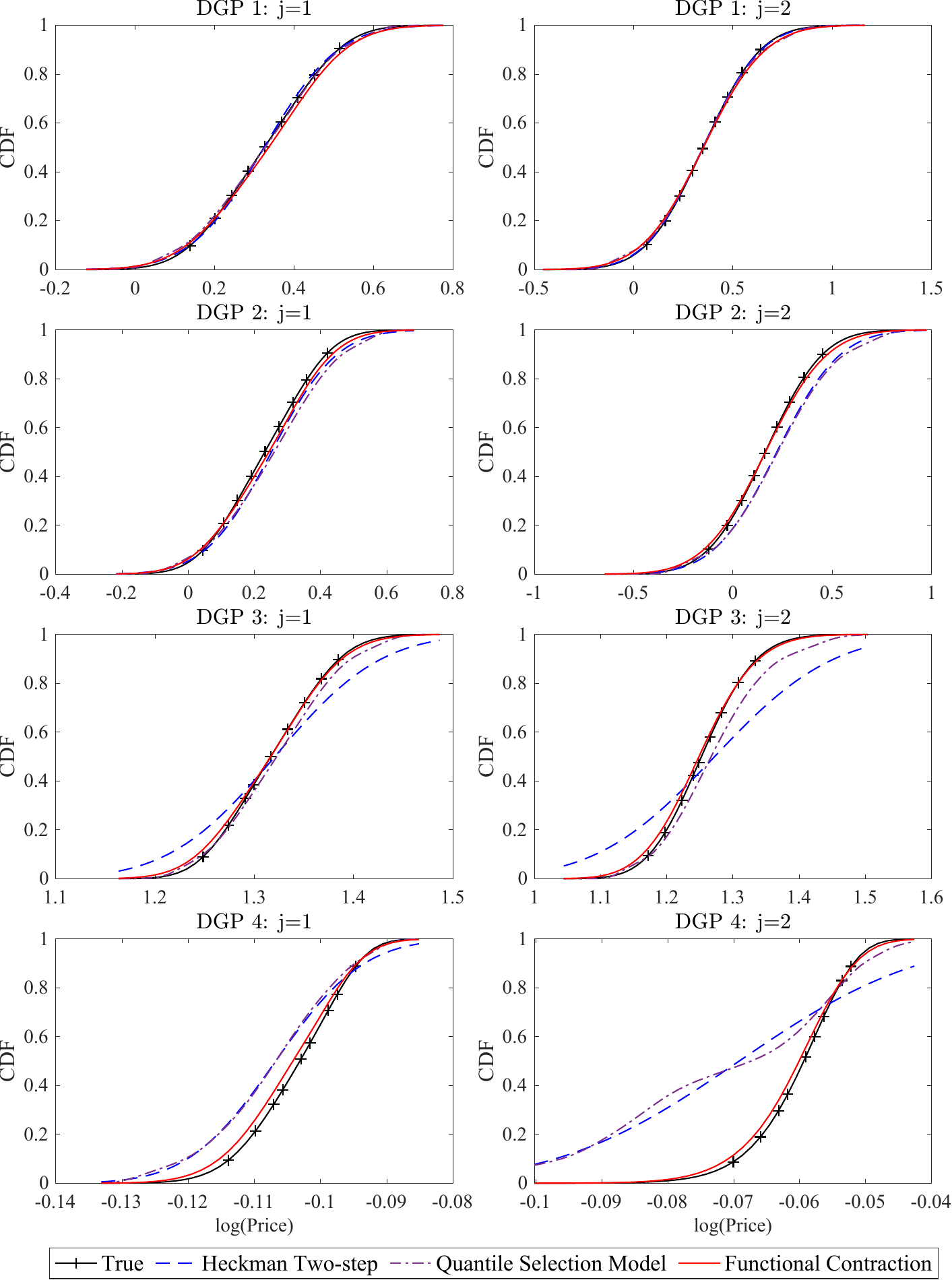}
 \caption{CDFs of $log(price)$ for alternatives 1 and 2, conditional on $x_{i2}=0.25$. The black curve with “+” markers, the blue dashed curve, the purple dash–dotted curve, and the red solid curve correspond to the true CDF, the Heckman two-step estimate, the quantile selection estimate, and the functional contraction estimate, respectively, based on 500 Monte Carlo replications with a sample size of 2,000.
 }
	\label{fig:compare_cdf_x2}
\end{figure}

\begin{figure}[htbp!]
	\centering
 \includegraphics[width=0.9\linewidth]{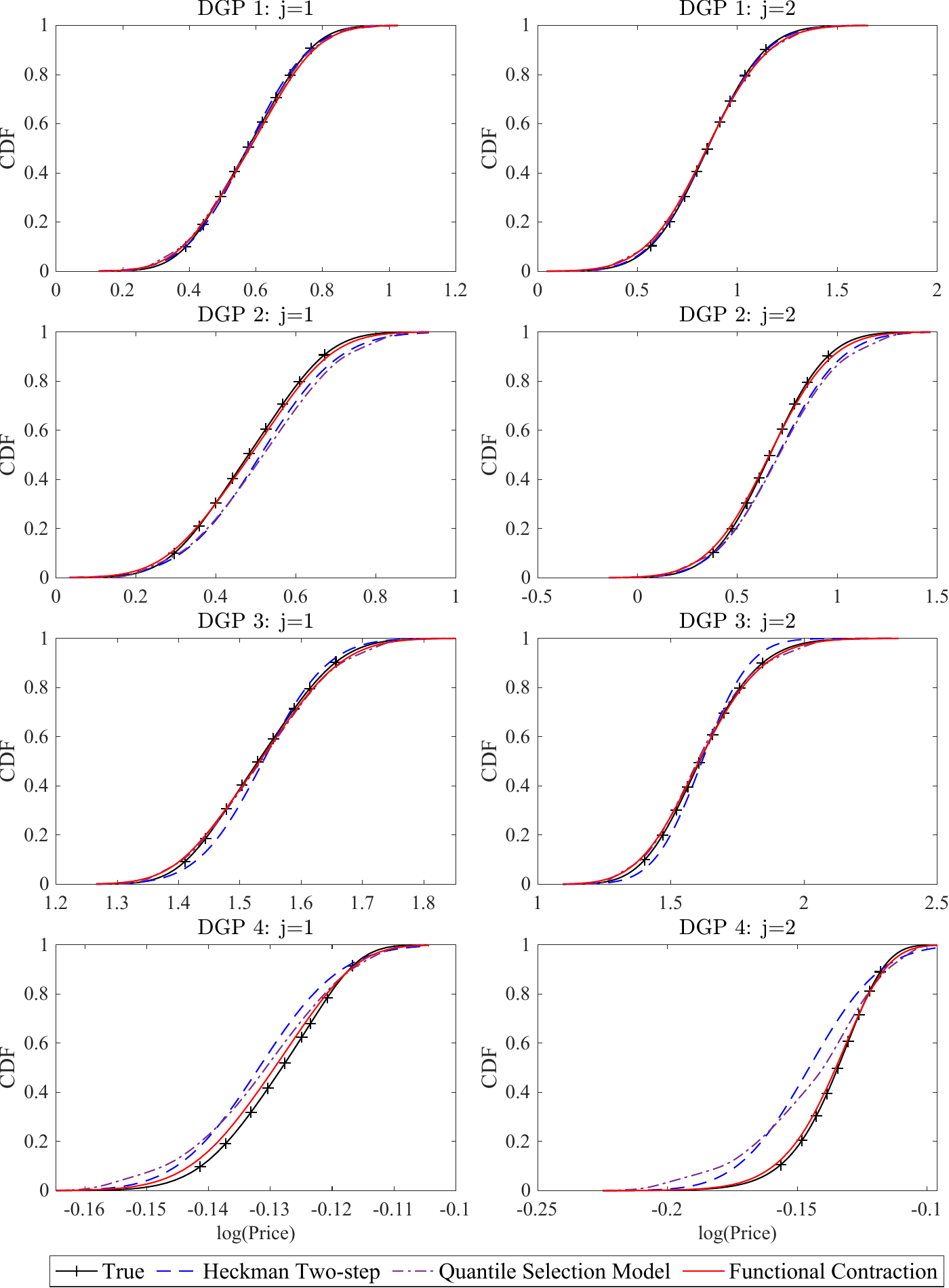}
 \caption{CDFs of $log(price)$ for alternatives 1 and 2, conditional on $x_{i2}=0.75$. The black curve with “+” markers, the blue dashed curve, the purple dash–dotted curve, and the red solid curve correspond to the true CDF, the Heckman two-step estimate, the quantile selection estimate, and the functional contraction estimate, respectively, based on 500 Monte Carlo replications with a sample size of 2,000.
 }
	\label{fig:compare_cdf_x4}
\end{figure}

Our method requires the econometrician to correctly specify the functional form of the selection function. To evaluate how the estimator performs under misspecification, we conduct a series of Monte Carlo simulations in which the econometrician assumes that 
$\varepsilon$ follows a logistic distribution, while in truth it is generated from a normal distribution. In Tables \ref{table:est_theta_mis}--\ref{table:est_price_mis} in Appendix \ref{sec:tables_app}, we report the estimation results for the utility parameters and CDFs of $log(price)$ under this misspecification.
For the utility parameters, we rescale the estimates by the scale parameter of the logit model to make them comparable to those in the original probit specification. After this adjustment, the biases are small. The estimator for the offered price distributions also performs well: the integrated squared biases and mean squared errors of the CDFs remain close to those in Tables \ref{tab:est_price1} and \ref{tab:est_price2}. These results suggest that the estimator of the offered price distributions may be relatively insensitive to moderate misspecification of the selection function. This property could be useful in practice, particularly when limited prior information is available regarding the appropriate functional form.

Finally, we briefly discuss how the functional contraction performs computationally in practice.
The average numbers of iterations needed to reach convergence (with a tolerance of $10^{-5}$) are 3.8, 3.8, 3.4, and 2.1 for DGPs 1 through 4, respectively (averaged over 500 replications).
These results indicate that the proposed estimator is computationally efficient, converges rapidly, and remains stable across a range of data generating processes, making it attractive for applied work.

\section{Discussion of Empirical Applications}
\label{sec:applications}

The estimator introduced in Section \ref{sec:estimation} is broadly applicable to a wide range of empirical settings. It offers an alternative approach to addressing the challenge of selection bias that arises when only the outcomes of chosen alternatives are observed. 
The method has several features that can be useful in empirical applications. First, it imposes no parametric or separability restrictions on the potential outcome distributions and allows them to vary flexibly across alternatives. Second, the framework accommodates unobservable characteristics in both the outcome distributions and the selection model, capturing selection on unobservables. Third, the selection function can incorporate alternative-specific unobserved heterogeneity and does not require an excluded variable to exogenously shift the choice probability. Such an excluded variable may not be readily available in certain empirical settings.

A natural application of the proposed method is consumer demand estimation in markets where only transaction prices are observed.
In classic differentiated product demand estimation pioneered by \cite{berry1994estimating} and \cite{berry1995automobile}, the price of a product is often assumed to be uniform across all consumers (e.g., the list price of a vehicle). But this assumption does not hold in contexts involving price discrimination or personalized pricing 
\citep{d2019automobile,
sagl2023dispersion, buchholz2020value, dube2023personalized}, discount negotiation \citep{goldberg1996dealer, allen2014price}, or risk-based pricing 
\citep{crawford2018asymmetric, cosconati2024competing}.
In these contexts, researchers can relatively easily gather data on the transaction prices consumers pay, but it is challenging to gain access to competing prices offered to consumers. 

In \cite{cosconati2024competing}, we apply the method to estimate demand and insurance companies' information technology in the auto insurance market, where only the transaction prices of selected insurance plans are observed. In this market, insurance companies employ risk-based pricing. For each consumer, an insurance company generates a noisy estimate of their risk type and prices accordingly. 
Our goal is to quantify the heterogeneity in insurers' information technology, as measured by the dispersion of their risk estimates. Since the shape of the offered price distribution reflects the distribution of risk estimates, allowing for flexible estimation of the offered price distribution is crucial.

In this application, we assume that the offered prices across different firms are independent conditional on observable characteristics and the consumer’s true unobserved risk type. At the same time, the consumer’s risk type may also influence their preferences over insurance products. For example, higher-risk consumers may prefer insurers with higher service quality. We therefore allow the true risk type to affect both the pricing distributions and the utility parameters. 
Our data include realized claim records for each consumer over multiple years, and we use these records as instruments for the latent risk type in the first-step estimation.

We nonparametrically estimate each insurance company’s offered price distribution using our functional contraction approach. 
In Figure \ref{fig:insurance_CDF}, we plot the CDFs of offered price for several firms based on estimates in \cite{cosconati2024competing}.
The distributions differ substantially across firms, indicating significant heterogeneity in their pricing strategies. 
Building on this result, we estimate each firm's information precision parameter using supply-side model restrictions. These estimates provide important insights for analyzing competition under heterogeneous information structures in this market.

\begin{figure}[htpb!]
    \centering
    \includegraphics[width=0.6\linewidth]{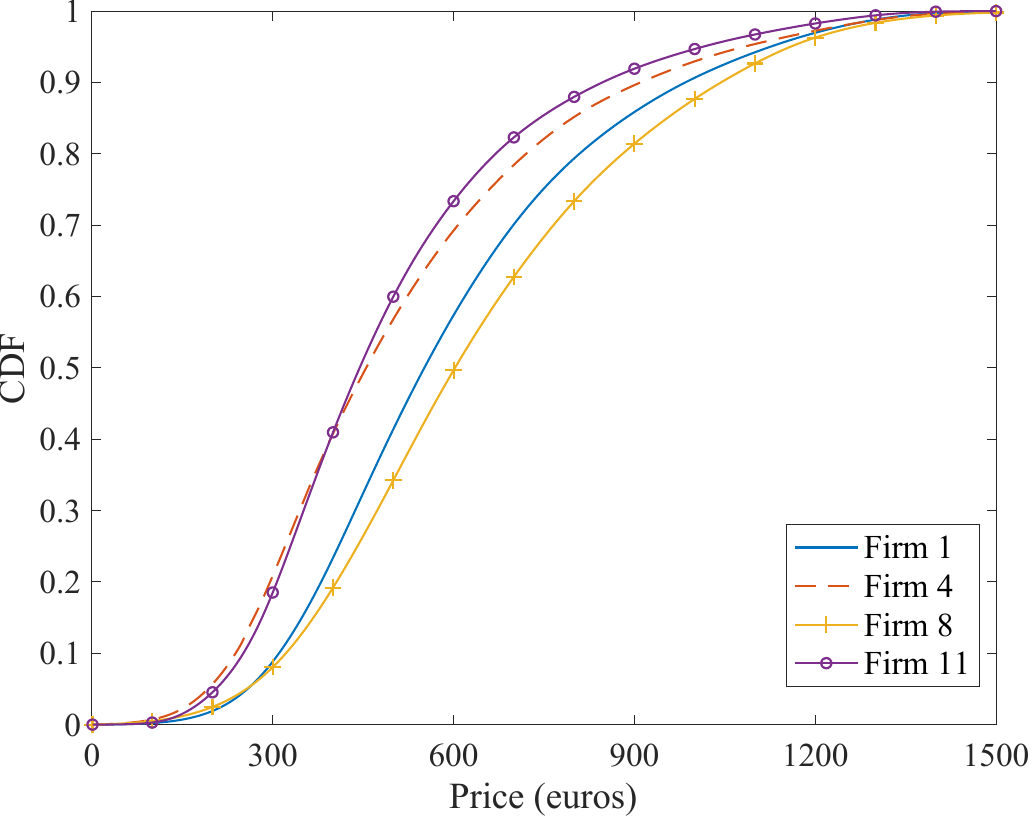}
    \caption{CDFs of the offered price distributions for firms 1, 4, 8, and 11 based on estimates from \cite{cosconati2024competing}. The CDFs are averaged across different characteristic groups.}
    \label{fig:insurance_CDF}
\end{figure}

From a practical point of view, our iterative procedure to numerically solve for the offered price distributions given demand parameters is easy to implement and performs well in practice. In our empirical application using data from 11 insurers, the iterative algorithm converges very quickly, typically requiring only 6--7 iterations.

The usefulness of this method is not limited to consumer demand. It can also be applied to auction models and Roy models, where similar selection issues arise. 
For example, in multi-attribute auctions, our approach can be used to nonparametrically recover the full bid distribution and the auctioneer’s scoring weights when only the winning bids and the winner’s identity are observed, even in the presence of bidder asymmetry.\footnote{Flexibly accommodating bidder asymmetries is a well-known challenge in auction models (see discussions in  \cite{athey2007nonparametric}). Bidder asymmetries may arise from factors such as distance to the contract location \citep{flambard2006asymmetry}, information advantages \citep{hendricks1988empirical,de2009effect}, varying risk attitudes \citep{campo2012risk}, or strategic sophistication \citep{hortaccsu2019does}.
} 
Auctions in
many settings have used the scoring rule that departs from the pure price-based criterion by accounting for quality differences,\footnote{See, for example, \cite{asker2008properties,lewis2011procurement,nakabayashi2013small,yoganarasimhan2016estimation,takahashi2018strategic,krasnokutskaya2020role,allen2024resolving}.} and our framework can flexibly accommodate these multi-attribute scoring mechanisms with both observable and unobservable components.
A similar application arises in Roy-type models, where our method can recover the distribution of potential wages flexibly when only realized wages in the chosen sector are observed. 
Unlike identification strategies that rely on excluded variables in the selection equation, our approach instead requires proxy variables to recover the selected outcome distribution conditional on the latent characteristics.
For example, in studies of occupational choice and wage
inequality, cognitive and noncognitive abilities may be treated as latent characteristics that jointly affect occupational choices and wages, while standardized test scores
or other skill assessments can be used as proxies for these latent characteristics.

\section{Conclusion}
\label{sec:conclusion}

We introduce a new method for estimating nonseparable selection models. We show that, for a given selection function and under suitable contraction conditions, the potential outcome distributions are nonparametrically identified from the distribution of selected outcomes and can be recovered using a simple iterative algorithm.
We achieve this by constructing an operator whose fixed point is the potential outcome distributions and establishing sufficient conditions under which the operator is a functional contraction. 
Building on this theoretical result, we propose a two-step estimation strategy for both the selection function and potential outcome distributions.
Under suitable regularity conditions, we establish consistency of both components and semiparametric asymptotic normality of the selection-parameter estimator without imposing finite price support. The recovered outcome CDFs inherit the first-step convergence rate and admit Gaussian process limits when the corresponding first-step functional central limit theorem holds.

Our method has several features that are relevant for empirical applications.
First, we allow the outcome equation to be fully nonparametric and nonseparable in error terms, and we recover the entire distribution of potential outcomes rather than focusing on specific moments or quantiles. 
In essence, we correct for sample selection bias 
by examining how the bias is \emph{systematically} generated by the selection model. 
Second, our approach allows for fully heterogeneous effects of covariates on outcomes. 
Third, it does not require an excluded variable that shifts the selection probability
without entering the outcome equation, nor does it rely on identification-at-infinity arguments. However, our approach does require a different class of instruments satisfying the conditional independence assumptions commonly imposed in the nonclassical measurement-error literature. 
Finally, the framework accommodates asymmetry in outcome distributions across alternatives and flexibly incorporates unobserved alternative-specific heterogeneity in the selection model.

We find that the proposed estimation strategy performs well in both simulations and real-world data applications; see our demand estimation using insurance market data in \cite{cosconati2024competing}. 
The approach is straightforward to implement and computationally efficient, making it appealing to empirical researchers.
More broadly, the estimator can be applied in a wide range of settings in which only selected outcomes are observed, including consumer demand models with only transaction prices, multi-attribute auctions with incomplete bid data, and various selection models in labor economics.
Our method is particularly valuable in applications where the entire distribution of outcomes is of interest.

\newpage

\begin{singlespace}
    \bibliographystyle{ecta}
    \bibliography{ins_ref}

@article{honore2020selection,
  title={Selection without exclusion},
  author={Honor{\'e}, Bo E and Hu, Luojia},
  journal={Econometrica},
  volume={88},
  number={3},
  pages={1007--1029},
  year={2020},
  publisher={Wiley Online Library}
}

@article{chernozhukov2025distribution,
author = {Chernozhukov, Victor and Fern\'{a}ndez-Val, Iv\'{a}n and Luo, Siyi},
title = {Distribution regression with sample selection and UK wage                         decomposition},
journal = {Journal of Political Economy},
volume = {133},
number = {12},
pages = {3952-3992},
year = {2025},
doi = {10.1086/738148},

URL = { 
    
        https://doi.org/10.1086/738148
    
    

},
eprint = { 
    
        https://doi.org/10.1086/738148
    
    

}
}

@article{d2019automobile,
  title={Automobile prices in market equilibrium with unobserved price discrimination},
  author={D’Haultf{\oe}uille, Xavier and Durrmeyer, Isis and F{\'e}vrier, Philippe},
  journal={The Review of Economic Studies},
  volume={86},
  number={5},
  pages={1973--1998},
  year={2019},
  publisher={Oxford University Press}
}

@article{andrews1998semiparametric,
  title={Semiparametric estimation of the intercept of a sample selection model},
  author={Andrews, Donald WK and Schafgans, Marcia MA},
  journal={The Review of Economic Studies},
  volume={65},
  number={3},
  pages={497--517},
  year={1998},
  publisher={Wiley-Blackwell}
}

@article{gronau1974wage,
  title={Wage comparisons--A selectivity bias},
  author={Gronau, Reuben},
  journal={Journal of Political Economy},
  volume={82},
  number={6},
  pages={1119--1143},
  year={1974},
  publisher={The University of Chicago Press}
}

@article{newey1994large,
  title={Large sample estimation and hypothesis testing},
  author={Newey, Whitney K and McFadden, Daniel},
  journal={Handbook of Econometrics},
  volume={4},
  pages={2111--2245},
  year={1994},
  publisher={Elsevier}
}

@article{meilijson1981estimation,
  title={Estimation of the lifetime distribution of the parts from the autopsy statistics of the machine},
  author={Meilijson, Isaac},
  journal={Journal of Applied Probability},
  volume={18},
  number={4},
  pages={829--838},
  year={1981},
  publisher={Cambridge University Press}
}

@article{athey2007nonparametric,
  title={Nonparametric approaches to auctions},
  author={Athey, Susan and Haile, Philip A},
  journal={Handbook of econometrics},
  volume={6},
  pages={3847--3965},
  year={2007},
  publisher={Elsevier}
}

@article{d2018extremal,
  title={Extremal quantile regressions for selection models and the black--white wage gap},
  author={D’Haultf{\oe}uille, Xavier and Maurel, Arnaud and Zhang, Yichong},
  journal={Journal of Econometrics},
  volume={203},
  number={1},
  pages={129--142},
  year={2018},
  publisher={Elsevier}
}

@article{arellano2017quantile,
  title={Quantile selection models with an application to understanding changes in wage inequality},
  author={Arellano, Manuel and Bonhomme, St{\'e}phane},
  journal={Econometrica},
  volume={85},
  number={1},
  pages={1--28},
  year={2017},
  publisher={Wiley Online Library}
}

@article{fernandez2024nonseparable,
  title = {Nonseparable sample selection models with censored selection rules},
journal = {Journal of Econometrics},
volume = {240},
number = {2},
pages = {105088},
year = {2024},
issn = {0304-4076},
doi = {https://doi.org/10.1016/j.jeconom.2021.01.009},
url = {https://www.sciencedirect.com/science/article/pii/S0304407621000567},
author = {Ivan Fernández-Val and Aico {van Vuuren} and Francis Vella}
}

@article{newey2007nonparametric,
  title={Nonparametric continuous/discrete choice models},
  author={Newey, Whitney K},
  journal={International Economic Review},
  volume={48},
  number={4},
  pages={1429--1439},
  year={2007},
  publisher={Wiley Online Library}
}

@article{chen2003semiparametric,
  title={Semiparametric estimation of a heteroskedastic sample selection model},
  author={Chen, Songnian and Khan, Shakeeb},
  journal={Econometric Theory},
  volume={19},
  number={6},
  pages={1040--1064},
  year={2003},
  publisher={Cambridge University Press}
}

@incollection{heckman1976common,
  title={The common structure of statistical models of truncation, sample selection and limited dependent variables and a simple estimator for such models},
  author={Heckman, James J},
  booktitle={Annals of economic and social measurement, volume 5, number 4},
  pages={475--492},
  year={1976},
  publisher={NBER}
}

@article{lee1983generalized,
  title={Generalized econometric models with selectivity},
  author={Lee, Lung-Fei},
  journal={Econometrica: Journal of the Econometric Society},
  pages={507--512},
  year={1983},
  publisher={JSTOR}
}

@article{lee1982some,
  title={Some approaches to the correction of selectivity bias},
  author={Lee, Lung-Fei},
  journal={The Review of Economic Studies},
  volume={49},
  number={3},
  pages={355--372},
  year={1982},
  publisher={Wiley-Blackwell}
}

@article{vella1998estimating,
  title={Estimating models with sample selection bias: A survey},
  author={Vella, Francis},
  journal={Journal of Human Resources},
  pages={127--169},
  year={1998},
  publisher={JSTOR}
}

@article{heckman1974shadow,
  title={Shadow prices, market wages, and labor supply},
  author={Heckman, James J},
  journal={Econometrica: journal of the econometric society},
  pages={679--694},
  year={1974},
  publisher={JSTOR}
}

@article{newey2009two,
  title={Two-step series estimation of sample selection models},
  author={Newey, Whitney K},
  journal={The Econometrics Journal},
  volume={12},
  number={suppl\_1},
  pages={S217--S229},
  year={2009},
  publisher={Oxford University Press Oxford, UK}
}

@article{ahn1993semiparametric,
  title={Semiparametric estimation of censored selection models with a nonparametric selection mechanism},
  author={Ahn, Hyungtaik and Powell, James L},
  journal={Journal of Econometrics},
  volume={58},
  number={1-2},
  pages={3--29},
  year={1993},
  publisher={Elsevier}
}

@article{heckman1979sample,
  title={Sample selection bias as a specification error},
  author={Heckman, James J},
  journal={Econometrica},
  volume={47},
  pages={153--161},
  year={1979}
}

@incollection{french2011identification,
  title={Identification of models of the labor market},
  author={French, Eric and Taber, Christopher},
  booktitle={Handbook of labor economics},
  volume={4},
  pages={537--617},
  year={2011},
  publisher={Elsevier}
}

@article{das2003nonparametric,
  title={Nonparametric estimation of sample selection models},
  author={Das, Mitali and Newey, Whitney K and Vella, Francis},
  journal={The Review of Economic Studies},
  volume={70},
  number={1},
  pages={33--58},
  year={2003},
  publisher={Wiley-Blackwell}
}

@article{hortaccsu2019does,
  title={Does strategic ability affect efficiency? Evidence from electricity markets},
  author={Horta{\c{c}}su, Ali and Luco, Fernando and Puller, Steven L and Zhu, Dongni},
  journal={American Economic Review},
  volume={109},
  number={12},
  pages={4302--4342},
  year={2019},
  publisher={American Economic Association 2014 Broadway, Suite 305, Nashville, TN 37203}
}

@article{campo2012risk,
  title={Risk aversion and asymmetry in procurement auctions: Identification, estimation and application to construction procurements},
  author={Campo, Sandra},
  journal={Journal of Econometrics},
  volume={168},
  number={1},
  pages={96--107},
  year={2012},
  publisher={Elsevier}
}

@article{de2009effect,
  title={The effect of information on the bidding and survival of entrants in procurement auctions},
  author={De Silva, Dakshina G and Kosmopoulou, Georgia and Lamarche, Carlos},
  journal={Journal of Public Economics},
  volume={93},
  number={1-2},
  pages={56--72},
  year={2009},
  publisher={Elsevier}
}

@article{flambard2006asymmetry,
  title={Asymmetry in procurement auctions: Evidence from snow removal contracts},
  author={Flambard, V{\'e}ronique and Perrigne, Isabelle},
  journal={The Economic Journal},
  volume={116},
  number={514},
  pages={1014--1036},
  year={2006},
  publisher={Oxford University Press Oxford, UK}
}

@article{takahashi2018strategic,
  title={Strategic design under uncertain evaluations: Structural analysis of design-build auctions},
  author={Takahashi, Hidenori},
  journal={The RAND Journal of Economics},
  volume={49},
  number={3},
  pages={594--618},
  year={2018},
  publisher={Wiley Online Library}
}

@article{krasnokutskaya2020role,
  title={The role of quality in internet service markets},
  author={Krasnokutskaya, Elena and Song, Kyungchul and Tang, Xun},
  journal={Journal of Political Economy},
  volume={128},
  number={1},
  pages={75--117},
  year={2020},
  publisher={The University of Chicago Press Chicago, IL}
}

@article{yoganarasimhan2016estimation,
  title={Estimation of beauty contest auctions},
  author={Yoganarasimhan, Hema},
  journal={Marketing Science},
  volume={35},
  number={1},
  pages={27--54},
  year={2016},
  publisher={INFORMS}
}

@article{nakabayashi2013small,
  title={Small business set-asides in procurement auctions: An empirical analysis},
  author={Nakabayashi, Jun},
  journal={Journal of Public Economics},
  volume={100},
  pages={28--44},
  year={2013},
  publisher={Elsevier}
}

@article{asker2008properties,
  title={Properties of scoring auctions},
  author={Asker, John and Cantillon, Estelle},
  journal={The RAND Journal of Economics},
  volume={39},
  number={1},
  pages={69--85},
  year={2008},
  publisher={Wiley Online Library}
}

@article{lewis2011procurement,
  title={Procurement contracting with time incentives: Theory and evidence},
  author={Lewis, Gregory and Bajari, Patrick},
  journal={The Quarterly Journal of Economics},
  volume={126},
  number={3},
  pages={1173--1211},
  year={2011},
  publisher={MIT Press}
}

@article{guerre2019nonparametric,
  title={Nonparametric identification of first-price auction with unobserved competition: A density discontinuity framework},
  author={Guerre, Emmanuel and Luo, Yao},
  journal={arXiv preprint arXiv:1908.05476},
  year={2019}
}

@article{komarova2013new,
  title={A new approach to identifying generalized competing risks models with application to second-price auctions},
  author={Komarova, Tatiana},
  journal={Quantitative Economics},
  volume={4},
  number={2},
  pages={269--328},
  year={2013},
  publisher={Wiley Online Library}
}

@article{athey2002identification,
  title={Identification of standard auction models},
  author={Athey, Susan and Haile, Philip A},
  journal={Econometrica},
  volume={70},
  number={6},
  pages={2107--2140},
  year={2002},
  publisher={Wiley Online Library}
}

@article{allen2024resolving,
  title={Resolving failed banks: Uncertainty, multiple bidding and auction design},
  author={Allen, Jason and Clark, Robert and Hickman, Brent and Richert, Eric},
  journal={Review of Economic Studies},
  volume={91},
  number={3},
  pages={1201--1242},
  year={2024},
  publisher={Oxford University Press UK}
}

@article{dube2023personalized,
  title={Personalized pricing and consumer welfare},
  author={Dub{\'e}, Jean-Pierre and Misra, Sanjog},
  journal={Journal of Political Economy},
  volume={131},
  number={1},
  pages={131--189},
  year={2023},
  publisher={The University of Chicago Press Chicago, IL}
}

@article{buchholz2020value,
  title={The value of time: Evidence from auctioned cab rides},
  author={Buchholz, Nicholas and Doval, Laura and Kastl, Jakub and Mat{\v{e}}jka, Filip and Salz, Tobias},
  year={2020},
  journal={CEPR Discussion Paper No. DP14666}
}

@article{cosconati2024competing,
  title={Competing under information heterogeneity: Evidence from auto insurance},
  author={Cosconati, Marco and Xin, Yi and Wu, Fan and Jin, Yizhou},
  journal={The Review of Economic Studies},
  year    = {2025},
  note    = {Forthcoming}
}

@article{d2013another,
  title={Another look at the identification at infinity of sample selection models},
  author={D’Haultfoeuille, Xavier and Maurel, Arnaud},
  journal={Econometric Theory},
  volume={29},
  number={1},
  pages={213--224},
  year={2013},
  publisher={Cambridge University Press}
}

@article{lee2013nonparametric,
  title={Nonparametric identification of accelerated failure time competing risks models},
  author={Lee, Sokbae and Lewbel, Arthur},
  journal={Econometric Theory},
  volume={29},
  number={5},
  pages={905--919},
  year={2013},
  publisher={Cambridge University Press}
}

@article{heckman1990empirical,
  title={The empirical content of the Roy model},
  author={Heckman, James J and Honore, Bo E},
  journal={Econometrica: Journal of the Econometric Society},
  pages={1121--1149},
  year={1990},
  publisher={JSTOR}
}

@article{roy1951some,
  title={Some thoughts on the distribution of earnings},
  author={Roy, Andrew Donald},
  journal={Oxford Economic Papers},
  volume={3},
  number={2},
  pages={135--146},
  year={1951},
  publisher={Oxford University Press}
}

@article{crawford2018asymmetric,
  title={Asymmetric information and imperfect competition in lending markets},
  author={Crawford, Gregory S and Pavanini, Nicola and Schivardi, Fabiano},
  journal={American Economic Review},
  volume={108},
  number={7},
  pages={1659--1701},
  year={2018},
  publisher={American Economic Association 2014 Broadway, Suite 305, Nashville, TN 37203}
}

@article{sagl2023dispersion,
  title={Dispersion, discrimination, and the price of your pickup},
  author={Sagl, Stephan},
  journal={Working paper},
year={2023}
}

@article{berry1995automobile,
 ISSN = {00129682, 14680262},
 URL = {http://www.jstor.org/stable/2171802},
 author = {Steven Berry and James Levinsohn and Ariel Pakes},
 journal = {Econometrica},
 number = {4},
 pages = {841--890},
 publisher = {[Wiley, Econometric Society]},
 title = {Automobile prices in market equilibrium},
 urldate = {2024-01-20},
 volume = {63},
 year = {1995}
}

@article{hendricks1988empirical,
  title={An empirical study of an auction with asymmetric information},
  author={Hendricks, Kenneth and Porter, Robert H},
  journal={The American Economic Review},
  pages={865--883},
  year={1988},
  publisher={JSTOR}
}

@article{berry1994estimating,
  title={Estimating discrete-choice models of product differentiation},
  author={Berry, Steven T},
  journal={The RAND Journal of Economics},
  pages={242--262},
  year={1994},
  publisher={JSTOR}
}

@article{allen2014price,
  title={Price dispersion in mortgage markets},
  author={Allen, Jason and Clark, Robert and Houde, Jean-Fran{\c{c}}ois},
  journal={The Journal of Industrial Economics},
  volume={62},
  number={3},
  pages={377--416},
  year={2014},
  publisher={Wiley Online Library}
}

@article{allen2019search,
  title={Search frictions and market power in negotiated-price markets},
  author={Allen, Jason and Clark, Robert and Houde, Jean-Fran{\c{c}}ois},
  journal={Journal of Political Economy},
  volume={127},
  number={4},
  pages={1550--1598},
  year={2019},
  publisher={The University of Chicago Press Chicago, IL}
}

@article{goldberg1996dealer,
  title={Dealer price discrimination in new car purchases: Evidence from the consumer expenditure survey},
  author={Goldberg, Pinelopi Koujianou},
  journal={Journal of Political Economy},
  volume={104},
  number={3},
  pages={622--654},
  year={1996},
  publisher={The University of Chicago Press}
}

@article{cicala2015does,
  title={When does regulation distort costs? Lessons from fuel procurement in us electricity generation},
  author={Cicala, Steve},
  journal={American Economic Review},
  volume={105},
  number={1},
  pages={411--444},
  year={2015},
  publisher={American Economic Association 2014 Broadway, Suite 305, Nashville, TN 37203}
}

@article{hu2008identification,
  title={Identification and estimation of nonlinear models with misclassification error using instrumental variables: A general solution},
  author={Hu, Yingyao},
  journal={Journal of Econometrics},
  volume={144},
  number={1},
  pages={27--61},
  year={2008},
  publisher={Elsevier}
}

@article{hu2008instrumental,
  title={Instrumental variable treatment of nonclassical measurement error models},
  author={Hu, Yingyao and Schennach, Susanne M},
  journal={Econometrica},
  volume={76},
  number={1},
  pages={195--216},
  year={2008},
  publisher={Wiley Online Library}
}

@article{bushell1973hilbert,
  title={Hilbert's metric and positive contraction mappings in a Banach space},
  author={Bushell, Peter J},
  journal={Archive for Rational Mechanics and Analysis},
  volume={52},
  number={4},
  pages={330--338},
  year={1973},
  publisher={Springer}
}

@book{lemmens2012nonlinear,
  title={Nonlinear Perron-Frobenius Theory},
  author={Lemmens, Bas and Nussbaum, Roger},
  volume={189},
  year={2012},
  publisher={Cambridge University Press}
}
\end{singlespace}
\newpage

\appendix

\numberwithin{equation}{section}

\clearpage 
\section{Omitted Proofs}
\label{sec:proof}

All proofs are developed using the measure-theoretic formulation of the model presented in Section \ref{sec:formal}.

\subsection{Hilbert’s Projective Metric}\label{sec:Radon Nikodym}

For two probability measures $\Psi_j,\Phi_j \in \Delta([\underline{p}_j,\overline{p}_j])$, we measure their distance using Hilbert’s projective metric. 
If $\Psi_j$ and $ \Phi_j$ are mutually absolutely continuous,
$$d_H(\Psi_j,\Phi_j)
=
\ln
\frac{
\operatorname*{ess\,sup}_{p\in[\underline{p}_j,\overline{p}_j]}
\frac{d\Psi_j}{d\Phi_j}(p)
}{
\operatorname*{ess\,inf}_{p\in[\underline{p}_j,\overline{p}_j]}
\frac{d\Psi_j}{d\Phi_j}(p)
},$$
where $\frac{d\Psi_j}{d\Phi_j}(p)$ denotes the Radon-Nikodym derivative. If $\Psi_j$ and $\Phi_j$ are not mutually absolutely continuous, set $d_H(\Psi_j,\Phi_j)=+\infty$.
When both $\Psi_j$ and $\Phi_j$ have continuous densities, the Radon-Nikodym derivative simplifies to the ratio of densities as in Equation \eqref{eq:hilbert}.

Recall that two probability measures $\Psi_j$ and $\Phi_j$ are equivalent, denoted $\Psi_j\sim \Phi_j$, if they are mutually absolutely continuous. Since the selection function $f_j$ is strictly positive, $\operatorname{Pr}_j(\cdot; G)$ is also strictly positive.  
It follows from Equation \eqref{eq: expost_app} that  $G_j\sim G^{s}_j$. When $\Psi_j\sim \Phi_j$, the Radon-Nikodym derivative, $$\frac{d\Psi_j}{d\Phi_j}\colon [\underline p_j, \overline p_j] \to (0,\infty),$$ exists, as guaranteed by the Radon-Nikodym Theorem. Note that
$$\Psi_j=\Phi_j\quad\Leftrightarrow \quad \frac{d\Psi_j}{d\Phi_j}(p)=1 \quad\quad \Phi_j\text{-a.e. }$$

Given our operator $T$ in Equation \eqref{eq:operator_app}, for all $\Psi_j,\Phi_j\in   \Delta([\underline p_j, \overline p_j])$,
$$(T\Psi)_j\sim G^{s}_j\sim (T\Phi)_j.$$
Moreover, since $f_j>0$ is continuous with compact support, $\operatorname{Pr}_j$ is bounded away from $0$. Thus, $d_H((T\Psi)_j,G_j)$, $d_H((T\Psi)_j,G^{s}_j)$ and $d_H((T\Psi)_j,(T\Phi)_j)$ are all finite.

\subsection{Proof of Theorem \ref{thm: contraction}}
\label{sec:proof_1}

Fix a probability measure
\(G_j\in\Delta([\underline p_j,\overline p_j])\) and define
\[
\mathcal X(G_j)
=
\left\{
\Psi\in\Delta([\underline p_j,\overline p_j]):
d_H(\Psi,G_j)<\infty
\right\}.
\]
Equivalently, \(\Psi\in\mathcal X(G_j)\) if and only if \(\Psi\sim G_j\) and
\[
\log\frac{d\Psi}{dG_j}\in L^\infty(G_j).
\]
\begin{lem}[Completeness of the finite-distance component]
\leavevmode\\
The metric space \((\mathcal X(G_j),d_H)\) is complete. Consequently,
\[
\mathcal X( G)
=
\prod_{j=1}^J\mathcal X( G_j)
\]
is complete under the max metric $D$.
\end{lem}

\begin{proof}
Consider the one-dimensional component first. Let \((\Phi_n)_{n\ge1}\) be a Cauchy sequence in
\((\mathcal X(G_j),d_H)\). Then
\[
\log\frac{d\Phi_n}{dG_j}\in L^\infty(G_j)
\]
and
\[
\int \exp\left(\log\frac{d\Phi_n}{dG_j}\right)dG_j=1
\]
for all \(n\).

For any \(n,m\), we have
\[
\begin{aligned}
d_H(\Phi_n,\Phi_m)
&=
\operatorname{osc}_{G_j}
\left(
\log\frac{d\Phi_n}{dG_j}
-
\log\frac{d\Phi_m}{dG_j}
\right),
\end{aligned}
\]
where
\[
\operatorname{osc}_G(v)
=
\operatorname*{ess\,sup}_G v
-
\operatorname*{ess\,inf}_G v.
\]
Moreover,
\[
\int
\exp\left(
\log\frac{d\Phi_n}{dG_j}
-
\log\frac{d\Phi_m}{dG_j}
\right)
d\Phi_m
=
\int \frac{d\Phi_n}{d\Phi_m}\,d\Phi_m
=
1.
\]
Hence,
\[
\operatorname*{ess\,inf}_{G_j}
\left(
\log\frac{d\Phi_n}{dG_j}
-
\log\frac{d\Phi_m}{dG_j}
\right)
\le 0
\le
\operatorname*{ess\,sup}_{G_j}
\left(
\log\frac{d\Phi_n}{dG_j}
-
\log\frac{d\Phi_m}{dG_j}
\right),
\]
where the essential infimum and supremum can equivalently be taken with respect to
\(G_j\) or \(\Phi_m\), since \(\Phi_m\sim G_j\). Therefore,
\[
\left\|
\log\frac{d\Phi_n}{dG_j}
-
\log\frac{d\Phi_m}{dG_j}
\right\|_\infty
\le
\operatorname{osc}_{G_j}
\left(
\log\frac{d\Phi_n}{dG_j}
-
\log\frac{d\Phi_m}{dG_j}
\right)
=
d_H(\Phi_n,\Phi_m).
\]
Since \((\Phi_n)\) is Cauchy under \(d_H\), it follows that
\[
\left(
\log\frac{d\Phi_n}{dG_j}
\right)_{n\ge1}
\]
is a Cauchy sequence in \(L^\infty(G_j)\). Because \(L^\infty(G_j)\) is a Banach space,
there exists \(\ell\in L^\infty(G_j)\) such that
\[
\left\|
\log\frac{d\Phi_n}{dG_j}
-
\ell
\right\|_\infty
\to 0.
\]

Define a measure \(\Phi\) on \([\underline p_j,\overline p_j]\) by
\[
d\Phi=e^\ell dG_j.
\]
We first show that \(\Phi\) is a probability measure. Since
\[
\left\|
\log\frac{d\Phi_n}{dG_j}
-
\ell
\right\|_\infty
\to0,
\]
for all sufficiently large \(n\),
\[
\exp\left(\log\frac{d\Phi_n}{dG_j}\right)
\le
e^{\ell+1}.
\]
The function \(e^{\ell+1}\) is \(G_j\)-integrable because \(\ell\in L^\infty(G_j)\) and
\(G_j\) is a probability measure. Hence, by dominated convergence,
\[
\int e^\ell dG_j
=
\lim_{n\to\infty}
\int
\exp\left(\log\frac{d\Phi_n}{dG_j}\right)dG_j
=
1.
\]
Thus \(\Phi\in\Delta([\underline p_j,\overline p_j])\).

Since \(\ell\in L^\infty(G_j)\), the density \(e^\ell=d\Phi/dG_j\) is bounded above
and bounded away from zero. Hence \(\Phi\sim G_j\) and
\[
d_H(\Phi,G_j)<\infty.
\]
Therefore \(\Phi\in\mathcal X(G_j)\).

Finally,
\[
\begin{aligned}
d_H(\Phi_n,\Phi)
&=
\operatorname{osc}_{G_j}
\left(
\log\frac{d\Phi_n}{dG_j}
-
\log\frac{d\Phi}{dG_j}
\right) \\
&=
\operatorname{osc}_{G_j}
\left(
\log\frac{d\Phi_n}{dG_j}
-
\ell
\right) \\
&\le
2
\left\|
\log\frac{d\Phi_n}{dG_j}
-
\ell
\right\|_\infty
\to0.
\end{aligned}
\]
Therefore \((\mathcal X(G_j),d_H)\) is complete.

Since the full space
\[
\mathcal X( G)
=
\prod_{j=1}^J\mathcal X( G_j)
\]
is a finite product of complete metric spaces, it is complete under the max metric
\[
D(\Psi,\Phi)
=
\max_{1\leq j\leq J}d_H(\Psi_j,\Phi_j).
\]
\end{proof}

\begin{proof}[Proof of Theorem \ref{thm: contraction}]
If $J=1$, then $\operatorname{Pr}_1(p_1;\Psi)=f_1(p_1)$ is independent of
$\Psi$, so $T$ is constant and the assertions follow immediately.
In the rest of the proof take $J\geq2$.

Fix $j\in\{1,\ldots,J\}$. By the definition of \(T\),
\[
d(T\Psi)_j(p_j)
=
\frac{dG^{s}_j(p_j)/\operatorname{Pr}_j(p_j;\Psi)}
{\int dG^{s}_j(t)/\operatorname{Pr}_j(t;\Psi)}.
\]
Therefore,
\[
\frac{d(T\Psi)_j}{d(T\Phi)_j}(p_j)
=
C_j(\Psi,\Phi)
\frac{\operatorname{Pr}_j(p_j;\Phi)}{\operatorname{Pr}_j(p_j;\Psi)},
\]
where \(C_j(\Psi,\Phi)>0\) is a normalizing constant that does not depend on
\(p_j\). Since Hilbert's projective metric is invariant to multiplication by
positive constants, the normalizing constant cancels out. Hence,
\[
d_H((T\Psi)_j,(T\Phi)_j)
\le
\operatorname{osc}_{p_j}
\log
\frac{\operatorname{Pr}_j(p_j;\Psi)}{\operatorname{Pr}_j(p_j;\Phi)}.
\] where \(\operatorname{osc} g=\sup g-\inf g\). The inequality becomes equality when \(G^{s}_j\) has full support on
\([\underline p_j,\bar p_j]\).

Write
\[
\Psi_{-j}=\bigotimes_{k\ne j}\Psi_k,
\qquad
\Phi_{-j}=\bigotimes_{k\ne j}\Phi_k.
\]
For fixed \(j\), the map
\[
\Psi_{-j}\mapsto \operatorname{Pr}_j(\cdot;\Psi)
\]
is the restriction to product probability measures of a positive linear operator from
the cone of finite nonnegative measures on
\[
\prod_{k\ne j}[\underline p_k,\overline p_k]
\]
to the cone of positive functions on \([\underline p_j,\overline p_j]\). Indeed, for
any nonzero finite nonnegative measure \(\Psi_{-j}\) on
\(\prod_{k\ne j}[\underline p_k,\overline p_k]\), define
\[
\operatorname{Pr}_j(p_j;\Psi_{-j})
=
\int f_j(p_j,p_{-j})\,d\Psi_{-j}(p_{-j}).
\]
This map is linear in \(\Psi_{-j}\) and positive because \(f_j>0\). 
To invoke Birkhoff contraction theorem, the first step is to bound the projective diameter of this positive linear operator. See \cite{bushell1973hilbert} for an excellent introduction on Hilbert's metric and Birkhoff theorem.\footnote{\cite{lemmens2012nonlinear} provide a detailed discussion on Perron-Frobenius theory and Birkhoff theorem.} 

Take any two nonzero finite nonnegative measures \(\Phi_{-j}\) and \(\Psi_{-j}\) on
\(\prod_{k\ne j}[\underline p_k,\bar p_k]\). For any \(p_j,p_j'\),
\[
\frac{
\operatorname{Pr}_j(p_j;\Phi_{-j})\operatorname{Pr}_j(p_j';\Psi_{-j})
}{
\operatorname{Pr}_j(p_j;\Psi_{-j})\operatorname{Pr}_j(p_j';\Phi_{-j})
}
=
\frac{
\int\!\!\int
f_j(p_j,\boldsymbol{p}_{-j})f_j(p_j',\boldsymbol{p}_{-j}')\,
d\Phi_{-j}(\boldsymbol{p}_{-j})d\Psi_{-j}(\boldsymbol{p}_{-j}')
}{
\int\!\!\int
f_j(p_j,\boldsymbol{p}_{-j}')f_j(p_j',\boldsymbol{p}_{-j})\,
d\Phi_{-j}(\boldsymbol{p}_{-j})d\Psi_{-j}(\boldsymbol{p}_{-j}')
}.
\]
By the definition of \(\Delta_j\), for all \(p_j,p_j',\boldsymbol{p}_{-j},\boldsymbol{p}_{-j}'\),
\[
f_j(p_j,\boldsymbol{p}_{-j})f_j(p_j',\boldsymbol{p}_{-j}')
\le
\exp(\Delta_j)
f_j(p_j,\boldsymbol{p}_{-j}')f_j(p_j',\boldsymbol{p}_{-j}).
\]
Integrating both sides with respect to
\(d\Phi_{-j}(\boldsymbol{p}_{-j})d\Psi_{-j}(\boldsymbol{p}_{-j}')\), we obtain
\[
\frac{
\operatorname{Pr}_j(p_j;\Phi_{-j})\operatorname{Pr}_j(p_j';\Psi_{-j})
}{
\operatorname{Pr}_j(p_j;\Psi_{-j})\operatorname{Pr}_j(p_j';\Phi_{-j})
}
\le
\exp(\Delta_j).
\]
Taking logs and then the supremum over \(p_j,p_j'\) gives
\[
d_H\!\left(\operatorname{Pr}_j(\cdot;\Phi_{-j}),
\operatorname{Pr}_j(\cdot;\Psi_{-j})\right)
\le
\Delta_j.
\]
Thus, the projective diameter of the positive linear operator
\(\Psi_{-j}\mapsto \operatorname{Pr}_j(\cdot;\Psi)\) is at most \(\Delta_j\).

By the Birkhoff--Hopf contraction theorem for positive linear operators under
Hilbert's projective metric,
\begin{equation}\label{eq: Pr distance}  d_H\!\left(\operatorname{Pr}_j(\cdot;\Psi),\operatorname{Pr}_j(\cdot;\Phi)\right)
\le
\tanh\left(\frac{\Delta_j}{4}\right)
d_H(\Psi_{-j},\Phi_{-j}).
\end{equation}
Combining this with the previous bound on \(d((T\Psi)_j,(T\Phi)_j)\), we have
\[
d_H((T\Psi)_j,(T\Phi)_j)
\le
\tanh\left(\frac{\Delta_j}{4}\right)
d_H(\Psi_{-j},\Phi_{-j}).
\]

It remains to bound the distance between the product measures. If
\(D(\Psi,\Phi)<\infty\), then \(\Psi_k\sim\Phi_k\) for every \(k\), and
\[
\frac{d\Psi_{-j}}{d\Phi_{-j}}(p_{-j})
=
\prod_{k\ne j}\frac{d\Psi_k}{d\Phi_k}(p_k).
\]
Therefore,
\[
\log\operatorname*{ess\,sup}
\frac{d\Psi_{-j}}{d\Phi_{-j}}
\le
\sum_{k\ne j}
\log\operatorname*{ess\,sup}
\frac{d\Psi_k}{d\Phi_k},
\]
and similarly,
\[
\log\operatorname*{ess\,sup}
\frac{d\Phi_{-j}}{d\Psi_{-j}}
\le
\sum_{k\ne j}
\log\operatorname*{ess\,sup}
\frac{d\Phi_k}{d\Psi_k}.
\]
Hence,
\[
d_H(\Psi_{-j},\Phi_{-j})
\le
\sum_{k\ne j}d_H(\Psi_k,\Phi_k)
\le
(J-1)D(\Psi,\Phi).
\]
This inequality is used only for pairs at finite input distance. No convention for $0\cdot\infty$ is needed; arbitrary initial conjectures are handled by the first-update argument below.

Therefore, for each \(j\),
\[
d_H((T\Psi)_j,(T\Phi)_j)
\le
(J-1)
\tanh\left(\frac{\Delta_j}{4}\right)
D(\Psi,\Phi).
\]
Taking the maximum over \(j\) yields
\[
D(T\Psi,T\Phi)
\le
(J-1)\max_{1\leq j\leq J}
\tanh\left(\frac{\Delta_j}{4}\right)
D(\Psi,\Phi).
\]
Thus, if
\[
\rho
=
(J-1)\max_{1\leq j\leq J}
\tanh\left(\frac{\Delta_j}{4}\right)
<1,
\]
then \(T\) is a contraction under \(D\) with modulus at most \(\rho\).

To obtain existence and uniqueness of the fixed point, fix the selected distribution
\(G^{s}\) and consider the finite-distance component
\[
\mathcal X(G^{s})
=
\prod_{j=1}^J
\left\{
\Psi_j\in \Delta([\underline p_j,\bar p_j]):
d_H(\Psi_j,G^{s}_j)<\infty
\right\}.
\]
By Lemma 1, \(\mathcal X(G^{s})\) is complete under the metric \(D\).
Moreover, for any
\[
\Psi\in \prod_{j=1}^J\Delta([\underline p_j,\bar p_j]),
\]
the update \(T\Psi\) satisfies
\[
(T\Psi)_j \sim G^{s}_j
\qquad \text{for all } 1\leq j\leq J .
\]
Since \(f_j\) is continuous and strictly positive on the compact support, the selection
probability \(\operatorname{Pr}_j(\cdot;\Psi)\) is bounded above and bounded away from zero. Hence
\[
d_H((T\Psi)_j,G^{s}_j)<\infty
\qquad \text{for all } 1\leq j\leq J ,
\]
and therefore
\[
T\Psi\in \mathcal X(G^{s}).
\]
Thus, \(T\) maps into \(\mathcal X(G^{s})\).

The contraction inequality established above therefore applies on the complete metric
space \(\mathcal X(G^{s})\). By the Banach fixed-point theorem, \(T\) has a
unique fixed point in $\mathcal X(G^{s})$. Denote it by $G$.
The Bayes identity implies $F(G)=G^{s}$, and any probability fixed
point must lie in this same finite-distance component.

Finally, if the initial conjecture
\[
\Psi\in \prod_{j=1}^J\Delta([\underline p_j,\bar p_j])
\]
does not lie in \(\mathcal X(G^{s})\), then its first iterate \(T\Psi\) does.
Applying the contraction result from the second iterate onward yields
\[
D(T^n\Psi,G)
\le
\rho^{\,n-1}D(T\Psi,G)
\to 0 .
\]
Hence, for any initial conjecture \(\Psi\),
\[
\lim_{n\to\infty}T^n\Psi = G.
\]

\end{proof}

\subsection{Proof of Theorem \ref{thm: contraction special}}
\label{sec:proof_2}

\begin{proof}[Proof of Theorem \ref{thm: contraction special}]
By Theorem \ref{thm: contraction}, it suffices to show that, under Assumption
\ref{asmp: monotone}, the projective diameter for each \(j\) satisfies
\[
\Delta_j
=
\ln f_j(\overline{\boldsymbol p})
-\ln f_j(\underline p_j,\overline{\boldsymbol p}_{-j})
-\ln f_j(\overline p_j,\underline{\boldsymbol p}_{-j})
+\ln f_j(\underline{\boldsymbol p}).
\]

Fix $j\in\{1,\ldots,J\}$. For any
\(p_j,p_j'\in[\underline p_j,\overline p_j]\) and
\(\boldsymbol p_{-j},\boldsymbol p_{-j}'\in
\prod_{k\ne j}[\underline p_k,\overline p_k]\), consider
\[
\ln
\frac{
f_j(p_j,\boldsymbol p_{-j})
f_j(p_j',\boldsymbol p_{-j}')
}{
f_j(p_j',\boldsymbol p_{-j})
f_j(p_j,\boldsymbol p_{-j}')
}.
\]
Since interchanging \(p_j\) and \(p_j'\) changes the sign of this expression, we may
assume without loss of generality that \(p_j\ge p_j'\). Then
\[
\begin{aligned}
&\ln
\frac{
f_j(p_j,\boldsymbol p_{-j})
f_j(p_j',\boldsymbol p_{-j}')
}{
f_j(p_j',\boldsymbol p_{-j})
f_j(p_j,\boldsymbol p_{-j}')
} \\
&=
\int_{p_j'}^{p_j}
\left[
\frac{\partial \ln f_j(s,\boldsymbol p_{-j})}{\partial s}
-
\frac{\partial \ln f_j(s,\boldsymbol p_{-j}')}{\partial s}
\right]ds .
\end{aligned}
\]
By Assumption \ref{asmp: monotone},
\(\partial \ln f_j(s,\boldsymbol p_{-j})/\partial s\) is weakly increasing in each
component of \(\boldsymbol p_{-j}\). Therefore, for every \(s\),
\[
\left|
\frac{\partial \ln f_j(s,\boldsymbol p_{-j})}{\partial s}
-
\frac{\partial \ln f_j(s,\boldsymbol p_{-j}')}{\partial s}
\right|
\le
\frac{\partial \ln f_j(s,\overline{\boldsymbol p}_{-j})}{\partial s}
-
\frac{\partial \ln f_j(s,\underline{\boldsymbol p}_{-j})}{\partial s}.
\]
It follows that
\[
\begin{aligned}
&\left|
\ln
\frac{
f_j(p_j,\boldsymbol p_{-j})
f_j(p_j',\boldsymbol p_{-j}')
}{
f_j(p_j',\boldsymbol p_{-j})
f_j(p_j,\boldsymbol p_{-j}')
}
\right| \\
&\le
\int_{\underline p_j}^{\overline p_j}
\left[
\frac{\partial \ln f_j(s,\overline{\boldsymbol p}_{-j})}{\partial s}
-
\frac{\partial \ln f_j(s,\underline{\boldsymbol p}_{-j})}{\partial s}
\right]ds \\
&=
\ln f_j(\overline{\boldsymbol p})
-\ln f_j(\underline p_j,\overline{\boldsymbol p}_{-j})
-\ln f_j(\overline p_j,\underline{\boldsymbol p}_{-j})
+\ln f_j(\underline{\boldsymbol p}).
\end{aligned}
\]
The upper bound is attained by taking
\[
p_j=\overline p_j,\qquad
p_j'=\underline p_j,\qquad
\boldsymbol p_{-j}=\overline{\boldsymbol p}_{-j},\qquad
\boldsymbol p_{-j}'=\underline{\boldsymbol p}_{-j}.
\]
Hence,
\[
\Delta_j
=
\ln f_j(\overline{\boldsymbol p})
-\ln f_j(\underline p_j,\overline{\boldsymbol p}_{-j})
-\ln f_j(\overline p_j,\underline{\boldsymbol p}_{-j})
+\ln f_j(\underline{\boldsymbol p}).
\]
By Theorem \ref{thm: contraction},
\[
D(T\Psi,T\Phi)
\le
(J-1)\max_{1\leq j\leq J}
\tanh\left(\frac{\Delta_j}{4}\right)
D(\Psi,\Phi).
\]
Substituting the expression for \(\Delta_j\) derived above gives
\[
D(T\Psi,T\Phi)\le \rho^*D(\Psi,\Phi).
\]
Therefore, if \(\rho^*<1\), the operator \(T\) is a contraction with modulus at most
\(\rho^*\).

\end{proof}

\subsection{Proof of Proposition \ref{prop:common_kernel}}
\label{sec:proof_common_kernel}

\begin{proof}
Fix the conditioning characteristics and the selection function.
Suppose that two collections of offered distributions,
$G=(G_1,\ldots,G_J)$ and $G'=(G'_1,\ldots,G'_J)$, satisfy
\[
 F(G)=F(G')=G^s.
\]
We show that $G_j=G'_j$ for every $j$.

For any collection of offered distributions $G$, define
\[
 Z_G
 =
 \int_{\mathcal P}
 K(\boldsymbol p)\prod_{k=1}^J dG_k(p_k),
 \qquad
 dQ^G(\boldsymbol p)
 =
 \frac{K(\boldsymbol p)}{Z_G}
 \prod_{k=1}^J dG_k(p_k).
\]
Continuity and strict positivity of $K$ on the compact set
$\mathcal P$ imply that $Z_G$ is finite and strictly positive.
Thus $Q^G$ is a probability measure. Let $Q_j^G$ denote its
$j$th marginal. By construction,
\[
 dQ_j^G(p_j)
 =
 \frac{
 \displaystyle
 \int_{\prod_{k\ne j}\mathcal P_k}
 K(p_j,\boldsymbol p_{-j})
 \prod_{k\ne j}dG_k(p_k)
 }{Z_G}\,dG_j(p_j).
\]

Using \eqref{eq:common_kernel} in the forward mapping
\eqref{eq: expost_app} gives
\[
 dF_j(G)(p_j)
 =
 \frac{a_j(p_j)\,dQ_j^G(p_j)}
      {\displaystyle
       \int_{\mathcal P_j}a_j(t)\,dQ_j^G(t)}.
\]
Consequently, the common selected distribution $G_j^s$ determines
the auxiliary marginal:
\begin{equation}\label{eq:common_kernel_marginals}
 dQ_j^G(p_j)
 =
 dQ_j^{G'}(p_j)
 =
 \frac{a_j(p_j)^{-1}\,dG_j^s(p_j)}
      {\displaystyle
       \int_{\mathcal P_j}a_j(t)^{-1}\,dG_j^s(t)},
 \qquad j=1,\ldots,J.
\end{equation}
In particular, this recovery uses only the conditional selected
distributions and the given selection function, not choice shares.

We next establish that equality of these marginals implies
$Q^G=Q^{G'}$. Write
\[
 \underline K=\min_{\boldsymbol p\in\mathcal P}K(\boldsymbol p)>0,
 \qquad
 \overline K=\max_{\boldsymbol p\in\mathcal P}K(\boldsymbol p)<\infty.
\]
The marginal formula above implies
\[
 \frac{\underline K}{\overline K}
 \leq
 \frac{dQ_j^G}{dG_j}
 \leq
 \frac{\overline K}{\underline K}
 \qquad G_j\text{-almost everywhere},
\]
and the same bounds hold with $G$ replaced by $G'$.
Since $Q_j^G=Q_j^{G'}$, the measures $G_j$ and $G'_j$ are
mutually absolutely continuous. Moreover, their density ratio
\[
 r_j(p_j):=\frac{dG_j}{dG'_j}(p_j)
\]
satisfies
\[
 \left(\frac{\underline K}{\overline K}\right)^2
 \leq r_j
 \leq
 \left(\frac{\overline K}{\underline K}\right)^2
 \qquad G'_j\text{-almost everywhere}.
\]
Indeed, equality of the auxiliary marginals gives
\[
 r_j=\frac{dQ_j^{G'}/dG'_j}{dQ_j^G/dG_j}
 \qquad G'_j\text{-almost everywhere}.
\]
Thus $\log r_j\in L^\infty(G'_j)$.

It follows that $Q^G$ and $Q^{G'}$ are also mutually absolutely
continuous, and cancellation of the common kernel gives
\[
 \log\frac{dQ^G}{dQ^{G'}}(\boldsymbol p)
 =
 \log\frac{Z_{G'}}{Z_G}
 +
 \sum_{j=1}^J\log r_j(p_j).
\]
All logarithms in this identity are bounded, so the following
integrals are finite. Since both joint measures have total mass
one and their corresponding marginals agree,
\[
\begin{aligned}
 &\int_{\mathcal P}
 \log\frac{dQ^G}{dQ^{G'}}(\boldsymbol p)
 \bigl[dQ^G(\boldsymbol p)-dQ^{G'}(\boldsymbol p)\bigr]\\
 &\qquad=
 \sum_{j=1}^J
 \int_{\mathcal P_j}
 \log r_j(p_j)
 \bigl[dQ_j^G(p_j)-dQ_j^{G'}(p_j)\bigr]\\
 &\qquad=0.
\end{aligned}
\]
The integral on the left is the symmetrized relative entropy.
Equivalently, it is
\[
 \int_{\mathcal P}
 \left(\frac{dQ^G}{dQ^{G'}}-1\right)
 \log\left(\frac{dQ^G}{dQ^{G'}}\right)
 \,dQ^{G'}.
\]
For every $z>0$, $(z-1)\log z\geq0$, with equality only at $z=1$.
Therefore,
\[
 \frac{dQ^G}{dQ^{G'}}=1
 \quad Q^{G'}\text{-almost everywhere},
 \qquad\text{and hence}\qquad
 Q^G=Q^{G'}.
\]

Finally, the offered product measure is uniquely recovered by
removing the strictly positive common weight:
\[
 \prod_{j=1}^J dG_j(p_j)
 =
 \frac{K(\boldsymbol p)^{-1}\,dQ^G(\boldsymbol p)}
      {\displaystyle
       \int_{\mathcal P}
       K(\boldsymbol q)^{-1}\,dQ^G(\boldsymbol q)}.
\]
The same formula applies to $G'$. Since $Q^G=Q^{G'}$, we obtain
\[
 \bigotimes_{j=1}^J G_j
 =
 \bigotimes_{j=1}^J G'_j.
\]
Taking marginals yields $G_j=G'_j$ for every $j$, proving that
$F$ is injective.
\end{proof}

\subsection{Proof of Theorem \ref{thm: consistency}}
\label{sec:proof_3}

\begin{proof}[Proof of Proposition \ref{prop:homeomorphism}]
By Theorem \ref{thm: contraction}, for every inside selected
collection $G^{s}$, the operator $T_{G^{s}}$ has a unique fixed point.
The Bayes identity \eqref{eq: expost_app} implies that this is the
unique $G$ with $F(G)=G^{s}$; hence $F$ is bijective.
For $D(G,G')<\infty$, reweighting and normalization give
\[
\begin{split}
 d_H(F_j(G),F_j(G'))
 &\leq d_H(G_j,G_j')
 +\operatorname{osc}_{p_j}\log
   \frac{\operatorname{Pr}_j(p_j;G)}{\operatorname{Pr}_j(p_j;G')}\\
 &\leq(1+\rho)D(G,G'),\qquad j=1,\ldots,J,
\end{split}
\]
using \eqref{eq: Pr distance} and the product-measure bound.
Taking the maximum over the inside indices proves the forward bound.

For $D(G^{s},G^{s\prime})<\infty$, let $G=F^{-1}(G^{s})$ and
$G'=F^{-1}(G^{s\prime})$. Strict positivity on compact support
gives $D(G,G^{s})<\infty$ and $D(G',G^{s\prime})<\infty$,
and thus $D(G,G')<\infty$. Reweighting both selected measures by
the same positive function preserves their projective distance, so
\[
 D(T_{G^{s}}G,T_{G^{s\prime}}G)=D(G^{s},G^{s\prime}).
\]
The two fixed-point equations and contraction imply
\[
\begin{split}
 D(G,G')
 &\leq D(T_{G^{s}}G,T_{G^{s\prime}}G)
       +D(T_{G^{s\prime}}G,T_{G^{s\prime}}G')\\
 &\leq D(G^{s},G^{s\prime})+\rho D(G,G').
\end{split}
\]
Rearranging gives the inverse bound. The finite-distance components
are preserved by $F$ and its inverse, and these inequalities give
continuity on each component and hence in the extended-metric
topology. The outside option adds no coordinate to these spaces.
\end{proof}

For the CDF-space argument, fix $(x,x^\ast)$ and suppress it in
expressions local to this conditioning pair. Let
\[
 \mathbb D_j
 =\{v_j:\mathcal P_j\to\mathbb R:
       v_j\text{ is c\`adl\`ag and }v_j(\overline p_j)=0\},
 \qquad
 \mathbb D=\prod_{j=1}^J\mathbb D_j,
\]
with $\|v\|_\infty=\max_{1\leq j\leq J}
\sup_{t\in\mathcal P_j}|v_j(t)|$.
This is a Banach space: uniform limits preserve the c\`adl\`ag
property and the endpoint constraint. Each perturbation is extended
by zero to the left of $\underline p_j$, so
$v_j(\underline p_j-)=0$. This convention is automatic for differences
of CDFs supported on $\mathcal P_j$ and does not require
$v_j(\underline p_j)=0$; different masses at the lower endpoint are
allowed. For $q\in C^1(\mathcal P_j)$, define
\begin{equation}\label{eq:extended_stieltjes}
 \int_{[\underline p_j,t]}q(p)\,dv_j(p)
 :=q(t)v_j(t)-\int_{\underline p_j}^{t}v_j(p)q'(p)\,dp.
\end{equation}
For bounded-variation perturbations this is the usual Stieltjes
integral, including the atom at $\underline p_j$. For other elements
of $\mathbb D_j$, the right-hand side defines its continuous
extension. In particular,
\begin{equation}\label{eq:stieltjes_bound}
 \sup_{t\in\mathcal P_j}\left|
       \int_{[\underline p_j,t]}q\,dv_j\right|
 \leq\bigl(\|q\|_\infty+
  (\overline p_j-\underline p_j)\|q'\|_\infty\bigr)\|v_j\|_\infty.
\end{equation}
This makes the integral a bounded linear map in the CDF
perturbation. All subsequent integrals against such perturbations
use \eqref{eq:extended_stieltjes}. Repetition in different coordinates
defines the corresponding continuous multilinear product integrals.

\begin{lem}[CDF stability and differentiation of the inverse]
\label{lem: equicontinuity}
Under Assumption \ref{asmp:compact}(ii)--(iii),
\[
 \mathcal R(G^{s},\theta;x,x^\ast)
 :=F^{-1}(G^{s};\theta,x,x^\ast)
\]
is uniformly continuous on the domain
\[
 \left(\prod_{j=1}^J\Delta(\mathcal P_j)\right)\times\Theta,
\]
using the CDF supremum norm and Euclidean parameter norm, uniformly
over $(x,x^\ast)$. Around every probability input it admits a local
extension to the affine CDF space and a parameter
neighborhood.\footnote{For $j=1,\ldots,J$, the affine space is
\[
 \mathbb A_j=\{\Phi_j:\mathcal P_j\to\mathbb R:
 \Phi_j\text{ is c\`adl\`ag},\ \Phi_j(\overline p_j)=1\},
 \qquad \mathbb A=\prod_{j=1}^J\mathbb A_j,
\]
with $\Phi_j(\underline p_j-)=0$ by left extension. For any probability
CDF $G_j^\circ$, $\mathbb A_j=G_j^\circ+\mathbb D_j$.
Elements need not be nondecreasing, lie in $[0,1]$, or have bounded
variation; ``total mass one'' refers to their endpoint normalization.
A local extension of $\mathcal R$ is a map from an open neighborhood
in $\mathbb A\times\mathbb R^{d_\theta}$ into $\mathbb A$ that agrees
with the original inverse at nearby probability inputs. In the proof,
$T$ is first extended in both its distribution arguments, to a
neighborhood in $\mathbb A\times\mathbb A\times\mathbb R^{d_\theta}$.}
The first and second Fr\'echet derivatives of this extension are
continuous and bounded. Mixed derivatives of total order at most
three, with at most two parameter derivatives, also exist,
are continuous, and are locally bounded. The neighborhood sizes
and bounds are uniform over $(x,x^\ast)$, $\theta\in\Theta$, and
probability inputs. In particular, for sufficiently close inputs,
\begin{equation}\label{eq:cdf_inverse_lipschitz}
\begin{split}
 &\|\mathcal R(G^{s\prime},\theta';x,x^\ast)
       -\mathcal R(G^{s},\theta;x,x^\ast)\|_\infty\\
 &\qquad\leq C\bigl(\|G^{s\prime}-G^{s}\|_\infty
                         +\|\theta'-\theta\|\bigr).
\end{split}
\end{equation}
For a CDF direction $\eta$ and a parameter direction $a$,
\begin{equation}\label{eq:inverse_derivative}
 D\mathcal R[\eta,a]
 =(I-\mathcal A)^{-1}\{\mathcal B\eta+D_\theta T[a]\},
 \qquad
 \mathcal A=D_\Psi T,\quad \mathcal B=D_{G^{s}}T,
\end{equation}
where derivatives of $T$ are evaluated at
$\Psi=\mathcal R(G^{s},\theta;x,x^\ast)$ and $(G^{s},\theta)$.
The same statements apply to measurable collections indexed by
$(x,x^\ast)$ with norm \eqref{eq:cdf_norm}.
\end{lem}

\begin{proof}
All constants and neighborhood sizes below are uniform over the
conditioning pairs and probability inputs. We suppress $(x,x^\ast)$
unless a supremum over these arguments is being displayed.

\emph{Step 1: Smooth extension of the Bayes operator.}
The original operator takes two probability-CDF collections and a
parameter as inputs and returns a probability-CDF collection.
We extend it locally to
$\mathbb A\times\mathbb A\times\mathbb R^{d_\theta}$, with
$\mathbb A$ defined in the lemma's footnote. For
$\Psi=\Psi^\circ+v$, $v\in\mathbb D$, the product integral
\[
 \operatorname{Pr}_j(p_j;\Psi,\theta)
 =\int f_j(\boldsymbol p;\theta)
       \prod_{k\ne j}d\Psi_k(p_k),\qquad j=1,\ldots,J,
\]
is a polynomial of degree at most $J-1$ in $v$.
By repeated integration by parts, the derivatives in Assumption
\ref{asmp:compact}(ii) make it a continuous polynomial map into
$C^1(\mathcal P_j)$. Each integrated coordinate uses at most one
price derivative. Differentiating the resulting function in $p_j$
uses one additional derivative in the own-price coordinate.
The same argument applies after up to two parameter derivatives.
On a uniformly bounded neighborhood of probability inputs all these
multilinear maps have uniform bounds.

At probability inputs,
$\varepsilon\leq\operatorname{Pr}_j\leq1$, and
\[
 \int_{\mathcal P_j}
 \frac{dG^{s}_j(p)}{\operatorname{Pr}_j(p;\Psi,\theta)}\geq1.
\]
The preceding continuity estimates and \eqref{eq:stieltjes_bound}
therefore provide a common small neighborhood on which both these
positive quantities stay bounded away from zero. For an affine
input $\Phi_j$, integration means the extension of
\eqref{eq:extended_stieltjes} to $\Phi_j=G_j^\circ+v_j$;
it agrees with the usual integral when $\Phi_j$ is a CDF.
Thus the same formula
\[
 T_j(\Psi;G^{s},\theta)(t)
 =\frac{\displaystyle
  \int_{[\underline p_j,t]}
  \operatorname{Pr}_j(p;\Psi,\theta)^{-1}\,dG^{s}_j(p)}
 {\displaystyle
  \int_{\mathcal P_j}
  \operatorname{Pr}_j(p;\Psi,\theta)^{-1}\,dG^{s}_j(p)}
\]
defines a local extension, still denoted $T$, into $\mathbb A$.
It need not return a probability CDF at nonprobability inputs,
but $T_j(\overline p_j)=1$ still holds.
Reciprocation is smooth on functions bounded away from zero in the
Banach algebra $C^1(\mathcal P_j)$, and weighted integration is
continuous bilinear in its weight and CDF argument.
Consequently, $T$ is smooth in its distribution arguments and twice
continuously differentiable in $\theta$, with the uniform mixed
derivative bounds stated in the lemma. This supplies an open domain
for differentiation without assuming that the probability-CDF set
has nonempty interior.

\emph{Step 2: Continuity of the probability-valued inverse.}
Fix an initial probability collection $\Psi^\circ$ independent of
the inputs. Set $G^{(m)}=T^m(\Psi^\circ;G^{s},\theta)$ and
$G=\mathcal R(G^{s},\theta)$. The first update and the fixed point
are obtained by reweighting the same selected distribution by
weights in $[1,1/\varepsilon]$. Their projective distance is
therefore uniformly bounded, and contraction gives, for $m\geq1$,
\[
 D(G^{(m)},G)\leq C\bar\rho^{m-1}.
\]
For probability measures at finite projective distance $d$,
normalization implies $e^{-d}\leq dU/dV\leq e^d$ and hence
$\sup_t|U(t)-V(t)|\leq e^d-1$. It follows that
\begin{equation}\label{eq:uniform_cdf_iteration}
 \|G^{(m)}-G\|_\infty\leq C\bar\rho^{m-1},
\end{equation}
uniformly over the selected probability inputs, parameters, and
conditioning pairs.

The common local neighborhoods and first-derivative bounds from
Step 1 give uniform continuity of each finite iterate in
$(G^{s},\theta)$. More explicitly, for each fixed $m$ there are
$C_m,\delta_m>0$, independent of the base input, such that two
inputs at distance $d<\delta_m$ have $m$th iterates at CDF distance
at most $C_md$. Applying \eqref{eq:uniform_cdf_iteration} at both
inputs gives
\[
 \|\mathcal R(G^{s\prime},\theta')-
          \mathcal R(G^{s},\theta)\|_\infty
 \leq 2C\bar\rho^{m-1}+C_md,
 \quad
 d=\|G^{s\prime}-G^{s}\|_\infty+\|\theta'-\theta\|.
\]
Given any output tolerance, first choose a finite $m$ to make the
first term small and then choose $d$ to make the second small.
This proves uniform continuity on
$\bigl(\prod_j\Delta(\mathcal P_j)\bigr)\times\Theta$.
No uniform bound on $C_m$ as $m\to\infty$ is required.
This step will identify the subsequent local implicit solution
with the actual probability-valued inverse.

\emph{Step 3: First derivatives at a fixed point.}
Evaluate $T$ at $G=\mathcal R(G^{s},\theta)$. The Bayes identity is
\[
 dG^{s}_j(p)
 =\frac{\operatorname{Pr}_j(p;G,\theta)}
        {\pi_j(\theta,G^{s})}\,dG_j(p),
 \qquad j=1,\ldots,J,
\]
where only the fixed conditioning arguments are suppressed.
For a CDF direction $v\in\mathbb D$,
\begin{equation}\label{eq:linearized_selection}
 D_G\operatorname{Pr}_j(p_j;G,\theta)[v]
 =\sum_{k\ne j}\int f_j(\boldsymbol p;\theta)
   \,dv_k(p_k)\prod_{\ell\ne j,k}dG_\ell(p_\ell).
\end{equation}
Differentiating the numerator and denominator of the update gives
\begin{equation}\label{eq:A_operator}
\begin{split}
 (\mathcal A v)_j(t)
 &=-\int_{[\underline p_j,t]}
   \frac{D_G\operatorname{Pr}_j(p;G,\theta)[v]}
        {\operatorname{Pr}_j(p;G,\theta)}\,dG_j(p)\\
 &\quad+G_j(t)\int_{\mathcal P_j}
   \frac{D_G\operatorname{Pr}_j(p;G,\theta)[v]}
        {\operatorname{Pr}_j(p;G,\theta)}\,dG_j(p),
\end{split}
\end{equation}
and
\begin{equation}\label{eq:B_operator}
\begin{split}
 (\mathcal B\eta)_j(t)
 =\pi_j(\theta,G^{s})\bigg\{
 &\int_{[\underline p_j,t]}
       \frac{d\eta_j(p)}{\operatorname{Pr}_j(p;G,\theta)}\\
 &-G_j(t)\int_{\mathcal P_j}
       \frac{d\eta_j(p)}{\operatorname{Pr}_j(p;G,\theta)}\bigg\}.
\end{split}
\end{equation}
For a parameter direction $a\in\mathbb R^{d_\theta}$,
\[
 D_\theta\operatorname{Pr}_j(p_j;G,\theta)[a]
 =\int\partial_\theta f_j(\boldsymbol p;\theta)a
       \prod_{k\ne j}dG_k(p_k),
\]
where the distribution is held fixed in this partial derivative.
For $(D_\theta T[a])_j$, use \eqref{eq:A_operator} and replace
the distribution derivative $D_G\operatorname{Pr}_j[v]$ by the
parameter derivative $D_\theta\operatorname{Pr}_j[a]$.
All these maps are bounded linear operators by
\eqref{eq:stieltjes_bound} and uniform positivity.

\emph{Step 4: A bounded inverse on the CDF space.}
We do not assert $\|\mathcal A\|_{\infty\to\infty}<1$.
Instead, consider the relative-density perturbation space
\[
 \mathbb E_G=\prod_{j=1}^J
 \{b_j\in L^\infty(G_j):\textstyle\int b_j\,dG_j=0\},
 \qquad
 \|b\|_{\mathrm{osc}}=
 \max_{1\leq j\leq J}\operatorname{osc}_{G_j}(b_j).
\]
On this mean-zero space,
$\|b\|_{L^\infty}\leq\|b\|_{\mathrm{osc}}
\leq2\|b\|_{L^\infty}$, so the oscillation norm is complete.
The symbols $b_j$ here denote density directions, not price endpoints.
Define
\[
 (J_Gb)_j(t)=\int_{[\underline p_j,t]}b_j(p)\,dG_j(p)
\]
and
\[
\begin{split}
 (\mathcal L_Gv)_j(p)
 &=-\frac{D_G\operatorname{Pr}_j(p;G,\theta)[v]}
          {\operatorname{Pr}_j(p;G,\theta)}\\
 &\quad+\int_{\mathcal P_j}
  \frac{D_G\operatorname{Pr}_j(r;G,\theta)[v]}
       {\operatorname{Pr}_j(r;G,\theta)}\,dG_j(r).
\end{split}
\]
Then $\mathcal A=J_G\mathcal L_G$,
$\mathcal L_G:\mathbb D\to\mathbb E_G$, and
\begin{equation}\label{eq:smoothing_bound}
 \|J_Gb\|_\infty\leq\|b\|_{\mathrm{osc}},
 \qquad
 \|\mathcal L_Gv\|_{\mathrm{osc}}\leq C\|v\|_\infty.
\end{equation}

For fixed $b\in\mathbb E_G$, take the probability path
\[
 dG_{j,\epsilon}=(1+\epsilon b_j)\,dG_j
\]
for sufficiently small $|\epsilon|$. Boundedness and the mean-zero
constraint ensure positivity and total mass one. The selected
collection $G^{s}$, the parameter $\theta$, and $(x,x^\ast)$ remain
fixed throughout; in particular the same operator $T$ is used and
$TG=G$. In CDF coordinates,
\[
 G_{j,\epsilon}(t)-G_j(t)
 =\epsilon\int_{[\underline p_j,t]}b_j(p)\,dG_j(p)
 =\epsilon(J_Gb)_j(t).
\]
Since
$\log(1+\epsilon b_j)=\epsilon b_j+O(\epsilon^2)$ in
$L^\infty(G_j)$ for this fixed direction,
\[
 D(G_\epsilon,G)
 =|\epsilon|\|b\|_{\mathrm{osc}}+o(|\epsilon|).
\]

Expanding the product of perturbed marginal measures gives
\[
\begin{split}
 \operatorname{Pr}_j(p_j;G_\epsilon,\theta)
 &=\int f_j(\boldsymbol p;\theta)
   \prod_{k\ne j}(1+\epsilon b_k(p_k))
   \prod_{k\ne j}dG_k(p_k)\\
 &=\operatorname{Pr}_j(p_j;G,\theta)
   +\epsilon D_G\operatorname{Pr}_j(p_j;G,\theta)[J_Gb]
   +O(\epsilon^2),
\end{split}
\]
uniformly over $p_j$ and $j$ for fixed $b$. Indeed,
$d(J_Gb)_k=b_k\,dG_k$, and all higher product terms have at least
two factors of $\epsilon$. Uniform positivity therefore implies
\[
 \frac{\operatorname{Pr}_j(p;G,\theta)}
      {\operatorname{Pr}_j(p;G_\epsilon,\theta)}
 =1-\epsilon
 \frac{D_G\operatorname{Pr}_j(p;G,\theta)[J_Gb]}
      {\operatorname{Pr}_j(p;G,\theta)}+O(\epsilon^2).
\]
Substitution of the fixed-point Bayes identity into the update gives
the exact relative-density formula
\[
 \frac{d(TG_\epsilon)_j}{dG_j}(p)
 =\frac{\displaystyle
    \operatorname{Pr}_j(p;G,\theta)/
    \operatorname{Pr}_j(p;G_\epsilon,\theta)}
  {\displaystyle\int_{\mathcal P_j}
    \frac{\operatorname{Pr}_j(r;G,\theta)}
         {\operatorname{Pr}_j(r;G_\epsilon,\theta)}\,dG_j(r)}.
\]
The normalizing denominator equals
\[
 1-\epsilon\int_{\mathcal P_j}
 \frac{D_G\operatorname{Pr}_j(r;G,\theta)[J_Gb]}
      {\operatorname{Pr}_j(r;G,\theta)}\,dG_j(r)
 +O(\epsilon^2).
\]
Expanding the ratio consequently yields
\[
\begin{split}
 \frac{d(TG_\epsilon)_j}{dG_j}(p)
 &=1+\epsilon\bigg\{
   -\frac{D_G\operatorname{Pr}_j(p;G,\theta)[J_Gb]}
         {\operatorname{Pr}_j(p;G,\theta)}\\
 &\hspace{1.8cm}+\int_{\mathcal P_j}
   \frac{D_G\operatorname{Pr}_j(r;G,\theta)[J_Gb]}
        {\operatorname{Pr}_j(r;G,\theta)}\,dG_j(r)\bigg\}
   +O(\epsilon^2)\\
 &=1+\epsilon(\mathcal L_GJ_Gb)_j(p)+O(\epsilon^2).
\end{split}
\]
Thus the output relative-density direction is $\mathcal L_GJ_Gb$,
and the centering term makes it mean zero under $G_j$.
More precisely,
\[
 \max_j\left\|
 \frac{d(TG_\epsilon)_j}{dG_j}
 -1-\epsilon(\mathcal L_GJ_Gb)_j
 \right\|_{L^\infty(G_j)}=o(|\epsilon|).
\]
These are Radon--Nikodym densities relative to $G_j$, not Lebesgue
densities. Taking logarithms preserves this first-order expansion.
For bounded $q,r$,
\[
 |\operatorname{osc}_{G_j}(q+r)-\operatorname{osc}_{G_j}(q)|
 \leq2\|r\|_{L^\infty(G_j)},
 \qquad
 \operatorname{osc}_{G_j}(\epsilon q)
 =|\epsilon|\operatorname{osc}_{G_j}(q).
\]
Hence
\[
\begin{split}
 D(TG_\epsilon,G)
 &=\max_j\operatorname{osc}_{G_j}
   \left(\log\frac{d(TG_\epsilon)_j}{dG_j}\right)\\
 &=|\epsilon|\|\mathcal L_GJ_Gb\|_{\mathrm{osc}}
   +o(|\epsilon|).
\end{split}
\]
Applying $D(TG_\epsilon,TG)\leq\bar\rho D(G_\epsilon,G)$ and
$TG=G$, dividing by $|\epsilon|$, and letting $\epsilon\to0$ gives
\begin{equation}\label{eq:tangent_hilbert_contraction}
 \|\mathcal L_GJ_Gb\|_{\mathrm{osc}}
 \leq\bar\rho\|b\|_{\mathrm{osc}}.
\end{equation}
This argument uses a first-order expansion of the distance; it does
not differentiate an essential supremum or require differentiability
of the metric itself.

For any $v\in\mathbb D$ and $m\geq1$,
\[
 \mathcal A^m v
 =J_G(\mathcal L_GJ_G)^{m-1}\mathcal L_Gv.
\]
Combining \eqref{eq:smoothing_bound} and
\eqref{eq:tangent_hilbert_contraction} gives
\begin{equation}\label{eq:powers_A}
 \|\mathcal A^m v\|_\infty
 \leq C\bar\rho^{m-1}\|v\|_\infty.
\end{equation}
The Neumann series therefore converges in operator norm, with
\begin{equation}\label{eq:cdf_resolvent}
 (I-\mathcal A)^{-1}=\sum_{m=0}^{\infty}\mathcal A^m,
 \qquad
 \|(I-\mathcal A)^{-1}\|\leq1+\frac{C}{1-\bar\rho}.
\end{equation}
Indeed multiplication of each finite partial sum by $I-\mathcal A$
gives $I-\mathcal A^{m+1}$, and the last term tends to zero in
operator norm. This proves both left and right invertibility, not
merely uniqueness of a formal derivative.

\emph{Step 5: Implicit differentiation and uniform bounds.}
Apply the Banach-space implicit function theorem to the extension of
\[
 \Psi-T(\Psi;G^{s},\theta)=0.
\]
After translation by a fixed probability triple, this is a map from
an open subset of
$\mathbb D\times\mathbb D\times\mathbb R^{d_\theta}$ into
$\mathbb D$. Step 1 gives differentiability on this open set, and
Step 4 gives the bounded inverse of its derivative in $\Psi$.
Step 2 identifies the local implicit solution with the original
inverse at nearby probability inputs. Differentiation proves
\eqref{eq:inverse_derivative}.

The inverse bound and local second-derivative bounds are uniform.
Consequently, a common small ball can be chosen on which the
linearization-based implicit-function construction is contractive.
The input tolerance can also be chosen uniformly, since Step 2
places nearby probability solutions inside this ball. This proves
\eqref{eq:cdf_inverse_lipschitz} with a common constant.

For input directions $(\eta_1,a_1)$ and $(\eta_2,a_2)$, twice
differentiating the fixed-point identity yields
\[
\begin{split}
 &(I-\mathcal A)D^2\mathcal R
       [(\eta_1,a_1),(\eta_2,a_2)]\\
 &\quad=D^2T\bigl[
  (D\mathcal R[\eta_1,a_1],\eta_1,a_1),
  (D\mathcal R[\eta_2,a_2],\eta_2,a_2)\bigr].
\end{split}
\]
The arguments on the right correspond, in order, to
$(\Psi,G^{s},\theta)$. Repeating the differentiation for pure CDF
directions and for mixed directions with at most two parameter
derivatives proves the claimed third-order mixed bounds.
Each new derivative is obtained by applying
$(I-\mathcal A)^{-1}$ to bounded derivatives of $T$ and already
bounded lower-order derivatives of $\mathcal R$.
No third derivative solely in $\theta$ is required.
The bounded second derivative gives a uniformly quadratic Taylor
remainder. Pointwise constructions and their limits preserve
measurability in $(x,x^\ast)$, and the common bounds yield the
conditional-collection statements under \eqref{eq:cdf_norm}.
\end{proof}

\begin{lem}\label{lem: Q_n^* uniform converge}
Under Assumption \ref{asmp:compact},
$\sup_{\theta\in\Theta}|\hat Q_n(\theta)-Q_0(\theta)|=o_p(1)$,
and $Q_0$ is continuous.
\end{lem}
\begin{proof}
For every choice $j=0,\ldots,J$, uniform positivity gives
\[
 \varepsilon\leq\pi_j(x,x^\ast;\theta,G^{s})\leq1.
\]
The product integral has bounded CDF derivatives, so Lemma
\ref{lem: equicontinuity} and the chain rule give
\[
\begin{split}
 &\sup_{x,x^\ast,\theta,\,0\leq j\leq J}
 |\pi_j(x,x^\ast;\theta,G^{s})
       -\pi_j(x,x^\ast;\theta,G^{s}_0)|\\
 &\qquad\leq C\|G^{s}-G^{s}_0\|_\infty
\end{split}
\]
for sufficiently close CDF inputs. Splitting the difference in
\eqref{eq:individual_criterion} into the change in log probability
and the change in the mixing measure gives
\begin{equation}\label{eq:criterion_nuisance_lipschitz}
 \sup_{\omega,\theta}
 |\ell(\omega;\theta,h)-\ell(\omega;\theta,h_0)|
 \leq C\|h-h_0\|_{\mathcal H}.
\end{equation}
In particular,
\[
\begin{split}
 \sup_\theta|\hat Q_n(\theta)-Q_0(\theta)|
 &\leq\sup_{\omega,\theta}
 |\ell(\omega;\theta,\hat h)-\ell(\omega;\theta,h_0)|\\
 &\quad+\sup_\theta\left|
 \frac1n\sum_{i=1}^n\ell(\omega_i;\theta,h_0)
 -\mathbb E_\omega\ell(\omega;\theta,h_0)\right|.
\end{split}
\]
The first term is $o_p(1)$ by first-step consistency. For fixed $h_0$,
the criterion is continuous in $\theta$ and bounded in absolute
value by $|\log\varepsilon|$. Compactness and i.i.d. sampling allow
application of \citet[Lemma 2.4]{newey1994large}; this controls the
second term and proves continuity of $Q_0$.
\end{proof}

\begin{proof}[Proof of Theorem \ref{thm: consistency}]
For every conditioning pair, the vectors
$(\pi_j(x,x^\ast;\theta,G^{s}_0))_{j=0}^J$
are positive probability vectors, including the outside choice.
Thus
\[
\begin{split}
 &Q_0(\theta_0)-Q_0(\theta)\\
 &\quad=\mathbb E_{x,x^\ast}\left[
 \sum_{j=0}^J\pi_j(x,x^\ast;\theta_0,G^{s}_0)
 \log\frac{\pi_j(x,x^\ast;\theta_0,G^{s}_0)}
               {\pi_j(x,x^\ast;\theta,G^{s}_0)}
 \right]\geq0.
\end{split}
\]
The integrand is a Kullback--Leibler divergence. Equality requires
agreement of the complete probability vectors almost surely under
the joint distribution of $(x,x^\ast)$, and hence $\theta=\theta_0$
by Assumption \ref{asmp: identification}. The outside term is essential
when outside choices occur; it cannot be discarded while retaining
the same unnormalized inside log probabilities.

Lemma \ref{lem: Q_n^* uniform converge}, compactness, and
\citet[Theorem 2.1]{newey1994large} imply consistency of the exact
maximizer. For an approximate maximizer, the population criterion
shortfall is at most twice the uniform criterion approximation
error plus the sample optimization gap, and is therefore $o_p(1)$.
Continuity and uniqueness on compact $\Theta$ then give the same
parameter consistency.
On an event with probability tending to one, the CDF and parameter
inputs lie in the common neighborhood of Lemma
\ref{lem: equicontinuity}, so
\[
 \|\hat G-G_0\|_\infty
 \leq C\{\|\widehat{G^{s}}-G^{s}_0\|_\infty
             +\|\hat\theta-\theta_0\|\}=o_p(1).
\]
This uses CDF distance, not the potentially infinite Hilbert distance
between estimated and population distributions.
\end{proof}

\subsection{Proof of Theorem \ref{thm: normal}}
\label{sec:proof_4}

We first make the nuisance derivative in
\eqref{eq:semiparametric_score} explicit. Fix $(x,x^\ast)$
temporarily and suppress it where this causes no ambiguity.
For a selected-CDF direction $\eta$, put
$v=D_{G^{s}}\mathcal R[\eta]=(I-\mathcal A)^{-1}\mathcal B\eta$.
Differentiating the product integral in \eqref{eq: prob j} gives,
for every choice $j=0,1,\ldots,J$,
\begin{equation}\label{eq:choice_cdf_derivative}
 D_{G^{s}}\pi_j[\eta]
 =\sum_{k=1}^J\int f_j(\boldsymbol p;x,x^\ast,\theta)
    \,dv_k(p_k)\prod_{\ell\ne k}dG_\ell(p_\ell).
\end{equation}
Here $G=\mathcal R(G^{s},\theta;x,x^\ast)$, and the derivatives
include the effect of changing the recovered inside distributions.
Similarly, for a parameter direction $a$, put
$v^a=D_\theta\mathcal R[a]=(I-\mathcal A)^{-1}D_\theta T[a]$.
Then
\begin{equation}\label{eq:choice_parameter_derivative}
\begin{split}
 D_\theta\pi_j[a]
 &=\int\partial_\theta f_j(\boldsymbol p;x,x^\ast,\theta)a
                                     \prod_{k=1}^JdG_k(p_k)\\
 &\quad+\sum_{k=1}^J\int f_j(\boldsymbol p;x,x^\ast,\theta)
        \,dv^a_k(p_k)\prod_{\ell\ne k}dG_\ell(p_\ell).
\end{split}
\end{equation}
The sum over $k$ always contains exactly $J$ price marginals,
including when $j=0$. There is no price direction for the outside
option.

For a nuisance direction $(\eta,\nu)$, where
$\nu(\cdot\mid x,j)$ is a signed mixing-measure perturbation of
mass zero, the derivative at the truth is
\begin{equation}\label{eq:explicit_score_derivative}
\begin{split}
 &D_h\mathfrak g(\omega;\theta_0,h_0)[\eta,\nu]\\
 &=\int_{X^\ast}
   \left.D_{G^{s}}\nabla_\theta\log
       \pi_j(x,x^\ast;\theta,G^{s})
       [\eta(x,x^\ast)]\right|_{\theta=\theta_0,\,G^{s}=G^{s}_0}
       \,dH_0(x^\ast\mid x,j)\\
 &\quad+\int_{X^\ast}
       \left.\nabla_\theta\log
       \pi_j(x,x^\ast;\theta,G^{s}_0)
       \right|_{\theta=\theta_0}
       \,d\nu(x^\ast\mid x,j).
\end{split}
\end{equation}
For fixed arguments, the derivative in the first integrand can be
computed as
\[
 D_{G^{s}}\nabla_\theta\log\pi_j[\eta]
 =\nabla_\theta\left(
 \frac{D_{G^{s}}\pi_j[\eta]}
      {\pi_j}\right).
\]
These formulas use \eqref{eq:inverse_derivative},
\eqref{eq:choice_cdf_derivative}, and the mixed derivatives in
Lemma \ref{lem: equicontinuity}, not a closed-form inverse.

We next prove the two bounds needed for the asymptotic expansion.
The criterion is the individual second-step objective
$\ell(\omega;\theta,h)$ in \eqref{eq:individual_criterion}, with
$h=(G^{s},H)$, and its score is
\[
 \mathfrak g(\omega;\theta,h)
 =\int_{X^\ast}\nabla_\theta\log
   \pi_j(x,x^\ast;\theta,G^{s})
   \,dH(x^\ast\mid x,j).
\]
Choose a sufficiently small convex neighborhood $\mathscr N$ of
$\theta_0$ with closure contained in $\operatorname{int}(\Theta)$,
and a sufficiently small convex neighborhood $\mathscr U$ of
$h_0$ in the extended nuisance space. The local extensions and
all derivative bounds are uniform on these neighborhoods. For the
mixing measures in a neighborhood of radius $\delta$,
\[
 \sup_{x,j}\|H(\cdot\mid x,j)\|_{\mathrm{TV}}\leq1+\delta.
\]

\emph{Bounds for the integrand derivatives.}
The map
\[
 (G,\theta)\longmapsto
 \int_{\mathcal P}f_j(\boldsymbol p;x,x^\ast,\theta)
               \prod_{k=1}^JdG_k(p_k)
\]
has bounded mixed derivatives of the required orders by the price
smoothness assumption and the extended Stieltjes integrals.
Composing with $G=\mathcal R(G^{s},\theta;x,x^\ast)$ gives
bounded mixed derivatives of $\pi_j$ of total order
at most three, with at most two parameter derivatives. At probability
inputs, $\pi_j\geq\varepsilon$ for every
$j\in\mathcal J$, including $0$. Its extension is at least
$\varepsilon/2$ after shrinking $\mathscr U$ if necessary.

In the next displays suppress the common arguments
$(x,x^\ast;\theta,G^{s})$, without introducing another notation
for the choice probability. The chain rule gives
\[
\begin{split}
 D_{G^{s}}\nabla_\theta\log\pi_j[\eta]
 &=\frac{\nabla_\theta D_{G^{s}}\pi_j[\eta]}
        {\pi_j}\\
 &\quad-\frac{\nabla_\theta\pi_j}
              {\pi_j^2}
              D_{G^{s}}\pi_j[\eta],
\end{split}
\]
\[
\begin{split}
 &D_{G^{s}}^2\nabla_\theta\log
       \pi_j[\eta_1,\eta_2]\\
 &\quad=\nabla_\theta\left\{
 \frac{D_{G^{s}}^2\pi_j[\eta_1,\eta_2]}
      {\pi_j}
 -\frac{D_{G^{s}}\pi_j[\eta_1]
        D_{G^{s}}\pi_j[\eta_2]}
       {\pi_j^2}\right\},
\end{split}
\]
and
\[
 D_{G^{s}}\nabla_\theta^2\log\pi_j[\eta]
 =\nabla_\theta^2\left\{
 \frac{D_{G^{s}}\pi_j[\eta]}
      {\pi_j}\right\}.
\]
The highest derivatives required are
$\nabla_\theta D_{G^{s}}^2\pi_j$ and
$\nabla_\theta^2D_{G^{s}}\pi_j$.
All other terms involve lower-order derivatives and inverse powers
of $\pi_j$. Uniform positivity and the derivative
bounds consequently imply
\[
\begin{aligned}
 \|\nabla_\theta\log\pi_j\|
 +\|\nabla_\theta^2\log\pi_j\|&\leq C,\\
 \|D_{G^{s}}\nabla_\theta\log\pi_j[\eta]\|
 &\leq C\|\eta\|_\infty,\\
 \|D_{G^{s}}^2\nabla_\theta\log
           \pi_j[\eta_1,\eta_2]\|
 &\leq C\|\eta_1\|_\infty\|\eta_2\|_\infty,\\
 \|D_{G^{s}}\nabla_\theta^2\log\pi_j[\eta]\|
 &\leq C\|\eta\|_\infty.
\end{aligned}
\]
The bounds are uniform over the conditioning pairs, choice indices,
and indicated neighborhoods. A direction in an integrand is
evaluated at $(x,x^\ast)$; its local norm is bounded by the
conditional-collection norm \eqref{eq:cdf_norm}.

\emph{Nuisance derivatives of the score.}
Take $v_r=(\eta_r,\nu_r)$, $r=1,2$. Differentiating the integral
and its mixing measure gives
\[
\begin{split}
 D_h^2\mathfrak g(\omega;\theta,h)[v_1,v_2]
 &=\int_{X^\ast}D_{G^{s}}^2\nabla_\theta\log
       \pi_j[\eta_1,\eta_2] \,dH\\
 &\quad+\int_{X^\ast}D_{G^{s}}\nabla_\theta\log
       \pi_j[\eta_1] \,d\nu_2\\
 &\quad+\int_{X^\ast}D_{G^{s}}\nabla_\theta\log
       \pi_j[\eta_2] \,d\nu_1,
\end{split}
\]
where the integrands use $(x,x^\ast;\theta,G^{s})$ and the measures
are conditional on $(x,j)$. There is no term with two derivatives
of the mixing measure, because $\ell$ is linear in $H$.
The two cross terms need not vanish. Similarly, for $v=(\eta,\nu)$,
\[
\begin{split}
 D_h\nabla_\theta\mathfrak g(\omega;\theta,h)[v]
 &=\int_{X^\ast}D_{G^{s}}\nabla_\theta^2\log
        \pi_j[\eta] \,dH\\
 &\quad+\int_{X^\ast}\nabla_\theta^2\log
        \pi_j\,d\nu.
\end{split}
\]
For a bounded vector- or matrix-valued integrand $a$,
\[
 \left\|\int a(x^\ast)\,d\nu(x^\ast)\right\|
 \leq\sup_{x^\ast}\|a(x^\ast)\|\,\|\nu\|_{\mathrm{TV}}.
\]
Combining this with the preceding derivative bounds gives
\[
\begin{split}
 &\|D_h^2\mathfrak g(\omega;\theta,h)[v_1,v_2]\|\\
 &\quad\leq C\bigg\{
 \|\eta_1\|_\infty\|\eta_2\|_\infty
 +\|\eta_1\|_\infty\sup_{x,k}\|\nu_2(\cdot\mid x,k)\|_{\mathrm{TV}}\\
 &\hspace{2.2cm}
 +\|\eta_2\|_\infty\sup_{x,k}\|\nu_1(\cdot\mid x,k)\|_{\mathrm{TV}}
 \bigg\}\\
 &\quad\leq C\|v_1\|_{\mathcal H}\|v_2\|_{\mathcal H},
\end{split}
\]
and
\[
 \|D_h\nabla_\theta\mathfrak g(\omega;\theta,h)[v]\|
 \leq C\|v\|_{\mathcal H}.
\]
These constants are uniform over $\omega$, $\theta\in\mathscr N$,
and $h\in\mathscr U$.

\emph{Taylor expansion along a nuisance path.}
Let $\Delta h=h-h_0$ and $h_t=h_0+t\Delta h$, $0\leq t\leq1$,
with the segment contained in $\mathscr U$.
Taylor's formula with integral remainder gives
\[
\begin{split}
 &\mathfrak g(\omega;\theta_0,h)
 -\mathfrak g(\omega;\theta_0,h_0)
 -D_h\mathfrak g(\omega;\theta_0,h_0)[\Delta h]\\
 &\quad=\int_0^1(1-t)
 D_h^2\mathfrak g(\omega;\theta_0,h_t)
       [\Delta h,\Delta h] \,dt.
\end{split}
\]
The uniform bilinear bound implies
\begin{equation}\label{eq:score_quadratic_remainder}
\begin{split}
 &\sup_\omega\bigl\|
 \mathfrak g(\omega;\theta_0,h)-\mathfrak g(\omega;\theta_0,h_0)
 -D_h\mathfrak g(\omega;\theta_0,h_0)[h-h_0]\bigr\|\\
 &\qquad\leq C\|h-h_0\|_{\mathcal H}^2.
\end{split}
\end{equation}
The factor $\int_0^1(1-t)\,dt=1/2$ is absorbed into $C$.
For each fixed $\theta\in\mathscr N$, the fundamental theorem of
calculus similarly gives
\[
\begin{split}
 &\nabla_\theta\mathfrak g(\omega;\theta,h)
 -\nabla_\theta\mathfrak g(\omega;\theta,h_0)\\
 &\quad=\int_0^1D_h\nabla_\theta\mathfrak g
             (\omega;\theta,h_t)[\Delta h] \,dt.
\end{split}
\]
Thus
\begin{equation}\label{eq:score_jacobian_lipschitz}
 \sup_{\omega,\theta\in\mathscr N}
 \|\nabla_\theta\mathfrak g(\omega;\theta,h)
      -\nabla_\theta\mathfrak g(\omega;\theta,h_0)\|
 \leq C\|h-h_0\|_{\mathcal H}.
\end{equation}
Since $\mathfrak g=\nabla_\theta\ell$, the two derivative orders are
\[
 D_h^2\mathfrak g=D_h^2\nabla_\theta\ell,
 \qquad
 D_h\nabla_\theta\mathfrak g=D_h\nabla_\theta^2\ell.
\]
The first bound uses two nuisance derivatives and one parameter
derivative of the individual criterion; the second uses one nuisance
derivative and two parameter derivatives. The same calculations
bound the first two parameter derivatives of $\ell$ and give their
continuity in $\theta$; no pure third parameter derivative is needed.

\begin{proof}[Proof of Theorem \ref{thm: normal}]
Theorem \ref{thm: consistency} gives consistency. By iterated
expectations and the complete-choice identity,
\[
\begin{split}
 &\mathbb E_\omega\mathfrak g(\omega;\theta_0,h_0)\\
 &=\mathbb E_{x,x^\ast}\left[
 \sum_{k=0}^J\pi_k(x,x^\ast;\theta_0,G^{s}_0)
 \left.\nabla_\theta\log
 \pi_k(x,x^\ast;\theta,G^{s}_0)
 \right|_{\theta=\theta_0}\right]\\
 &=\mathbb E_{x,x^\ast}\left[
 \left.\nabla_\theta\sum_{k=0}^J
 \pi_k(x,x^\ast;\theta,G^{s}_0)
 \right|_{\theta=\theta_0}\right]=0.
\end{split}
\]
The outside contribution is included, whereas all price integrations
still involve only inside alternatives. Uniform derivative bounds
justify interchanging differentiation and expectation and imply
$\mathbb E\|\mathfrak g(\omega;\theta_0,h_0)\|^2<\infty$.

By \eqref{eq:score_quadratic_remainder} and Assumption
\ref{asmp: normal}(ii),
\begin{equation}\label{eq:score_expansion}
\begin{split}
 \sqrt n\,\mathbb P_n\mathfrak g(\theta_0,\hat h)
 &=\frac1{\sqrt n}\sum_{i=1}^n\mathfrak g(\omega_i;\theta_0,h_0)\\
 &\quad+\frac1{\sqrt n}\sum_{i=1}^n
 D_h\mathfrak g(\omega_i;\theta_0,h_0)[\hat h-h_0]+o_p(1).
\end{split}
\end{equation}
The remainder has norm at most
$C\sqrt n\|\hat h-h_0\|_{\mathcal H}^2=o_p(1)$.
Assumption \ref{asmp: normal}(iii) replaces the second sum by
\[
 \sqrt n\,\mathbb E_\omega
 \bigl[D_h\mathfrak g(\omega;\theta_0,h_0)[\hat h-h_0]\bigr]+o_p(1),
\]
where only the fresh observation is integrated out. Part (iv) then
gives
\begin{equation}\label{eq:score_influence_expansion}
 \sqrt n\,\mathbb P_n\mathfrak g(\theta_0,\hat h)
 =\frac1{\sqrt n}\sum_{i=1}^n
 \{\mathfrak g(\omega_i;\theta_0,h_0)+\varphi(\omega_i)\}
 +o_p(1).
\end{equation}
This is the score decomposition in
\citet[Theorem 8.1]{newey1994large}.

The Jacobian bound \eqref{eq:score_jacobian_lipschitz} permits
replacing $\hat h$ by $h_0$ in the sample Jacobian uniformly on a
small compact convex neighborhood $\mathscr N_0\subset\mathscr N$.
At $h_0$ the Jacobian is continuous in $\theta$ and uniformly bounded.
By \citet[Lemma 2.4]{newey1994large},
\[
 \sup_{\theta\in\mathscr N_0}
 \left\|\mathbb P_n\nabla_\theta\mathfrak g(\theta,\hat h)
 -\mathbb E_\omega\nabla_\theta
       \mathfrak g(\omega;\theta,h_0)\right\|=o_p(1).
\]
With probability tending to one, an exact maximizer is interior.
Use the integral mean-value expansion
\[
 \mathbb P_n\mathfrak g(\hat\theta,\hat h)
 =\mathbb P_n\mathfrak g(\theta_0,\hat h)
    +M_n(\hat\theta-\theta_0),
\]
where
\[
 M_n=\int_0^1\mathbb P_n\nabla_\theta\mathfrak g
  (\theta_0+t(\hat\theta-\theta_0),\hat h)\,dt
 \xrightarrow{p}M_0.
\]
The left-hand side is zero for an exact maximizer, and may be
$o_p(n^{-1/2})$ for the stated numerical estimator. Since $M_0$ is
nonsingular, solving this equation and using
\eqref{eq:score_influence_expansion} gives
\eqref{eq:theta_asymptotic_linear}. Its summands are i.i.d., mean zero,
and square integrable. The multivariate central limit theorem yields
\eqref{eq:semiparametric_variance}, as in
\citet[Theorem 8.2]{newey1994large}.
\end{proof}

\paragraph{Variance estimation.}
Suppose the influence estimates satisfy
\[
 \frac1n\sum_{i=1}^n
 \|\hat\varphi_i-\varphi(\omega_i)\|^2=o_p(1).
\]
Let $\hat M=\mathbb P_n\nabla_\theta\mathfrak g(\hat\theta,\hat h)$
and define
\[
\begin{split}
 \hat\Omega
 &=\frac1n\sum_{i=1}^n
 \{\mathfrak g(\omega_i;\hat\theta,\hat h)+\hat\varphi_i\}
 \{\mathfrak g(\omega_i;\hat\theta,\hat h)+\hat\varphi_i\}',\\
 \hat V&=\hat M^{-1}\hat\Omega(\hat M^{-1})'.
\end{split}
\]
The Jacobian argument gives $\hat M\xrightarrow{p}M_0$.
Consistency and the uniform score derivative bounds imply that the
sample mean squared difference between the estimated summands and
$\mathfrak g(\omega_i;\theta_0,h_0)+\varphi(\omega_i)$ is $o_p(1)$.
Cauchy--Schwarz and the law of large numbers give
$\hat\Omega\xrightarrow{p}\Omega_0$ and $\hat V\xrightarrow{p}V$.
Centering the estimated summands by their sample mean yields the
same limit. Constructing $\hat\varphi_i$ remains a first-step-specific
problem and is not supplied by contraction alone.

\begin{proof}[Proof of Corollary \ref{col:cdf_asymptotics}]
The bounded second derivatives in Lemma \ref{lem: equicontinuity}
give, with derivatives evaluated at $(G^{s}_0,\theta_0)$,
\begin{equation}\label{eq:offered_cdf_expansion}
\begin{split}
 \hat G-G_0
 &=D_{G^{s}}\mathcal R[\widehat{G^{s}}-G^{s}_0]
   +D_\theta\mathcal R[\hat\theta-\theta_0]+R_n,\\
 \|R_n\|_\infty
 &\leq C\bigl(\|\widehat{G^{s}}-G^{s}_0\|_\infty
                      +\|\hat\theta-\theta_0\|\bigr)^2
\end{split}
\end{equation}
on an event with probability tending to one.
At each conditioning pair,
\[
 D_{G^{s}}\mathcal R=(I-\mathcal A)^{-1}\mathcal B,
 \qquad
 D_\theta\mathcal R=(I-\mathcal A)^{-1}D_\theta T.
\]
Part (i) follows from the local CDF Lipschitz bound and Theorem
\ref{thm: normal}. Under \eqref{eq:first_step_functional_clt}, both
input errors are $O_p(n^{-1/2})$, so
$\sqrt n\|R_n\|_\infty=o_p(1)$. The continuous mapping theorem for
the bounded linear derivative gives \eqref{eq:offered_functional_clt}.
At a fixed $(x,x^\ast)$, the same proof uses the first-step CDF rate
or joint functional convergence at that pair, along with the already
established parameter result.
\end{proof}

\subsection{Finite-Iteration Approximation}
\label{sec:proof_numerical}

Fix the selected probability measures, $(x,x^\ast)$, and an initial
probability collection $\Psi$ independent of $\theta$; suppress the
fixed conditioning arguments below. Put
$G^{(m)}_\theta=T^m_{\theta,G^{s}}\Psi$ and
$G_\theta=\mathcal R(G^{s},\theta)$.
The uniform CDF bound \eqref{eq:uniform_cdf_iteration}, bounded price
derivatives of $f_j$, and positivity give an error at most
$C\bar\rho^{m-1}$ in each log choice probability, for $j=0,\ldots,J$.
For $j=0$ the same argument integrates $f_0$ over the inside prices.
Integrating against the estimated mixing probability measures and
averaging over observations proves the criterion bound.

For the score bound, use logarithmic coordinates relative to
$G^{s}$ after the first update. Each coordinate is a class
$[w_j]\in L^\infty(G^{s}_j)/\mathbb R$, with norm
$\operatorname{osc}_{G^{s}_j}w_j$, representing
\[
 dG_j(w_j)=\frac{e^{w_j}\,dG^{s}_j}{\int e^{w_j}\,dG^{s}_j}.
\]
The product of these quotient spaces is Banach under the maximum
oscillation norm. The update is
\[
 \mathcal S_\theta(w)_j
 =[-\log\operatorname{Pr}_j(\cdot;G(w),\theta)],
 \qquad j=1,\ldots,J.
\]
It is well defined independently of the additive constants in the
representatives. The logarithmic chart is an isometry for the
projective distance, and Theorem \ref{thm: contraction} implies
$\|D_w\mathcal S_\theta\|\leq\bar\rho$.
The first update has representative
$w^{(1)}_{\theta,j}=-\log\operatorname{Pr}_j(\cdot;\Psi,\theta)$.
Every iterate and the fixed point $w_\theta$ has a representative
between $0$ and $\log(1/\varepsilon)$.

The parameter and second mixed derivatives of $\mathcal S_\theta$
are uniformly bounded. One may choose the centered representatives
$\int w_j\,dG^{s}_j=0$ to perform this calculation. Derivatives of
normalized exponential weights are finite products of bounded
centered directions, and the selection-probability denominators
are at least $\varepsilon$. The bounds are independent of atoms
and support gaps in $G^{s}_j$.
Contraction gives
\[
 \|w^{(m)}_\theta-w_\theta\|\leq C\bar\rho^{m-1},
\]
and differentiation of the fixed-point equation gives
$\|\partial_\theta w_\theta\|\leq C/(1-\bar\rho)$.
Define only the derivative error
\[
 d_m=\|\partial_\theta w^{(m)}_\theta
                    -\partial_\theta w_\theta\|.
\]
Differentiating the iteration and subtracting the fixed-point
derivative identity yields
\[
 d_{m+1}\leq\bar\rho d_m
              +C\|w^{(m)}_\theta-w_\theta\|
 \leq\bar\rho d_m+C\bar\rho^{m-1}.
\]
Since $d_1$ is uniformly bounded and $\bar\rho\in(0,1)$,
\[
 d_m\leq C m\bar\rho^{m-1},
\]
after enlarging a constant depending on $\bar\rho$.

For each $j=0,\ldots,J$, the function
\[
 (w,\theta)\longmapsto
 \log\int_{\mathcal P}f_j(\boldsymbol p;\theta)
                         \prod_{k=1}^JdG_k(w_k)(p_k)
\]
has uniformly bounded first and second derivatives on the relevant
set. The parameter derivatives along an iterate and the fixed point
therefore differ by at most
$C\{\|w^{(m)}_\theta-w_\theta\|+d_m\}$.
Integrating against the estimated latent-type weights, including
those for outside choices, proves
\[
 \sup_{\theta\in\mathscr N}
 \|\nabla_\theta\hat Q_{n,m}(\theta)-\nabla_\theta\hat Q_n(\theta)\|
 \leq C m\bar\rho^{m-1}.
\]

If $m(n)\to\infty$, the approximate criterion converges uniformly
to $Q_0$, so its global or $o_p(1)$-approximate maximizer is consistent.
For asymptotic equivalence, use the score bound and the nonsingular
limiting Jacobian from Theorem \ref{thm: normal}.
An integral mean-value expansion of the exact score between the
consistent exact and approximate maximizers gives
\[
 \|\hat\theta_{m(n)}-\hat\theta\|
 =O_p\bigl(m(n)\bar\rho^{m(n)-1}\bigr).
\]
The expansion is performed inside a common convex neighborhood of
$\theta_0$, where the averaged exact Jacobian converges to $M_0$.
For numerical optimizers with $o_p(n^{-1/2})$ score residuals, add
an $o_p(n^{-1/2})$ term. Uniform CDF stability and the iteration
bound also give
\[
 \|\hat T^{m(n)}_{\hat\theta_{m(n)},\hat h}\Psi
       -\hat T^\infty_{\hat\theta,\hat h}\Psi\|_\infty
 =O_p\bigl(m(n)\bar\rho^{m(n)-1}\bigr),
\]
with the same possible numerical residual. Condition
\eqref{eq:numerical_rate} makes both differences negligible at the
$\sqrt n$ scale. The strict lower bound on $m(n)/\log n$ implies
this condition: choose a smaller constant still exceeding
$1/(2\log(1/\bar\rho))$, and use eventual monotonicity of
$m\bar\rho^{m-1}$ to bound it at that constant times $\log n$.

\section{Additional Tables}
\label{sec:tables_app}

\begin{table}[htpb!]
    \centering
    \caption{Simulation Results for Utility Parameters: Removing the Excluded Variable} 
\scalebox{0.9}{\begin{tabular}{lccc|ccc}
\hline \hline 
 & \multicolumn{6}{c}{DGP 1} \\ \cline{2-7} 
  & \multicolumn{3}{c|}{$N=2000$}  & \multicolumn{3}{c}{$N=5000$}  \\ \cline{2-7} 
  & Bias & Std. Dev. & RMSE & Bias & Std. Dev. & RMSE \\ \cline{2-7} 
 $\gamma$ & -0.0658 & 0.1743 & 0.1861 & -0.0358 & 0.1159 & 0.1212 \\ 
$\kappa$ & -0.0223 & 0.0527 & 0.0572 & -0.0081 & 0.0351 & 0.0360 \\ 
$\xi_2$ & -0.0162 & 0.0458 & 0.0485 & -0.0090 & 0.0309 & 0.0322 \\ \hline 
 & \multicolumn{6}{c}{DGP 2} \\ \cline{2-7} 
  & \multicolumn{3}{c|}{$N=2000$}  & \multicolumn{3}{c}{$N=5000$}  \\ \cline{2-7} 
  & Bias & Std. Dev. & RMSE & Bias & Std. Dev. & RMSE \\ \cline{2-7} 
 $\gamma$ & -0.0708 & 0.1636 & 0.1781 & -0.0452 & 0.1091 & 0.1180 \\ 
$\kappa$ & -0.0192 & 0.0546 & 0.0579 & -0.0070 & 0.0342 & 0.0349 \\ 
$\xi_2$ & -0.0130 & 0.0381 & 0.0402 & -0.0083 & 0.0252 & 0.0265 \\ \hline 
 & \multicolumn{6}{c}{DGP 3} \\ \cline{2-7} 
  & \multicolumn{3}{c|}{$N=2000$}  & \multicolumn{3}{c}{$N=5000$}  \\ \cline{2-7} 
  & Bias & Std. Dev. & RMSE & Bias & Std. Dev. & RMSE \\ \cline{2-7} 
 $\gamma$ & -0.0623 & 0.2467 & 0.2542 & -0.0262 & 0.1634 & 0.1653 \\ 
$\kappa$ & -0.0151 & 0.0516 & 0.0537 & -0.0060 & 0.0345 & 0.0350 \\ 
$\xi_2$ & -0.0050 & 0.0313 & 0.0317 & -0.0023 & 0.0201 & 0.0202 \\ \hline 
 & \multicolumn{6}{c}{DGP 4} \\ \cline{2-7} 
  & \multicolumn{3}{c|}{$N=2000$}  & \multicolumn{3}{c}{$N=5000$}  \\ \cline{2-7} 
  & Bias & Std. Dev. & RMSE & Bias & Std. Dev. & RMSE \\ \cline{2-7} 
 $\gamma$ & 0.0453 & 0.8115 & 0.8119 & 0.0062 & 0.5276 & 0.5271 \\ 
$\kappa$ & -0.0138 & 0.0501 & 0.0519 & -0.0068 & 0.0323 & 0.0330 \\ 
$\xi_2$ & -0.0001 & 0.0308 & 0.0308 & -0.0003 & 0.0189 & 0.0189 \\ \hline 
 \hline 
\end{tabular}}

\label{table:est_theta_noz}
\vspace{0.3cm} \\
		\justifying \footnotesize 
	\noindent Note: In these specifications, we remove the excluded variable $x_{i1}$ from the selection function, so the parameter $\beta$ in $u_{i1}$ is not estimated. 
\end{table}

\begin{table}[htpb!]
    \centering
    \caption{Simulation Results for CDF of $log(price)$: Removing the Excluded Variable} 
\scalebox{0.9}{\begin{tabular}{lcc|cc}
\hline \hline 
 & \multicolumn{4}{c}{DGP 1} \\ \cline{2-5} 
  & \multicolumn{2}{c|}{$j=1$}  & \multicolumn{2}{c}{$j=2$}  \\ \cline{2-5} 
  & IBias$^2$ & IMSE & IBias$^2$ & IMSE \\ \cline{2-5} 
 $x_{i2}=0$ & 0.0004 & 0.0019 & 0.0001 & 0.0006 \\ 
$x_{i2}=0.25$ & 0.0002 & 0.0016 & 0.0001 & 0.0007 \\ 
$x_{i2}=0.5$ & 0.0001 & 0.0014 & 0.0001 & 0.0008 \\ 
$x_{i2}=0.75$ & 0.0001 & 0.0013 & 0.0002 & 0.0009 \\ 
$x_{i2}=1$ & 0.0001 & 0.0011 & 0.0001 & 0.0009 \\ \hline 
 & \multicolumn{4}{c}{DGP 2} \\ \cline{2-5} 
  & \multicolumn{2}{c|}{$j=1$}  & \multicolumn{2}{c}{$j=2$}  \\ \cline{2-5} 
  & IBias$^2$ & IMSE & IBias$^2$ & IMSE \\ \cline{2-5} 
 $x_{i2}=0$ & 0.0003 & 0.0019 & 0.0001 & 0.0006 \\ 
$x_{i2}=0.25$ & 0.0003 & 0.0018 & 0.0001 & 0.0007 \\ 
$x_{i2}=0.5$ & 0.0002 & 0.0016 & 0.0001 & 0.0007 \\ 
$x_{i2}=0.75$ & 0.0001 & 0.0015 & 0.0002 & 0.0008 \\ 
$x_{i2}=1$ & 0.0001 & 0.0011 & 0.0001 & 0.0008 \\ \hline 
 & \multicolumn{4}{c}{DGP 3} \\ \cline{2-5} 
  & \multicolumn{2}{c|}{$j=1$}  & \multicolumn{2}{c}{$j=2$}  \\ \cline{2-5} 
  & IBias$^2$ & IMSE & IBias$^2$ & IMSE \\ \cline{2-5} 
 $x_{i2}=0$ & 0.0007 & 0.0023 & 0.0029 & 0.0033 \\ 
$x_{i2}=0.25$ & 0.0002 & 0.0015 & 0.0005 & 0.0010 \\ 
$x_{i2}=0.5$ & 0.0001 & 0.0015 & 0.0002 & 0.0007 \\ 
$x_{i2}=0.75$ & 0.0001 & 0.0015 & 0.0002 & 0.0008 \\ 
$x_{i2}=1$ & 0.0001 & 0.0014 & 0.0001 & 0.0009 \\ \hline 
 & \multicolumn{4}{c}{DGP 4} \\ \cline{2-5} 
  & \multicolumn{2}{c|}{$j=1$}  & \multicolumn{2}{c}{$j=2$}  \\ \cline{2-5} 
  & IBias$^2$ & IMSE & IBias$^2$ & IMSE \\ \cline{2-5} 
 $x_{i2}=0$ & 0.0014 & 0.0025 & 0.0011 & 0.0016 \\ 
$x_{i2}=0.25$ & 0.0013 & 0.0023 & 0.0009 & 0.0014 \\ 
$x_{i2}=0.5$ & 0.0012 & 0.0022 & 0.0005 & 0.0010 \\ 
$x_{i2}=0.75$ & 0.0011 & 0.0023 & 0.0002 & 0.0008 \\ 
$x_{i2}=1$ & 0.0005 & 0.0019 & 0.0001 & 0.0006 \\ \hline 
 \hline 
\end{tabular}}

\label{table:est_price_noz}
\vspace{0.3cm} \\
		\justifying \footnotesize 
		\noindent Note: In these specifications, we remove the excluded variable from the selection function.
	The $\text{IBias}^2$ of a function $h$ is calculated as follows. Let $\hat{h}_r$ be the estimate of $h$ from the $r$-th simulated dataset, and $\bar{h}(p) = \frac{1}{R}\sum_{r=1}^{R} \hat{h}_r(p)$ be the point-wise average over $R$ simulations. The integrated squared bias is calculated by numerically integrating the point-wise squared bias $(\bar{h}(p)-h(p))^2$ over the distribution of $p$. The integrated MSE is computed in a similar way. The values reported in each row correspond to the price distributions conditional on a given value of $x_{i2}$. The results shown in this table are based on a 500 Monte Carlo replications with a sample size of 2,000. Corresponding results for a sample size of 5,000 are available upon request.
\end{table}

\begin{table}[htpb!]
    \centering
    \caption{Simulation Results for Utility Parameters: Misspecifying the Selection Function} 
\scalebox{0.9}{\begin{tabular}{lccc|ccc}
\hline \hline 
 & \multicolumn{6}{c}{DGP 1} \\ \cline{2-7} 
  & \multicolumn{3}{c|}{$N=2000$}  & \multicolumn{3}{c}{$N=5000$}  \\ \cline{2-7} 
  & Bias & Std. Dev. & RMSE & Bias & Std. Dev. & RMSE \\ \cline{2-7} 
 $\gamma$ & -0.0851 & 0.1829 & 0.2016 & -0.0540 & 0.1123 & 0.1245 \\ 
$\beta$ & 0.0009 & 0.0634 & 0.0633 & 0.0052 & 0.0378 & 0.0381 \\ 
$\kappa$ & -0.0210 & 0.0499 & 0.0541 & -0.0080 & 0.0356 & 0.0365 \\ 
$\xi_2$ & -0.0193 & 0.0596 & 0.0626 & -0.0083 & 0.0364 & 0.0373 \\ \hline 
 & \multicolumn{6}{c}{DGP 2} \\ \cline{2-7} 
  & \multicolumn{3}{c|}{$N=2000$}  & \multicolumn{3}{c}{$N=5000$}  \\ \cline{2-7} 
  & Bias & Std. Dev. & RMSE & Bias & Std. Dev. & RMSE \\ \cline{2-7} 
 $\gamma$ & -0.0834 & 0.1739 & 0.1927 & -0.0540 & 0.1103 & 0.1228 \\ 
$\beta$ & 0.0028 & 0.0642 & 0.0642 & 0.0062 & 0.0387 & 0.0391 \\ 
$\kappa$ & -0.0167 & 0.0497 & 0.0523 & -0.0054 & 0.0351 & 0.0355 \\ 
$\xi_2$ & -0.0116 & 0.0532 & 0.0544 & -0.0034 & 0.0333 & 0.0334 \\ \hline 
 & \multicolumn{6}{c}{DGP 3} \\ \cline{2-7} 
  & \multicolumn{3}{c|}{$N=2000$}  & \multicolumn{3}{c}{$N=5000$}  \\ \cline{2-7} 
  & Bias & Std. Dev. & RMSE & Bias & Std. Dev. & RMSE \\ \cline{2-7} 
 $\gamma$ & -0.0970 & 0.2606 & 0.2779 & -0.0569 & 0.1558 & 0.1657 \\ 
$\beta$ & 0.0053 & 0.0652 & 0.0653 & 0.0079 & 0.0374 & 0.0382 \\ 
$\kappa$ & -0.0179 & 0.0493 & 0.0524 & -0.0061 & 0.0336 & 0.0341 \\ 
$\xi_2$ & -0.0013 & 0.0474 & 0.0473 & 0.0030 & 0.0288 & 0.0289 \\ \hline 
 & \multicolumn{6}{c}{DGP 4} \\ \cline{2-7} 
  & \multicolumn{3}{c|}{$N=2000$}  & \multicolumn{3}{c}{$N=5000$}  \\ \cline{2-7} 
  & Bias & Std. Dev. & RMSE & Bias & Std. Dev. & RMSE \\ \cline{2-7} 
 $\gamma$ & 0.1083 & 0.8175 & 0.8239 & 0.0813 & 0.4889 & 0.4951 \\ 
$\beta$ & 0.0053 & 0.0624 & 0.0626 & 0.0086 & 0.0367 & 0.0376 \\ 
$\kappa$ & -0.0183 & 0.0480 & 0.0514 & -0.0091 & 0.0340 & 0.0351 \\ 
$\xi_2$ & 0.0047 & 0.0466 & 0.0468 & 0.0065 & 0.0289 & 0.0296 \\ \hline 
 \hline 
\end{tabular}}

\label{table:est_theta_mis}
\vspace{0.3cm} \\
		\justifying \footnotesize 
	\noindent Note: In these specifications, we misspecify the selection model in estimation, assuming that the error term $\varepsilon_{i}$ is drawn from $Logistic (0,1)$. For the utility parameters, we rescale the estimates by the scale parameter of the logit model to make them comparable to those in the original probit specification.
\end{table}

\begin{table}[htpb!]
    \centering
    \caption{Simulation Results for CDF of $log(price)$: Misspecifying the Selection Function} 
\scalebox{0.9}{\begin{tabular}{lcc|cc}
\hline \hline 
 & \multicolumn{4}{c}{DGP 1} \\ \cline{2-5} 
  & \multicolumn{2}{c|}{$j=1$}  & \multicolumn{2}{c}{$j=2$}  \\ \cline{2-5} 
  & IBias$^2$ & IMSE & IBias$^2$ & IMSE \\ \cline{2-5} 
 $x_{i2}=0$ & 0.0004 & 0.0016 & 0.0002 & 0.0009 \\ 
$x_{i2}=0.25$ & 0.0004 & 0.0015 & 0.0002 & 0.0009 \\ 
$x_{i2}=0.5$ & 0.0002 & 0.0012 & 0.0002 & 0.0010 \\ 
$x_{i2}=0.75$ & 0.0002 & 0.0010 & 0.0002 & 0.0011 \\ 
$x_{i2}=1$ & 0.0001 & 0.0010 & 0.0002 & 0.0012 \\ \hline 
 & \multicolumn{4}{c}{DGP 2} \\ \cline{2-5} 
  & \multicolumn{2}{c|}{$j=1$}  & \multicolumn{2}{c}{$j=2$}  \\ \cline{2-5} 
  & IBias$^2$ & IMSE & IBias$^2$ & IMSE \\ \cline{2-5} 
 $x_{i2}=0$ & 0.0005 & 0.0016 & 0.0002 & 0.0008 \\ 
$x_{i2}=0.25$ & 0.0005 & 0.0017 & 0.0002 & 0.0008 \\ 
$x_{i2}=0.5$ & 0.0003 & 0.0013 & 0.0002 & 0.0009 \\ 
$x_{i2}=0.75$ & 0.0002 & 0.0011 & 0.0002 & 0.0010 \\ 
$x_{i2}=1$ & 0.0001 & 0.0010 & 0.0003 & 0.0011 \\ \hline 
 & \multicolumn{4}{c}{DGP 3} \\ \cline{2-5} 
  & \multicolumn{2}{c|}{$j=1$}  & \multicolumn{2}{c}{$j=2$}  \\ \cline{2-5} 
  & IBias$^2$ & IMSE & IBias$^2$ & IMSE \\ \cline{2-5} 
 $x_{i2}=0$ & 0.0007 & 0.0019 & 0.0031 & 0.0036 \\ 
$x_{i2}=0.25$ & 0.0002 & 0.0013 & 0.0005 & 0.0011 \\ 
$x_{i2}=0.5$ & 0.0002 & 0.0013 & 0.0003 & 0.0010 \\ 
$x_{i2}=0.75$ & 0.0002 & 0.0012 & 0.0002 & 0.0009 \\ 
$x_{i2}=1$ & 0.0002 & 0.0013 & 0.0002 & 0.0011 \\ \hline 
 & \multicolumn{4}{c}{DGP 4} \\ \cline{2-5} 
  & \multicolumn{2}{c|}{$j=1$}  & \multicolumn{2}{c}{$j=2$}  \\ \cline{2-5} 
  & IBias$^2$ & IMSE & IBias$^2$ & IMSE \\ \cline{2-5} 
 $x_{i2}=0$ & 0.0014 & 0.0023 & 0.0011 & 0.0016 \\ 
$x_{i2}=0.25$ & 0.0014 & 0.0024 & 0.0009 & 0.0015 \\ 
$x_{i2}=0.5$ & 0.0011 & 0.0021 & 0.0005 & 0.0011 \\ 
$x_{i2}=0.75$ & 0.0012 & 0.0021 & 0.0003 & 0.0009 \\ 
$x_{i2}=1$ & 0.0005 & 0.0018 & 0.0002 & 0.0007 \\ \hline 
 \hline 
\end{tabular}}

\label{table:est_price_mis}
\vspace{0.3cm} \\
		\justifying \footnotesize 
		\noindent Note: In these specifications, we misspecify the selection model, assuming that the error term $\varepsilon_{i}$ is drawn from $Logistic (0,1)$.
		The $\text{IBias}^2$ of a function $h$ is calculated as follows. Let $\hat{h}_r$ be the estimate of $h$ from the $r$-th simulated dataset, and $\bar{h}(p) = \frac{1}{R}\sum_{r=1}^{R} \hat{h}_r(p)$ be the point-wise average over $R$ simulations. The integrated squared bias is calculated by numerically integrating the point-wise squared bias $(\bar{h}(p)-h(p))^2$ over the distribution of $p$. The integrated MSE is computed in a similar way. The values reported in each row correspond to the price distributions conditional on a given value of $x_{i2}$. The results shown in this table are based on a 500 Monte Carlo replications with a sample size of 2,000. Corresponding results for a sample size of 5,000 are available upon request. 
\end{table}

\end{document}